\documentclass[12pt]{article}

\usepackage{graphicx}
\graphicspath{{figs/}{}}

\usepackage{float}
\usepackage[makeroom]{cancel}
\usepackage{color}
\usepackage[normalem]{ulem}
\usepackage{amsmath}
\usepackage{amssymb}
\usepackage{xspace}
\usepackage{pifont}
\usepackage{subcaption}

\usepackage[colorlinks=true,citecolor=blue,linkcolor=blue,urlcolor=blue]{hyperref}
\hypersetup{linktocpage}
\usepackage{tcolorbox}

\usepackage[numbers,sort&compress]{natbib}
\usepackage{environ}
\usepackage{booktabs}
\usepackage{array}
\usepackage{multirow}
\usepackage[export]{adjustbox}
\usepackage{tcolorbox}
\usepackage[dvipsnames]{xcolor}

\usepackage{framed}
\usepackage{physics}
\usepackage{tensor}
\usepackage{array}
\usepackage{xcolor}
\usepackage{float}    % for [H]

\NewEnviron{eqs}{%
\begin{equation}\begin{split}
    \BODY
\end{split}\end{equation}
}
\definecolor{refblue}{RGB}{0,80,180}
\definecolor{refred}{RGB}{140,0,0}
\definecolor{lochiralblue}{RGB}{0,70,140}
\newcommand{\Yes}{\textcolor{ForestGreen}{\textbf{\checkmark}}}
\newcommand{\No}{\textcolor{red}{\textbf{\ding{55}}}}

\newlength{\fighskip} \fighskip=2pt
\newlength{\figvskip} \figvskip=3pt

\newcommand*{\figbox}[2]{{
  \def\figscale{#1}
  \def\arraystretch{0.8}
  \arraycolsep=0pt
  \begin{array}{c}
    \vbox{\vskip\figscale\figvskip
      \hbox{\hskip\figscale\fighskip
        \includegraphics[scale=\figscale]{#2}}}
  \end{array}}}

\usepackage{dcolumn}
\usepackage{bm}
\usepackage{verbatim}
\usepackage{amscd}
\usepackage{amsfonts}
\usepackage{setspace}
\usepackage{amsthm}
\usepackage{enumerate}
\usepackage{mathtools}

\theoremstyle{plain}
\newtheorem{principle}{Principle}
\theoremstyle{plain}

\theoremstyle{plain}

\theoremstyle{plain}

\theoremstyle{plain}
\newtheorem{theorem}{Theorem}
\theoremstyle{plain}
\newtheorem{lemma}{Lemma}
\theoremstyle{plain}
 
\theoremstyle{plain}

\theoremstyle{remark}

\theoremstyle{conjecture}
\newtheorem{conjecture}{Conjecture}
\theoremstyle{observation}
\newtheorem{observation}{Observation}
\theoremstyle{definition}

\theoremstyle{corollary}
\newtheorem{corollary}{Corollary}
\theoremstyle{definition}

\theoremstyle{definition}
\newtheorem{result}{Result}
\theoremstyle{result}
\newtheorem{definition}{Definition}
\theoremstyle{assumption}

\theoremstyle{definition}

\theoremstyle{problem}

\theoremstyle{fact}
\newtheorem{fact}{Fact}

\newcommand{\Z}{\mathbb{Z}}

\newcommand{\Edge}{{\mathcal{E}}}

\newcommand{\DL}[1]{ { \color{cyan} \footnotesize (\textsf{DL}) \textsf{\textsl{#1}} } }

\newcommand{\TE}[1]{ { \color{teal} \footnotesize (\textsf{TE}) \textsf{\textsl{#1}} } }

\definecolor{DeepForestGreen}{RGB}{20,110,20}

\usepackage{authblk}

\begin{document}

\title{
\bf %\boldmath 
Many-body chirality of topological \\ stabilizer states
}

\author[1,2,3]{{Tyler D.~Ellison}}
\author[3,4]{Dongjin Lee}
\author[5]{Zhi Li}
\author[6,3]{Amin Moharramipour}
\author[3,4]{\protect\\Yasamin Panahi}
\author[3]{Beni Yoshida}
% \author[3]{Beni Yoshida%
% \thanks{\href{mailto:byoshida@perimeterinstitute.ca}{byoshida@perimeterinstitute.ca}.}}

\affil[1]{Department of Physics and Astronomy, Purdue University, West Lafayette, IN, 47907, USA}
\affil[2]{Purdue Quantum Science and Engineering Institute, Purdue University, West Lafayette, IN, 47907, USA}
\affil[3]{Perimeter Institute for Theoretical Physics, Waterloo, Ontario N2L 2Y5, Canada}
\affil[4]{Department of Physics and Astronomy, University of Waterloo, Waterloo, Ontario N2L 3G1, Canada}
\affil[5]{IBM Quantum, T.J. Watson Research Center, Yorktown Heights, NY 10598, USA}
\affil[6]{Department of Physics, University of Toronto, Toronto, Ontario M5S 1A7, Canada}
% \affil[*]{{The authors are listed alphabetically.}}

\date{}

\maketitle

\begingroup
% \renewcommand\thefootnote{}
% \footnotetext{\normalfont\footnotesize *Authors are listed alphabetically.}
\renewcommand{\thefootnote}{\fnsymbol{footnote}}
\footnotetext[1]{Authors are listed alphabetically.}
\endgroup

% ---------- Abstract ----------
\begin{abstract}

A defining feature of chirality is the distinction between a system and its mirror image. 
Despite extensive experimental observations of chiral phases and theoretical advances, a quantum-information-theoretic characterization of chirality based solely on the entanglement structure of many-body quantum states remains elusive. Here, we introduce the notion of many-body chirality by formulating it as an obstruction to transforming a quantum state into its complex conjugate through finite-depth local operations. We rigorously establish many-body chirality for stabilizer realizations of $\mathbb{Z}_d^{(k)}$ anyon theories, proving that complex conjugation can be implemented by local quantum channels if and only if the underlying anyon data are mirror invariant. This reveals forms of chirality that evade conventional diagnostics, including examples with vanishing modular commutator, vanishing chiral central charge, and commuting-projector realizations. We further show that this obstruction is intrinsically four-partite, while invisible to tripartite entanglement structure. Finally, we prove that $\mathbb{Z}_d^{(k)}$ states with $d>2$ possess intrinsic many-body imaginarity: their complex phase structure cannot be removed by finite-depth local unitaries. Remarkably, this includes states that are not many-body chiral.

\end{abstract}
\newpage
\tableofcontents

%========================
\section{Introduction}\label{sec:introduction}
%========================

One of the central goals of quantum many-body physics is to understand how complex collective phenomena emerge from simple microscopic degrees of freedom. 
Among these, chirality describes an intrinsic absence of symmetry between the system and its mirror image. 
Chirality plays a fundamental role across physics, from chiral fermions and parity violation in high-energy physics to chiral edge transport in quantum Hall systems.
It is also ubiquitous in many-body systems, as many experimentally realized topological phases, including candidate platforms for topological quantum computation, exhibit intrinsically chiral behavior. Despite this, the characterization of chirality in strongly interacting quantum systems continues to pose a fundamental challenge.

Conventionally, chirality has been diagnosed through dynamical responses to external perturbations, for example via chiral edge transport in quantum Hall systems. 
In some cases, it also manifests as an intrinsic obstruction to constructing local commuting projector Hamiltonians with gapped boundaries~\cite{Kitaev_2006}. 
More generally, it is widely believed that, in topological phases, chirality leaves a robust imprint on the statistics of quasiparticle excitations.

Despite these perspectives, it remains unclear whether chirality can be characterized directly from the entanglement structure of a single ground-state wavefunction. 
Indeed, most existing quantum information-theoretic frameworks for characterizing many-body entanglement fail to capture chiral phases.
This raises a fundamental question: can chirality be detected purely from the entanglement structure of a quantum state?

A key quantitative characterization of chirality in two-dimensional topological phases is given by the \emph{chiral central charge} $c_-$, which governs universal edge properties such as thermal Hall conductance. 
The chiral central charge can be extracted directly from the ground state through a real-space formulation of the TKNN invariant for free-fermion systems~\cite{bellissard1994noncommutative,Kitaev_2006}. 
This approach was later generalized to interacting systems, leading to the notion of the modular commutator~\cite{PhysRevLett.128.176402, PhysRevB.106.075147}. For a many-body density matrix $\rho_{ABC}$, the modular commutator is defined as
\begin{align}
J(A,B,C) \equiv i \Tr(\rho_{ABC}[K_{AB}(\rho),K_{BC}(\rho)]) = \frac\pi3 c_{-},
\end{align}
where $K_{AB}(\rho) = -\log \rho_{AB}$ and $c_-$ is the chiral central charge.
A higher central charge generalization~\cite{PhysRevLett.132.016602} and R\'{e}nyi generalizations have also been proposed~\cite{sheffer2025probingchiraltopologicalstates,gass2025renyilikeentanglementprobechiral}.

Although the modular commutator correctly reproduces the chiral central charge for certain gapped ground states, it is not a complete diagnostic of chirality. In particular, it can be reproduced by spurious short-range entangled states~\cite{Gass_2024}. Moreover, there exist many-body systems that exhibit some signatures of chirality while having a vanishing modular commutator, including 3D Walker–Wang models and certain 2D stabilizer mixed states~\cite{walker201131tqftstopologicalinsulators, lee2025chiralcolorcode}.\footnote{
For stabilizer mixed states, the modular commutator necessarily vanishes since the reduced density matrices $\rho_{AB}$ and $\rho_{BC}$ are sums of commuting stabilizers, and thus $[\rho_{AB},\rho_{BC}]=0$.
}
Importantly, these systems are not merely artificial constructions, but arise naturally in settings relevant to quantum information and topological phases, including models for fault-tolerant quantum computation and higher-dimensional topological systems. 
These examples pose a fundamental tension in how chirality should be characterized. 
Then, what is the appropriate notion of chirality for many-body quantum states?

\subsubsection*{LO-chirality}

In many-body physics, pure-state phases of matter are characterized by finite-depth local unitary (LU) circuits, where two pure states belong to the same phase if they are related by such a circuit.
For mixed states, this notion of phase equivalence is naturally generalized by replacing LU circuits with local operations (LO), also referred to as local quantum channels~\cite{Coser2019classificationof,PhysRevX.14.031044}.

In this paper, we introduce a notion of chirality that can be diagnosed directly from a single many-body wavefunction. Here, we consider chirality as an 
% Our definition is motivated by the observation that time-reversal symmetry acts as complex conjugation, so that chirality can be understood as an
obstruction to relating a state $\rho$ to its complex conjugate $\rho^*$ by local transformations.
\begin{definition}[LU- and LO-chirality]\label{def:lu-lo-chirality}
Let $\rho$ be a many-body quantum state defined on a lattice.
We say that $\rho$ is \emph{LU-chiral} if
\begin{eqnarray}
\rho^\ast \neq U \rho U^\dagger
\end{eqnarray}
for any finite-depth unitary circuit $U$.
We say that $\rho$ is \emph{LO-chiral} if
\begin{eqnarray}
\rho^\ast \neq \Tr_A \left( U (\rho \otimes |0^n \rangle_A \langle 0^n|) U^\dagger \right)
\end{eqnarray}
for any finite-depth unitary $U$ and constant-density ancilla state $|0^n\rangle_A$.
\end{definition}
In this work, we do not impose any symmetry constraints on the finite-depth unitary $U$.
While LU- and LO-chirality coincide for the models studied in this work, they are inequivalent notions in general and can lead to different classifications of many-body states.
We comment briefly on this distinction in Sec.~\ref{subsec:LOvsLu}.

\subsubsection*{Mirror invariance of abelian anyons}

The goal of this work is to study LO-chirality in two-dimensional topological phases. 
To isolate the essential physics, we restrict attention to phases with abelian anyons.
 
In a two-dimensional topological phase, the anyon content can be described by a finite set, together with their braiding statistics and topological spins
\begin{align}
\mathcal{A} = \{ a_1, \cdots, a_n \},\qquad B(a_i,a_j), \ \theta(a_i) \in U(1).
\end{align}
Here, $B(a_i,a_j)$ encodes the phase acquired when $a_i$ braids around $a_j$, while $\theta(a_i)$ characterizes the self-statistics of $a_i$. These data fully determine an abelian anyon theory \cite{Etingof:2015TensorCategories,Davydov2010TheWG}.

Not all anyon theories admit the same type of microscopic realization. A special class of abelian anyon theories, namely those admitting a \emph{Lagrangian subgroup}, can be realized as ground states of commuting-projector Hamiltonians, such as string-net models~\cite{LevinWen2005,LinLevin2014}. 
By contrast, abelian anyon theories without a Lagrangian subgroup are widely believed not to admit commuting-projector realizations~\cite{Levin_2013}. 
Instead, they are expected to be realized by non-commuting Hamiltonians with protected edge modes~\cite{LinLevin2014,Wen1990chiralLuttinger,Wen1995edgeexcitations,WenZee1992,yang2026classification}.
Interestingly, such theories can nevertheless be realized as mixed states~\cite{Sohal_2025, Ellison_2025}, supported on extensively degenerate ground-state subspaces of commuting-projector models. 

Two anyon theories $\mathcal{A}$ and $\mathcal{A}'$ are said to be \emph{isomorphic}, $\mathcal{A} \simeq \mathcal{A}'$, if they can be related by a relabeling of anyons that preserves their braiding and topological spin data. It is generally expected that isomorphic anyon theories correspond to the same phase, although a complete proof is known only in certain cases, such as when a Lagrangian subgroup is present~\cite{KapustinSaulina2010,LinLevin2014}.\footnote{For instance, to the best of our knowledge, it remains an open question whether mixed-state realizations and non-commuting Hamiltonian realizations of the same anyon theory necessarily belong to the same phase.}

Given an anyon theory $\mathcal{A}$, we define its complex conjugate theory $\mathcal{A}^*$ by conjugating all statistical phases:
\begin{align}
\mathcal{A}^* = \{ \overline{a_1}, \cdots, \overline{a_n} \}, \qquad
B(\overline{a_i}, \overline{a_j}) = B(a_i, a_j)^*, \quad
\theta(\overline{a_i}) = \theta(a_i)^*.
\end{align}
This transformation corresponds physically to reversing the orientation of spacetime. 
This motivates the following notion.

\begin{definition}
An abelian anyon theory $\mathcal{A}$ is said to be \emph{mirror invariant} if
\begin{align}
\mathcal{A} \simeq \mathcal{A}^*,
\end{align}
that is, if $\mathcal{A}$ is isomorphic to its complex conjugate.
\end{definition}

From a physical perspective, one expects that LO-chirality will be reflected in the failure of mirror invariance. This leads to the following guiding principle:

\begin{principle}[Chirality and mirror symmetry]
A physical realization of a two-dimensional abelian anyon theory $\mathcal{A}$ is expected to be LO-chiral if and only if $\mathcal{A}$ is not mirror invariant.
\end{principle}

At first sight, this principle may appear to be well established. 
In particular, one might expect it to follow directly from the widely held belief that anyon content fully characterizes a topological phase. Indeed, many previous works implicitly adopt this viewpoint.
However, a precise formulation and proof of this statement turn out to be highly nontrivial, involving both conceptual and technical subtleties.

Early rigorous results focused on detecting long-range entanglement rather than distinguishing different anyon theories.
It was first proved that the toric code ground state on a torus cannot be prepared by a finite-depth quantum circuit, using Lieb–Robinson bounds~\cite{Bravyi_2006}.
The argument, however, relies crucially on ground-state degeneracy.
Later an analytical proof of long-range entanglement on a sphere was provided by introducing the notion of twist products, which beyond establishing long-range entanglement, also computes the $S$ matrix of the underlying anyon data~\cite{Haah_2016}.
A related argument was also given by Bravyi, although it remains unpublished.

In contrast, much less is known about rigorously separating phases with distinct anyon content~\cite{kim2024classifying2dtopologicalphases}.
The situation is even more subtle for mixed states, where the notion of topological order itself requires refinement.
It has been conjectured that mixed-state topological order is classified by premodular anyon theories~\cite{Sohal_2025,Ellison_2025}, and this conjecture has recently been established at a mathematical level in~\cite{ogata2025mixedstatetopologicalorder} using an operator-algebraic framework under the assumption of approximate Haag duality.

In this work, we establish this principle rigorously for a family of stabilizer mixed-state realizations of abelian anyon theories, namely $\mathbb{Z}_d^{(k)}$ theories. 
These are the subtheories of $\Z_d$ toric codes generated by the anyon $em^k$, and admit particularly simple commuting-projector realizations with vanishing modular commutator.
While this family does not exhaust all abelian anyon theories, it is sufficiently general in the sense that more general theories can be obtained from them, after stacking, via anyon condensation; see Appendix~\ref{app:stacking-Zdk}.

\subsubsection*{Chiral central charge}

It is also instructive to compare the notion of mirror-invariance with conventional diagnostics of chirality of an anyon theory.
The chiral central charge $c_{-}$ can be extracted from anyon data via the Gauss sum
\begin{align}
e^{\frac{2\pi i}{8}c_{-}} = \frac{1}{\sqrt{|\mathcal{A}|}}\sum_{a \in \mathcal{A}}\theta(a).
\end{align}
From this expression, one can deduce that anyon theories with $c_{-}\not=0,4 \ (\mathrm{mod}\ 8)$ are necessarily not mirror invariant, and hence their physical realizations must be LO-chiral.
By contrast, there exist theories with $c_{-}=4 \ (\mathrm{mod}\ 8)$ that are LO-non-chiral in our framework.
An important example is the three-fermion (3F) theory, which admits a qubit stabilizer mixed-state realization.

More importantly, we find that chiral central charge is not a faithful diagnostic of LO-chirality.
In particular, there exist families of anyon theories, such as $\mathbb{Z}_{p^2}^{(1)}$ with $p \equiv 3 \ (\mathrm{mod}\ 4)$ that have vanishing chiral central charge (modulo $8$) but are nevertheless LO-chiral due to the lack of mirror invariance.
Furthermore, these models admit pure-state realizations with commuting stabilizer Hamiltonians.
They therefore provide examples of topological phases with gapped boundaries that are nonetheless LO-chiral.\footnote{ 
One may also consider higher Gauss sums
$\tau_n \equiv \frac{1}{\sqrt{|\mathcal{A}|}}\sum_{a \in \mathcal{A}}\theta(a)^n$,
which encode additional topological information of anyon statistics~\cite{Kaidi_2022, geiko2022chernsimons}.
However, the resulting chirality criteria are typically more intricate and involve nontrivial number-theoretic structures.
}

\subsection{Summary of results}

\subsubsection*{LO-chirality}
We begin with LO-chirality for $\mathbb{Z}_d^{(k)}$ anyon theories. Our first main result is the following.
\begin{result}[LO-chirality]
A stabilizer mixed state supporting $\mathbb{Z}_d^{(k)}$ abelian anyon theory is LO-chiral if and only if its anyon content is not mirror invariant. Moreover, a pure state supporting $\mathbb{Z}_{p^{2k}}^{(1)}$ anyon theory is LO-chiral if and only if its anyon content is not mirror invariant.
\end{result}

For mixed-state realizations, this result shows that a broad class of states in this family possesses intrinsic chirality, even though conventional entanglement-based diagnostics such as the modular commutator vanishes identically. A representative class is given by $\mathbb{Z}_p^{(1)}$ with odd prime $p$ satisfying $p \equiv 3 \ (\mathrm{mod}\ 4)$, whose anyon theories have nonzero chiral central charge. 
This phenomenon also extends to pure-state realizations. In particular, the $\mathbb{Z}_{p^2}^{(1)}$ theories with odd prime $p$ satisfying $p \equiv 3 \ (\mathrm{mod}\ 4)$ provide examples of LO-chiral pure states with vanishing chiral central charge. These models are especially striking because they admit commuting-projector realizations and gapped boundaries, yet remain LO-chiral.

The key idea behind the proof is as follows.
Suppose that a finite-depth local unitary (or more generally, a local quantum channel) maps $\rho^*$ to $\rho$ despite $\mathcal{A}\not\simeq \mathcal{A}^*$. 
This would transform a realization of $\mathcal{A}^*$ into $\rho$, so that the state $\rho$ support excitations\footnote{Here, and throughout, we use the term ``excitation'' loosely. This can be made more precise for a mixed state by considering its canonical purification, as in Sec.~\ref{sec:2.2}.} corresponding to both $\mathcal{A}$ and $\mathcal{A}^*$, leading to a contradiction.

The nontrivial part is to make this argument rigorous. While it may appear obvious that $\rho$ supports only the anyons in $\mathcal{A}$, one must exclude possibilities, such as finely tuned string operators generating excitations outside $\mathcal{A}$ or coherent superpositions of different anyon charges.
To address this, we establish conservation and factorization of anyon charges, ensuring that excitations created by local operators carry well-defined, localized charges in $\mathcal{A}$.
We then show that anyon statistics depends only on the anyon type and is invariant under local transformations. We prove this by showing that anyon excitations can be locally concentrated into point-like objects through a recovery argument analogous to the proof of Uhlmann’s theorem. This enables a direct comparison of the statistical phases associated with $\mathcal{A}$ and $\mathcal{A}^*$.

As a byproduct of our proof, we show that abelian anyon theories $\mathcal{A}$ and $\mathcal{A}'$ that are not isomorphic cannot be related by two-way LO transformations, and therefore correspond to distinct quantum phases.
While this is widely expected, our proof clarifies which microscopic properties, such as charge conservation, factorization, and locality, are responsible for this distinction.

Finally, our result has a corollary concerning the classification of mixed-state phases.
Noting that the canonical purification $|\sqrt{\rho}\rangle$ encodes all information about the mixed state $\rho$, one might expect that mixed-state phases can be classified by studying the phases of their canonical purifications. However, we show that this expectation is not correct: two mixed states may belong to distinct phases even when their canonical purifications are related by local unitaries.

Let $\rho$ be an LO-chiral $\mathbb{Z}_{d}^{(1)}$ mixed state and $\rho^*$ its complex conjugate. Their canonical purifications are given by
\begin{align}
|\sqrt{\rho}\rangle \propto (\rho \otimes I )|\mathrm{EPR}\rangle,\qquad
|\sqrt{\rho^*}\rangle \propto (\rho^* \otimes I )|\mathrm{EPR}\rangle,
\end{align}
where we use that $\rho$ has a flat spectrum, so that $\sqrt{\rho}\propto \rho$.
We note that $|\sqrt{\rho}\rangle$ and $|\sqrt{\rho^*}\rangle$ belong to the same pure-state phase, since
\begin{align}
\mathrm{SWAP} |\sqrt{\rho}\rangle \propto (I \otimes \rho )|\mathrm{EPR}\rangle = (\rho^{*} \otimes I )|\mathrm{EPR}\rangle \propto |\sqrt{\rho^*}\rangle \ .
\end{align}
Here, $\mathrm{SWAP}$ can be realized by a finite-depth unitary circuit. We also used the identity $O\otimes I|\mathrm{EPR}\rangle = I \otimes O^T |\mathrm{EPR}\rangle$ together with the Hermiticity of $\rho$.

We conclude that classification based on canonical purification fails to distinguish certain mixed-state phases, in particular $\rho$ and $\rho^*$.

\begin{corollary}[Purification is not sufficient]
For an LO-chiral $\mathbb{Z}_{d}^{(1)}$ mixed state $\rho$, we have
\begin{align}
\rho \not\simeq \rho^*, \qquad
|\sqrt{\rho}\rangle \simeq |\sqrt{\rho^*}\rangle.
\end{align}
In other words, classification of pure-state phases via canonical purification is insufficient for classifying mixed-state topological phases.
\end{corollary}

\subsubsection*{$n$-partite chirality}

Next, we address the question concerning the entanglement structure required to realize many-body chirality. 
Rather than imposing geometric locality constraints, we restrict the allowed transformations to have a multipartite structure, and ask whether the state can be mapped to its complex conjugate. 
This leads to the notion of $n$-partite chirality, introduced in Ref.~\cite{vardhan2025chiralitymagicquantumcorrelations}. 

\begin{definition}[$n$-partite chirality]\label{def:n-partite-chirality}
Let $\rho_{A_1A_2\cdots A_n}$ be a many-body quantum state partitioned into subsystems $A_1,A_2,\dots,A_n$.
We say that $\rho_{A_1A_2\cdots A_n}$ is \emph{$n$-partite LO-chiral} if
\begin{eqnarray}
\rho_{A_1A_2\cdots A_n}^* \neq \Tr_{R} \big( U (\rho_{A_1A_2\cdots A_n} \otimes |0^n\rangle_{R} \langle 0^n|) U^\dagger \big)
\end{eqnarray}
for any unitary of the form $U = \bigotimes_{j=1}^n U_{A_j R_j}$ and any constant-density ancilla state $|0^n\rangle_R$. 
\end{definition}
This definition characterizes the minimal multipartite structure needed to detect chirality.
A previous work~\cite{vardhan2025chiralitymagicquantumcorrelations} showed that qubit stabilizer states are not $n$-partite chiral for any $n$.
This no-go result can be circumvented in qudit stabilizer systems, as we demonstrate.

\begin{result}[$n$-partite chirality]
A stabilizer mixed state supporting $\mathbb{Z}_d^{(k)}$ abelian anyon theory is four-partite chiral if and only if the $\mathbb{Z}_d^{(k)}$ anyon theory is not mirror-invariant. At the same time, these states are not three-partite chiral: their tripartite entanglement reduces to combinations of EPR- and GHZ-type correlations.
\end{result}

The proof of four-partite chirality follows a similar strategy to that of LO-chirality, but involves additional subtleties.
The key challenge is that the subsystems are not geometrically separated, which obscures localization of anyonic excitations.
In particular, this makes it nontrivial to assign charges near tri-junctions where three subsystems meet.
To address this, we introduce the notion of \emph{edge charges}, which assign anyonic charge to boundaries between subsystems. 
This allows us to establish conservation and factorization of charges even in the absence of spatial locality.

Our result has a key implication. 
We have already argued that the modular commutator, which is tri-partite entanglement measure, does not serve as faithful measure of chirality. 
This naturally raises the question of whether a genuine entanglement-based measure of chirality exists, and if so, what minimal degree of multipartiteness is required.
Our result suggests that any such measure for mixed states must involve at least four parties, as no tripartite entanglement measure can faithfully capture an obstruction to complex conjugation.

Finally, we highlight an important subtlety. 
Given that the LO-chirality proof extends straightforwardly to pure-state realizations, one might expect that the four-partite chirality result also extends. 
However, this is not the case. 
We point out an obstruction in the $\mathbb{Z}_{p^{2k}}^{(1)}$ models that prevents a straightforward extension of our proof to pure states.

\subsubsection*{Many-body imaginarity}

The notion of chirality naturally raises a fundamental question concerning the role of complex phases in quantum many-body systems. 
While physical observables are Hermitian and therefore have real eigenvalues, quantum states themselves are described by complex amplitudes. 
This raises the question of whether complex phases are an essential feature of quantum mechanics, or whether they can be removed by local basis transformations. See Refs.~\cite{Renou_2021, hoffreumon2025quantumtheorydoesneed,ying2025quantumtheoryneedscomplex} for recent debates on this issue. 

From a many-body perspective, this motivates the following notion. If a quantum state cannot be transformed into a representation with real coefficients by any local basis transformation, we say that it possesses an intrinsic complex phase structure. 
We formalize this notion as follows.

\begin{definition}[LU-imaginarity]
Let $\rho$ be a many-body quantum state defined on a lattice.
We say that $\rho$ is \emph{LU-imaginary} if
\begin{eqnarray}
\sigma \neq \sigma^* \qquad \sigma = U \rho U^\dagger,
\end{eqnarray}
for any finite-depth unitary circuit $U$.\footnote{LO-imaginarity would be trivial as a depolarizing channel transforms any $\rho$ into a maximally mixed state.}
\end{definition}

Our notion of LU-imaginarity is closely related to recent developments in the resource theory of imaginarity~\cite{Wu_2021,Wu_2024}. 
Here, we extend this perspective to many-body systems by adopting LU transformations as free operations.

By construction, LU-imaginarity is a weaker condition than LU-chirality: any LU-chiral state is necessarily LU-imaginary. A natural question is whether there exist many-body quantum states that are LU-imaginary but not LU-chiral.

\begin{result}[LU-imaginarity]
All stabilizer mixed states supporting $\mathbb{Z}_d^{(1)}$ anyon theory for $d>2$ are LU-imaginary. Moreover, all pure-states supporting $\mathbb{Z}_{p^{2k}}^{(1)}$ anyon theory are LU-imaginary. 
\end{result}

This implies that any lattice realization of the state necessarily involves intrinsically complex amplitudes (in a local basis) that cannot be removed by local transformations.
Notably, this includes examples that are not LO-chiral. The simplest instance is the $\mathbb{Z}_5^{(1)}$ model, which is LU-imaginary but LO-non-chiral. 
% This is particularly striking as, at the level of the effective Lagrangian in topological quantum field theory, the theory appears to be real~\cite{geiko2022chernsimons}.\DL{I am actually not sure about this statement. In~\cite{geiko2022chernsimons}, the necessary condition that MTC to be time reversal is that its all of higher Gauss sum is 1. (page 28) But $\mathbb{Z}_5^{(1)}$ theory has non-trivial higher Gauss sum.} \TE{Right, it is consistent with the field-theory predictions.}

Our result can be also viewed as a many-body generalization of the necessity of complex phases in topological wavefunctions. In particular, Hastings showed that the double semion ground state exhibits an intrinsic sign structure, in the sense that no local basis transformation can render all amplitudes nonnegative~\cite{Hastings_2015}.
In contrast, our result demonstrates that $\mathbb{Z}_d^{(1)}$ anyon theories ($d>2$) exhibit an intrinsic \emph{imaginary} structure where complex phases cannot be eliminated by any local transformation.

\subsection*{Plan of the paper}

The paper is organized as follows. In Sec.~\ref{sec2}, we introduce stabilizer mixed (and in some cases, pure) states that support $\mathbb{Z}_d^{(k)}$ anyons and present a concrete criterion for mirror invariance.
In Sec.~\ref{sec:3}, we prove the LO-chirality theorem, showing that $\rho \mapsto \rho^\ast$ is possible if and only if the underlying anyon theory is mirror invariant.
In Sec.~\ref{sec4:four=partite-proof}, we consider $n$-partite chirality and show that our stabilizer mixed states are four-partite chiral but nonchiral under three-partition. In Sec.~\ref{sec5}, we prove intrinsic imaginarity for stabilizer mixed states. In Sec.~\ref{sec6:discussion}, we conclude with discussions.

%\newpage

\section{Honeycomb model and $\mathbb{Z}_d^{(k)}$ anyons}
\label{sec2}

In this section, we introduce a concrete family of two-dimensional stabilizer mixed-state models realizing $\mathbb{Z}_d^{(k)}$ anyons, and review the associated anyon theory, including its braiding statistics, topological spins, and the mirror conjugation. 
We focus on this family as it admits an explicit lattice construction while remaining sufficiently broad.
In fact, arbitrary modular abelian anyon theories can be realized by stacking multiple copies of $\mathbb{Z}_d^{(k)}$ and condensing anyons~\cite{Davydov2010TheWG, Davydov2013}.

We first introduce the honeycomb stabilizer model in Sec.~\ref{sec:2.1}. We then discuss strong and weak symmetries of the model in Sec.~\ref{sec:2.2}. In Sec.~\ref{sec:2.3}, we explicitly show that its strong symmetries realize $\mathbb{Z}_d^{(k)}$ anyons by computing anyon statistics using string operators. The necessary and sufficient conditions for $\mathbb{Z}_d^{(k)}$ anyon theory to be invariant under complex conjugation are derived in Sec.~\ref{sec:2.4}.

\subsection{Honeycomb model}
\label{sec:2.1}

We start by introducing an explicit mixed state that realizes the $\mathbb{Z}_d^{(k)}$ anyon theory on a honeycomb lattice, inspired by the construction in Ref.~\cite{lee2025chiralcolorcode}. In general, the mixed state can be defined on any tri-valent, three-colorable lattice, but for simplicity, we focus on the honeycomb lattice.
We use labels $A,B,C$ for the colors of the plaquettes of the honeycomb lattice. We note that the vertices of the honeycomb model are bipartite, meaning that one can assign one of two labels (black or white) to each vertex. 
At each vertex, we assign a single qudit of dimension $d$. 
The model is specified in terms of generalized Pauli operators $X$ and $Z$ satisfying
\[
X^d = Z^d = I, \qquad ZX = \omega XZ ,
\]
where $\omega = e^{2\pi i/d}$.
Furthermore, we work in the computational basis where $X$ matrices are real (shift operator) and $Z$ matrices are diagonal with pure phases (clock operator).

\begin{definition}[Honeycomb model]\label{def:honeycombmodel}
The honeycomb model for the realization of $\mathbb{Z}_d^{(k)}$ anyons is defined by the stabilizer mixed state $\rho_{\cal S}$, associated with the stabilizer group ${\cal S}$ generated by local plaquette stabilizers. The plaquette stabilizer generators are given by 
\begin{align}
S_{f_A} = \figbox{3.0}{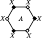}\ , \qquad S_{f_B} = \figbox{3.0}{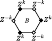}\ , \qquad S_{f_C} = \figbox{3.0}{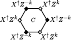} \label{eq:stabilizer-generators-mixed_state} \ .
\end{align}
Here $f_A$, $f_B$, and $f_C$ denote faces of color $A$, $B$, and $C$, respectively, and $d$ and $k$ are taken to be coprime.\footnote{The construction simply generalizes to the case when $d$ and $k$ are not coprime by dropping the $k$ exponent of the $Z$ operators in $S_{f_B}$.}
The associated stabilizer mixed state is the normalized projector onto the common $+1$ eigenspace of the plaquette stabilizers,
\begin{align}
\rho_{{\cal S}} \equiv \frac{1}{d^n}\sum_{O\in{\cal S}} O ,
\label{eq:mixed_state}
\end{align}
where $d$ is the local dimension of qudits and $n$ is the total number of qudits.
%which form basis for strong symmetries. 
Logical operators are generated from edge operators
\begin{align}
g_{e_{BC}} = \figbox{3.0}{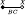} \ , \qquad
g_{e_{AC}} = \figbox{3.0}{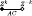} \ , \qquad
g_{e_{AB}} = \figbox{3.0}{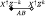} \ \label{eq:edge_operators} \ .
\end{align}
%which form basis for weak symmetries. 
\label{def:honeycomb-mixed}
\end{definition}
Here, the edge label indicates the colors of the two faces adjacent to the edge; for instance, $e_{BC}$ denotes an edge shared by a $B$ face and a $C$ face. The operator assignment is read from the endpoint colors (filled and open dots denote black and white vertices, respectively). Equivalently, for an edge $e=\langle v_\bullet,v_\circ\rangle$ with black endpoint $v_\bullet$ and white endpoint $v_\circ$, we have
\begin{align}
\label{eq:logicalop-endpointlabel}
g_{e_{BC}} = X_{v_\bullet}X_{v_\circ}, \ \qquad
g_{e_{AC}} = Z^k_{v_\bullet}Z^{-k}_{v_\circ}, \ \qquad
g_{e_{AB}} = \bigl(X^\dagger Z^{-k}\bigr)_{v_\bullet}
              \bigl(X^\dagger Z^{k}\bigr)_{v_\circ}.
\end{align}
With the same convention, the plaquette stabilizers in Eq.~\eqref{eq:stabilizer-generators-mixed_state} are generated by products of edge operators around elementary plaquette loops.
This construction naturally arise, for instance, as boundary theories of three-dimensional systems~\cite{lee2025chiralcolorcode} or as mixed states generated by decoherence processes~\cite{Ellison_2023}. 

\subsubsection*{Pure-state realizations}

Some $\mathbb{Z}_d^{(k)}$ anyon theories can be captured by pure states. Here, we focus on a particular subclass of such anyon theories, namely $\mathbb{Z}_{p^{2m}}^{(1)}$. These can be realized as a ground state of a commuting-projector stabilizer Hamiltonian. Notably, the $\mathbb{Z}_{p^{2m}}^{(1)}$ anyon theories have a vanishing chiral central charge, so they do not exhibit the conventional signature of chirality.\footnote{Moreover, they admit a gapped boundary, and hence, can be realized by stabilizer models~\cite{Ellison2022TQDs}.}
By a \emph{pure-state} realization, we mean that the stabilizer Hamiltonian has a unique ground state when defined on a sphere.

\begin{definition}[$\mathbb{Z}_{p^{2m}}^{(1)}$ pure state]\label{def:Zp1pure}
The local stabilizer generators of the $\mathbb{Z}_{p^{2m}}^{(1)}$ pure state are given by 
\begin{eqnarray}
\label{eq:stabilizers-purestate}
S_{f_A} = \figbox{3.0}{fig_SA.pdf}\ , \qquad S_{f_B} = \figbox{3.0}{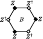}\ , \qquad S_{f_C} = \figbox{3.0}{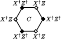} 
\cr\cr
S_{e_{BC}} = \figbox{3.0}{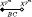} \ , \qquad \quad
S_{e_{AC}} = \figbox{3.0}{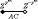} \ , \qquad \;\;
S_{e_{AB}} = \figbox{3.0}{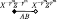} \ \quad
\end{eqnarray}
\end{definition}

This pure-state realization supports the same anyon content as the mixed-state construction.
The key observation is that taking the $p$-th power of the edge operators, $g_e^p$, makes them mutually commuting, allowing them to be included in the stabilizer group. 

Let us verify that this construction indeed defines a pure state on a sphere using a counting argument. For this, we consider an arbitrary tri-valent, three-colorable lattice, since a honeycomb lattice is not consistent with a sphere.
Let $V$, $E$, $F$ be the number of vertices, edges, and faces of the lattice, respectively.
There are $V$ qudits of dimension $p^{2m}$, together with $F-1$ independent plaquette stabilizers of order $p^{2m}$ and $E-F$ independent edge stabilizers of order $p^m$. 
Thus, the number of logical degrees of freedom $N_{\rm logical}$ is 

\begin{eqnarray}
N_{\rm logical} &=& V - (F-1) - \frac{1}{2}(E-F) = \frac{1}{2}(2 - \chi) = 0.
\end{eqnarray}
Here, the contribution of edge stabilizers are reduced by a factor $\frac{1}{2}$, since these stabilizers are of order $p^m$, as opposed to order $p^{2m}$. 
We have also used the tri-valent condition ($3V=2E$) in the second equality, and $\chi=2$ is the Euler characteristic of a sphere. Thus, on a sphere, the stabilizer generators determine a unique pure state. 

In the next section, we show that this pure-state realization of $\mathbb{Z}_{p^{2m}}^{(1)}$ anyons is LO-chiral when $p=3 \ (\mathrm{mod}\ 4)$.
This provides an explicit example of a two-dimensional \emph{stabilizer pure state} exhibiting chirality, despite having chiral central charge $c_{-}=0 \text{ mod } 8$.

\subsection{Strong and weak symmetry}
\label{sec:2.2}

To discuss the anyon theory associated with these stabilizer states, we first recall the notions of strong and weak symmetries in mixed states~\cite{Ellison_2025, Ma_2023,Lee_2025, Ma-Turzillo2025, Sala2024,Lessa_2025,Kuno2024,Orito_2025,Czhang2026}.

\begin{definition}[Strong and weak symmetry]
A unitary operator $O$ is a strong symmetry of $\rho$ if 
\begin{align}
O \rho = \rho = \rho O^{\dagger}. \label{eq:strong}
\end{align}
A unitary operator $g$ is a weak symmetry of $\rho$ if 
\begin{align}
g \rho g^{\dagger} = \rho. \label{eq:weak}
\end{align}
\end{definition}

\begin{fact}\label{fact:strong_commutation}
Weak and strong symmetry operators commute when acting on $\rho$:
\begin{align}
[O,g]\rho = 0 \label{eq:strong_commutation}
\end{align}
where $O$ and $g$ are strong and weak symmetry unitaries, respectively.
\end{fact}

The commutativity holds when acting on $\rho$, even though $O$ and $g$ need not commute as operators on the full Hilbert space. 
Fact~\ref{fact:strong_commutation} follows directly from 
\begin{align}
Og \rho = Og \rho g^{\dagger} g = O \rho g = \rho g = g \rho g^{\dagger} g
= g\rho = g O \rho. 
\end{align}

For a stabilizer mixed state, the notion of strong and weak symmetries naturally emerges from its association with a code~\cite{Sohal_2025,Ellison_2025}.
Let $\rho_{\mathcal{S}}$ be the maximally mixed state stabilized by the Pauli stabilizer group $\mathcal{S}$ (i.e. the maximal entropy state in the $+1$ eigenstate space of stabilizer generators in $\mathcal{S}$).  
\begin{align}
\label{eq:def-stabilizer-mixedstate}
\rho_{\mathcal{S}} \equiv \frac{1}{d^n} \sum_{O \in \mathcal{S}} O
\end{align}
where $d$ is the local dimension of qudits and $n$ is the total number of qudits. 
We can verify the following properties.

\begin{fact}[Stabilizer mixed state]
A stabilizer mixed state $\rho_{\mathcal{S}}$ satisfies:
\begin{enumerate}[i)]
\item Stabilizer operators are strong symmetries of $\rho_{\mathcal{S}}$. 
\item Logical unitary operators of the stabilizer code $\mathcal{S}$ are weak symmetries of $\rho_{\mathcal{S}}$.
\item Any strong symmetry unitary $U$ acts trivially on the code subspace:
\begin{align}
U |\psi\rangle = |\psi\rangle \qquad \text{for all $|\psi\rangle$ in $\rho_{\mathcal{S}}$}.
\end{align}
\end{enumerate}
\end{fact}

% In the stabilizer mixed state $\rho_{\mathcal{S}}$, strong symmetries act trivially on the code subspace, while weak symmetries correspond to nontrivial logical operations. 
In particular, for the honeycomb stabilizer mixed states, the strong symmetries are generated by plaquette stabilizers in Eq.~(\ref{eq:stabilizer-generators-mixed_state}), while the weak symmetries are generated by edge logical operators in Eq.~(\ref{eq:edge_operators}). 
Strong symmetries of the honeycomb mixed state will be used to construct string operators, allowing us to identify the anyon theory associated with the mixed state. 
 
\subsubsection*{Canonical purification}

To discuss the anyon theory, we find it convenient to work with the canonical purification. Thus, we first review the canonical purification of a stabilizer mixed state.

\begin{definition}[Canonical purification]
Given a mixed state $\rho$ on $A$, its canonical purification is defined as
\begin{align}
|\Psi_{\rho}\rangle_{AA'} \equiv \sum_j  \sqrt{\rho}|j\rangle_A \otimes |j\rangle_{A'} = \sum_j  |j\rangle_A \otimes \sqrt{\rho^*}|j\rangle_{A'}.
\end{align}
One can equivalently write
\begin{align}
|\Psi_{\rho}\rangle_{AA'} = \sum_{j}\sqrt{p_j}|\psi_j\rangle \otimes |\psi_j^*\rangle, \qquad \rho = \sum_{j} p_j |\psi_j\rangle \langle \psi_j |.
\end{align}
\end{definition}

It is convenient to graphically denote the canonical purification:
\begin{align}
|\Psi_{\rho}\rangle_{AA'} 
= \ \figbox{3.0}{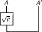} \ 
= \ \figbox{3.0}{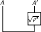} \ .
\end{align}
Here, the complex conjugation arises due to bringing $\sqrt{\rho}$ to the other side, which implements a transpose $T$.
(Since $\rho$ is hermitian, we have $\rho^T = \rho^*$.)

When $\rho_{\mathcal{S}}$ is a stabilizer mixed state, its canonical purification $|\Psi_{\rho_{\mathcal{S}}}\rangle$ is a stabilizer pure state defined on a doubled Hilbert space $A\otimes A'$. 
Namely, its stabilizer group is generated by 
\begin{align}
\mathcal{S}_{\Psi_{\rho_{\mathcal{S}}}} \equiv \big\langle \mathcal{S}\otimes I, \ I \otimes \mathcal{S}^* , \
\mathcal{G} \otimes \mathcal{G}^* \big\rangle
\end{align}
where $\mathcal{G}$ denotes the group of Pauli logical operators of the stabilizer code $\mathcal{S}$.
Elements in $\mathcal{G} \otimes \mathcal{G}^*$ commute with each other since the phases cancel exactly due to the complex conjugation.

\begin{fact}\label{fact:strongweak}
Strong symmetry unitary $U$ of $\rho$ acts trivially on the canonical purification:
\begin{align} \label{eq:fact3_strong}
(U\otimes I) |\Psi_{\rho}\rangle =  |\Psi_{\rho}\rangle \qquad \text{$U$ is a strong symmetry}.
\end{align}
While weak (but not strong) symmetry unitary $g$ changes $|\Psi_{\rho}\rangle$, $(g\otimes g^*)$ acts trivially:
\begin{align} \label{eq:fact3_weak}
(g\otimes I) |\Psi_{\rho}\rangle \not= |\Psi_{\rho}\rangle,~~~~(g\otimes g^*) |\Psi_{\rho}\rangle =  |\Psi_{\rho}\rangle \qquad \text{$g$ is a weak symmetry}.
\end{align}
\end{fact}

Here, it is useful to graphically verify these properties. For example,
\begin{align}
\langle \Psi_{\rho} | (U\otimes I) | \Psi_{\rho}\rangle 
= \figbox{3.0}{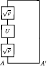} = \Tr(U\rho) = 1.
\end{align}
In fact, Fact \ref{fact:strongweak} holds generally for arbitrary mixed state $\rho$ and is not limited to stabilizer mixed state.
This follows from the equivalence that, for any unitary $g$ and mixed state $\rho$, we have $g\rho g^\dagger=\rho$ if and only if $g\sqrt\rho g^\dagger=\sqrt\rho$, and $g\rho=\rho$ if and only if $g\sqrt{\rho}=\sqrt{\rho}$.

Note that single-layer marginals of $(g\otimes I)|\Psi_{\rho}\rangle$ on $A$ and $A'$ remain the same under the action of a weak symmetry operator $(g\otimes I)$:
\begin{align}
\Tr_{A'}\big[(g\otimes I)|\Psi_{\rho}\rangle \langle \Psi_{\rho}| (g^{\dagger}\otimes I)\big] = \rho,\qquad \Tr_{A}\big[(g\otimes I)|\Psi_{\rho}\rangle \langle \Psi_{\rho}| (g^{\dagger}\otimes I)\big] = \rho^*.
\end{align}
This suggests that weak symmetry operator $g$ changes the entanglement pattern between $A$ and $A'$ while keeping the single-layer reduced states invariant.

\subsubsection*{Canonical purification of the honeycomb model}

The canonical purification $|\Psi_{\rho_{\mathcal{S}}}\rangle$ of the honeycomb stabilizer mixed state is defined on two layers of honeycomb lattices. 
Stabilizer generators include $\mathcal{S}\otimes I$ on the first layer as well as the complex conjugations $I \otimes \mathcal{S}^*$ on the second layer:
\begin{align}
S_{f_A}^* = \figbox{3.0}{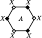}\ , \qquad S_{f_B}^* = \figbox{3.0}{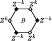}\ , \qquad S_{f_C}^* = \figbox{3.0}{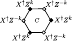} \ 
\end{align}
which realizes the $\mathbb{Z}_{d}^{(-k)}$ model. 
It is illuminating to observe that complex conjugation $\mathcal{S}^*$ can be implemented by exchanging the bipartition labels (black and white sublattices).
Equivalently, this corresponds to a global reflection of the lattice.
This relation suggests that our characterization of chirality in terms of complex conjugation is closely connected to spatial reflection symmetry.

Additional stabilizer generators arise from $\mathcal{G}\otimes \mathcal{G}^*$, which couple two layers, and are given by the following four-body operators:  
\begin{align}
g_{e_{BC}} \otimes g_{e_{BC}}^*  = \figbox{3.0}{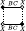}  \ , \quad
g_{e_{AC}} \otimes g_{e_{AC}}^*  =  \  \figbox{3.0}{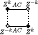} , \quad
g_{e_{AB}} \otimes g_{e_{AB}}^* =  \figbox{3.0}{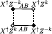} 
\label{Eq:weak_symmetry_stabilizer}
\end{align}
where the top and bottom represent two different layers.
While the edge operators $g_e$ generally do not commute, their doubled versions $g_e \otimes g_e^*$ commute due to cancellation of phases under complex conjugation. 

It is useful to note that the canonical purification $|\Psi_{\rho_{\mathcal{S}}}\rangle$, defined on two layers of honeycomb lattices, can be viewed as the thin 2D limit of the 3D $\mathbb{Z}_d^{(k)}$ anyon theory. 
From this perspective, the honeycomb construction captures boundary topological structures analogous to those appearing in Walker–Wang-type models~\cite{lee2025chiralcolorcode}.

\subsection{String operator and emergent $\mathbb{Z}_d^{(k)}$ anyons}
\label{sec:2.3} 

Here, we briefly review how the $\mathbb{Z}_d^{(k)}$ anyons emerge in the stabilizer mixed state $\rho_{\mathcal{S}}$.
It is convenient to work in the canonical purification picture $|\Psi_{\rho_{\mathcal{S}}}\rangle$, where weak-symmetry stabilizers $g \otimes g^*$ appearing in Eq.~\eqref{Eq:weak_symmetry_stabilizer} are included with the strong-symmetry stabilizers $S \otimes I$.
In this picture, anyons are identified as localized excitations created by string-like Pauli operators acting on the first copy.
A crucial point is that these string operators are not allowed to violate either strong-symmetry stabilizers $S\otimes I$ or weak-symmetry stabilizers $g\otimes g^*$ along the length of the string.
This constraint ensures that only the $\mathbb{Z}_d^{(k)}$ subtheory of the $\mathbb{Z}_d$ toric-code anyons survives as a consistent set of localized excitations.
This prescription is closely related to characterizing the anyon theory directly in terms of strong 1-form symmetries of $\rho_{\mathcal{S}}$~\cite{Sohal_2025,Ellison_2025}.

The simplest way to understand the $\mathbb{Z}_d^{(k)}$ anyon theory is to view it as a subset of the $\mathbb{Z}_d$ toric code anyons, where its anyon set is given by
\begin{align}
\mathcal{A}_d^{(k)} \equiv \{ a^j \}_{j=0}^{d-1}, \qquad a \equiv em^{k},
\end{align}
where $e$ and $m$ are charge and flux in the $\mathbb{Z}_d$ toric code. 
We call $a\equiv em^{k}$ an \emph{elementary anyon} generating $\mathbb{Z}_d^{(k)}$ anyons. 
Anyon braiding and self statistics are given by
\begin{align}\label{def:anyontheory}
B(a^{j_1}, a^{j_2}) = (\omega)^{2k j_1 j_2}, \qquad \theta(a^j) = (\omega)^{k j^2},\qquad \omega=e^{\frac{2\pi i}{d}}.
\end{align}
Several examples are discussed in Appendix~\ref{app:A1}.

Next, let us construct string operators and verify that $\mathbb{Z}_d^{(k)}$ anyons indeed emerge in the honeycomb mixed state $\rho_{\mathcal{S}}$. 

\subsubsection*{String operators}

Observe that face stabilizers of $\rho_{\mathcal{S}}$ defined in Eq.~\eqref{eq:stabilizer-generators-mixed_state} satisfy the global constraint, $\prod_f S_f = I$.
This relation suggests that one can form a strong symmetry loop stabilizer operator by multiplying $S_f$ on a connected region $R$:
\begin{align}
W_{\ell} = \prod_{f\in R} S_f,
\end{align}
where $\ell$ represents the boundary (loop) of $R$.
By truncating such loop operators, one obtains open string operators. An example is shown in Fig.~\ref{fig_loop}, where $W_{\ell}$ is decomposed into two string operators $M$ and $M'$.

\begin{figure}[h!]
\centering
\raisebox{\height}{\hspace{5pt}}\raisebox{-0.85\height}{\includegraphics[width=0.28\textwidth]{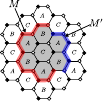}}
\hspace{10pt}
\caption{
A loop-like stabilizer $W_{\ell}$, and its truncations $M$ and $M'$. 
}
\label{fig_loop}
\end{figure}

A key property of these string operators, which arise from truncations of a strong symmetry loop $W_{\ell}$, is that they commute with both strong and weak symmetry generators ($S_f$ and $g_e$), except near their endpoints.
The resulting endpoint excitations can be interpreted as anyonic excitations of the stabilizer mixed state $\rho_{\mathcal{S}}$. 
Namely, in the canonical purification picture $|\Psi_{\rho_{\mathcal{S}}}\rangle$, a string operator $M\otimes I$ creates violations of $S_f \otimes I$ and $g_e \otimes g_e^*$ localized near the endpoints of the string.

It is useful to explicitly characterize the excitation pattern created by a string operator $M$:
\begin{align}
\figbox{3.0}{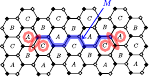} \label{eq:anyon_pair}
\end{align}
where violations of stabilizers are indicated by circles.
A string operator $M$ creates two violations of $S_f \otimes I$ and one violation of an edge-type stabilizer ($g_e \otimes g_e^*$) at each endpoint.
The expectation values of violated stabilizers are
\begin{align}
\figbox{3.0}{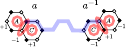}
\end{align}
where the exponents $j$ in $\omega^{jk}$ (i.e., $\omega^k$ or $\omega^{-k}$) are indicated.
These excitations can be identified with anyons $a$ and $a^{-1}$ in $\mathbb{Z}_d^{(k)}$, as we will verify below.
% \ZL{This figure is for $k=1$, right? Did we say that at somewhere? Otherwise, how does $k$ in Eq.\ref{eq:anyonbraiding} appear?} \TE{Sorry, do you have in mind $j$? $k$ is the index for the $\mathbb{Z}_d^{(k)}$, so it will determine the braiding relations of the generators in Eq.\ref{eq:anyonbraiding}.}
% \ZL{In is noted above this figure that the numbers ($\pm 1$) in the figure are eigenvalues of $S$ operators, so I think for $Z_d^{(k)}$ model the number should be $\pm k$. This also explains why the calculation in Eq.\ref{eq:anyonbraiding} will end up containing a $k$, giving rise to the desired braiding relation. } \TE{Oh, I see. Good catch. Or we can just change $\omega^j$ to $\omega^{jk}$ in the description below.}

Here, we observe that elementary anyonic excitations in this model involve pairs of plaquette-stabilizer violations.
One may then wonder whether other types of localized excitations are possible.
For instance, one might attempt to construct an isolated charge associated with a single plaquette violation.
However, string operators creating such single-plaquette excitations necessarily violate weak-symmetry stabilizers along the bulk of the string.
As a result, these excitations do not define localized anyons in the mixed-state theory.
This illustrates the usefulness of the canonical purification picture, where valid anyons are characterized by string operators that preserve both strong- and weak-symmetry stabilizers away from their endpoints.

As we demonstrate below, these string operators generate the $\mathbb{Z}_d^{(k)}$ anyons in the mixed state $\rho_{\mathcal{S}}$~\cite{lee2025chiralcolorcode}.
At a heuristic level, one is naturally led to expect that these string operators exhaust all possible anyonic excitations in $\rho_{\mathcal{S}}$.
A major technical challenge, however, is to rigorously establish that no additional anyon types can emerge beyond this construction.
In particular, one must rule out the possibility that more general or finely tuned string operators could generate excitations with different statistics, or that coherent superpositions of anyons could effectively mimic distinct anyon types.

In Sec.~\ref{sec4:LO chirality proof}, we prove rigorously that the $\mathbb{Z}_{d}^{(k)}$ anyons indeed exhaust all localized anyonic excitations in $\rho_{\mathcal{S}}$.
More precisely, we show that any localized excitation must correspond, up to local unitary transformations acting near its endpoints, to one of the anyons in the $\mathbb{Z}_d^{(k)}$ theory.
We further prove that the braiding statistics depend only on the anyon type itself.

\subsubsection*{Braiding statistics}

Anyon charges can be (partially) characterized by the braiding statistics. 
Let $W_{\ell}$ be a loop stabilizer which encloses an anyon $a$:
\begin{align}
\figbox{2.5}{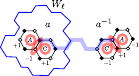}.
\end{align}
Recalling that $W_{\ell}$ is a product of $S_f$ inside $\ell$, we have
\begin{align}\label{eq:anyonbraiding}
B(a,a) = \Tr [M^{\dagger} W_{\ell} M \rho_{\mathcal{S}}] = \langle \Psi_{\rho_{\mathcal{S}}}|M^{\dagger} W_{\ell} M |\Psi_{\rho_{\mathcal{S}}} \rangle =  \omega^{2k}. 
\end{align}
In general, we find
\begin{align}
\text{$\langle W_{\ell} \rangle = \omega^{2kj}$ \quad  \text{for} \ $a^j$ anyon},
\end{align}
suggesting $B(a^j,a)= \omega^{2kj}$. 

When $d$ is odd and $k$ is coprime to $d$, anyon types $a^j$ can be fully characterized by the expectation value of $W_{\ell}$: $\langle W_{\ell} \rangle = \omega^{2kj}$.
For even $d$, however, anyons $a^{\frac{d}{2}}$ cannot be detected in this way, since 
\begin{align}
\langle W_{\ell} \rangle = \omega^{2k \frac{d}{2}}=1 \quad  \text{for} \ \text{$a^{\frac{d}{2}}$ anyon}
\end{align}
reflecting the fact that $a^{\frac{d}{2}}$ is transparent and braids trivially with all other anyons.\footnote{This can also be understood by recalling the 3D realization of $\mathbb{Z}_d^{(k)}$ models, such as the Walker-Wang model or the chiral color code \cite{lee2025chiralcolorcode}. 
In these models, the anyons $a^{\frac{d}{2}}$ correspond to $\mathbb{Z}_2$-charged fermions or bosons depending on whether $d=2$ or $d=0 \ (\mathrm{mod}\ 4)$.
Such excitations are not confined to the two-dimensional boundary, but can propagate into the three-dimensional bulk.
This implies that these excitations cannot be fully characterized by operators localized on a two-dimensional surface.
}
To capture these transparent anyons, it is again useful to work in the canonical purification picture $|\Psi_{\rho_{\mathcal{S}}}\rangle$.
% \TE{Here. More explicit explanation of extended loop.}
% \ZL{How about this: 
In the doubled system we define $\tilde{W}_\ell$ as the product of $g_e\otimes g_e^*$ over all edges inside and along the loop $\ell$. 
Because of Eq.~\eqref{eq:stabilizers-purestate}, the product of operators inside the loop cancels, leaving a loop-like operator supported on the doubled system along $\ell$ (and its mirror).
We term such an operator an extended loop operator.
It satisfies $\langle \tilde{W}_{\ell} \rangle = \omega^{kj}$ for the $a^j$ anyon.
% }
% One can construct extended loop operators by combining stabilizers over a region $R$ and its mirror region $R'$, involving operators of the form $S_f \otimes I$, $I \otimes S_f^* $, and $g_e \otimes g_e^* $. \newtext{More explicitly, these stabilizers are multiplied so that the local Pauli contributions away from the boundary of the doubled region $R\cup R'$ cancel, leaving a loop-like operator supported near that boundary. The doubled weak-symmetry stabilizers $g_e\otimes g_e^*$ provide the interlayer part of the extended loop, allowing it to probe charges in the effective bulk region between the two layers. These extended loop operators can therefore detect the $\mathbb{Z}_2$ charges that are invisible to ordinary loop operators $W_\ell$.}
% % These loop operators probe $\mathbb{Z}_2$ charges supported not only on the individual layers but also in the effective bulk region between them.

This feature introduces a mild technical subtlety in our analysis. In what follows, we will primarily focus on the case where $d$ is odd, for which all anyon types can be detected via loop operators. Nevertheless, our conclusions extend to even $d$ whenever the anyon theory is not invariant under complex conjugation.

\subsubsection*{Topological spin}

The topological spin (the self statistics) of an anyon can be extracted by the T-junction process, which involve three string operators $M_1, M_2, M_3$ (Fig.~\ref{fig_T_junction}):
\begin{align}  
\theta(a) = \Tr\Big[ M_3^{\dagger} M_2^{\dagger} M_1^{\dagger}  M_3 M_2 M_1 \rho_{\mathcal{S}}\Big] = \omega^k.
\end{align}
Moreover, one verifies $\theta(a^j) = \omega^{kj^2}$.

Importantly, the angle $\theta(a)$ is a topological invariant: it is independent of the specific choice of string operators $M_1, M_2, M_3$. An example of $M_1, M_2, M_3$, resulting from loop stabilizer operators, is illustrated in Fig.~\ref{fig_T_junction}.

\begin{figure}[h!]
\centering
\raisebox{\height}{\hspace{5pt}}\raisebox{-0.7\height}{\includegraphics[width=0.7\textwidth]{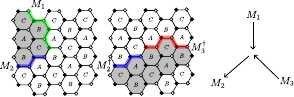}}
\hspace{10pt}
\caption{
Construction of string operators $M_1, M_2, M_3$ for T-junction.
}
\label{fig_T_junction}
\end{figure}

\subsection{$\mathbb{Z}_d^{(k)}$ anyons in mirror}
\label{sec:2.4}

Finally, we study when the $\mathbb{Z}_d^{(k)}$ anyon theory remains invariant under complex conjugation. The condition we derive in Theorem~\ref{thm:invcondition} will serve as the key criterion to determine when the $\mathbb{Z}_d^{(1)}$ mixed state becomes LO- or $n$-partite chiral in the following sections.

The complex conjugates of the $\mathbb{Z}_d^{(k)}$ anyons are obtained by replacing $k$ with $-k$, under which the anyon statistics are complex conjugated. Thus, our goal is to determine when two anyon theories, $\mathbb{Z}_d^{(k)}$ and  $\mathbb{Z}_d^{(-k)}$, are isomorphic:
\begin{align}
\mathcal{A}_d^{(k)} \underset{?}{\simeq} \mathcal{A}_d^{(-k)} 
\end{align}
meaning that the two theories have identical anyon statistics after appropriate relabelling.

We begin by recalling a basic fact that the $\mathbb{Z}_d^{(k)}$ anyon theory factorizes into prime-power building blocks, the $\mathbb{Z}_{p^m}^{(k)}$ anyons. 
By writing $d$ as the product of primes $p_j$: 
\begin{align}
d = p_1^{m_1} \cdots p_n^{m_n},
\end{align}
the following can be shown by the Chinese remainder theorem (CRT). (see Appendix~\ref{A6}).

\begin{lemma}[Factorization]\label{lemma:factorization}
The $\mathbb{Z}_d^{(k)}$ anyon theory is equivalent to the following tensor product:
\begin{align}
\mathcal{A}_d^{(k)} \simeq \boxtimes_{j} \mathcal{A}_{d_j}^{(k e_j)}
\end{align}
where $d_j \equiv p_j^{m_j}$ and $e_j \equiv \frac{d}{d_j}$.
\end{lemma}

An integer $q$ is a \emph{quadratic residue} modulo $d$ if there exists $\lambda$ such that
\begin{align}
\lambda^2 = q \qquad (\mathrm{mod\ } d).
\end{align}
Otherwise, $q$ is a \emph{quadratic non-residue}. 

\begin{lemma}[Quadratic residue]\label{lemma:residue}
The $\mathbb{Z}_d^{(k)}$ anyon theory is isomorphic to its complex conjugation, namely $\mathcal{A}_d^{(k)}\simeq \mathcal{A}_d^{(-k)}$ if and only if $-1$ is a quadratic residue modulo $d' = \frac{d}{\mathrm{gcd}(d,k)}$.
\end{lemma}

\begin{proof}
Let $a$ and $a^*$ be the elementary anyons for $\mathbb{Z}_d^{(k)}$  and $\mathbb{Z}_d^{(-k)}$ respectively. 
Assuming that $\mathcal{A}_d^{(k)} \simeq \mathcal{A}_d^{(-k)}$, there exists $\lambda$ such that $a^* \sim a^{\lambda}$ where $\sim$ represents an identification of anyon labels between two theories.
Since $\theta(a^*)=\omega^{-k}$ and $\theta(a^{\lambda})=\omega^{k \lambda^2}$, we must have 
\begin{align}
k\lambda^2 = - k \qquad (\mathrm{mod\ } d)
\end{align}
which is equivalent to
\begin{align}
\lambda^2 = - 1 \qquad (\mathrm{mod\ } d')
\end{align}
where $d' = \frac{d}{\mathrm{gcd}(d,k)}$. Hence, $-1$ is quadratic residue modulo $d'$. 

% The converse can be proven directly: given $\lambda^2 = -1$ (mod $d'$), we can identify $a^* \sim a^\lambda$. 
Conversely, given $\mu^2=-1$ (mod $d'$), let us find $\lambda$ such that $\lambda^2 = -1$ (mod $d'$) and $\mathrm{gcd}(\lambda,d)=1$.
Note that $\mu$ is coprime to $d'$, so a sufficient condition is $\lambda=1$ (mod $p$) for all $p\mid d$ but $p\nmid d'$, and $\lambda=\mu$ (mod $d'$).
Such $\lambda$ exists due to the CRT.
Then we can identify $a^* \sim a^\lambda$. 
It is straightforward to check $\mathcal{A}_d^{(k)}\simeq \mathcal{A}_d^{(-k)}$ under this identification.
% \ZL{edited the converse direction. I think the previous proof has a gap. We need $\mathrm{gcd}(\lambda,d)=1$.}
\end{proof}

From now on, we focus on the case where $(d,k)$ are coprime, as the analysis for non-coprime cases reduces to the case $d'=\frac{d}{\mathrm{gcd}(d,k)}$. 
The factorization property (Lemma~\ref{lemma:factorization}) implies that it suffices to focus on the $\mathbb{Z}_{d}^{(k)}$ anyon theory with $d= p^m$ and coprime $(d,k)$.
Namely, to determine whether $\mathbb{Z}_d^{(k)}$ anyon theory is invariant under its mirror image, we only need to verify whether $-1$ is a quadratic residue modulo $p^m$. We shall use the following number theoretic facts. 

\begin{fact}
\begin{enumerate}[i)]
\item Let $p$ be an odd prime. Then, $-1$ is quadratic residue modulo $p^m$ if and only if $p = 1$ $(\mathrm{mod\ } 4)$. 
\item For $p=2$, $-1$ is quadratic residue modulo $2^m$ only for $m=1$. 
\end{enumerate}
\end{fact}

This leads to the necessary and sufficient condition for $\mathcal{A}_d^{(k)} \simeq \mathcal{A}_d^{(-k)}$. 

\begin{theorem}[Mirror invariance]
\label{thm:invcondition}
The $\mathbb{Z}_d^{(k)}$ anyon theory (with coprime $d,k$) is isomorphic to its complex conjugation (the $\mathbb{Z}_d^{(-k)}$ anyon theory) if and only if
\begin{align}
d = \prod_j p_j^{m_j} \quad \mbox{or} \quad d = 2 \prod_j p_j^{m_j} 
\end{align}
where $p_j$ are odd primes satisfying $p_j = 1$ $(\mathrm{mod\ } 4)$.
\end{theorem}
Notably, the requirement for mirror invariance does not coincide with the condition for a zero chiral central charge or the existence of gapped boundary. 
We provide further discussion in Appendix~\ref{A7}.

\section{LO-chirality}\label{sec:3}
\label{sec4:LO chirality proof}

In the previous section, we have established the necessary and sufficient condition for when the anyon content satisfies $\mathcal{A}_d^{(k)}\simeq \mathcal{A}_d^{(-k)}$.
We now show that this criterion coincides exactly with the condition for LO-chirality. 
\begin{theorem}[LO-chirality]\label{thm:LO-chirality}
The $\mathbb{Z}_d^{(k)}$ mixed state $\rho$, as defined in Definition~\ref{def:honeycomb-mixed} is LO-chiral if and only if $\mathcal{A}_d^{(k)}\not\simeq \mathcal{A}_d^{(-k)}$.
\end{theorem}
One nontrivial aspect is to rigorously establish that the $\mathbb{Z}_d^{(k)}$ mixed state supports only the anyons in $\mathcal{A}_d^{(k)}$. Namely, one must rule out more complicated string operators whose endpoints carry charges outside $\mathcal A_d^{(k)}$, or coherent superpositions of different anyon charges.
A substantial portion of this section is devoted to establishing this property.

Although we focus on the specific models introduced in Definition~\ref{def:honeycomb-mixed}, we expect that our arguments extend more broadly to fixed-point mixed states realizing abelian chiral anyon theories.
This expectation is partly motivated by a recent result of Ref.~\cite{kim2024classifying2dtopologicalphases}, which showed that anyon theories with a Lagrangian subgroup can be transformed into string-net models via quasi-local quantum circuits.
Our analysis may provide useful insights toward such constructions for mixed states. 

For clarity, we focus on the $\mathbb{Z}_3^{(1)}$ anyon theory, where the proof can be presented in a particularly transparent form. 
The generalization to arbitrary cases with $\mathcal{A}_d^{(k)}\not\simeq \mathcal{A}_d^{(-k)}$ follows straightforwardly.
We will first establish the result in the LU setting, where the notation is simpler, and then explain how the argument extends to the LO setting in Sec.~\ref{sec:LO_generalization}. 

\subsection{Overview of the proof}
For the sake of deriving a contradiction, suppose that there exists a finite-depth local unitary $U$ of depth $r$ such that 
\begin{align} \label{eq: FDLU to rhostar}
U \rho U^{\dagger} = \rho^*.
\end{align} 
Recall that the anyon content of $\mathbb{Z}_3^{(1)}$ and $\mathbb{Z}_3^{(-1)}$ are distinct, i.e. $\mathcal{A}_3^{(1)}\not\simeq \mathcal{A}_3^{(-1)}$.
In particular, for $a\equiv em \in \mathcal{A}_3^{(1)}$ and $a^*\equiv em^{2} \in \mathcal{A}_3^{(-1)}$, we have
\begin{align}
\theta(a) = \theta(a^2) = \omega, \qquad \theta(a^*) = \theta({a^*}^2) = \omega^{2}.
\end{align}

\subsubsection*{String operators}

Consider Pauli string operators $N_1, N_2, N_3$ for $\rho^*$ that are part of loop stabilizers and create pairs of $a^*$ and ${a^*}^2$ anyons in $\rho^*$. 
Define their conjugates under $U$:
\begin{align}
\widetilde{N}_j \equiv U^{\dagger} N_j  U.
\end{align}
Since $U$ has depth $r$, each $\widetilde{N}_j$ is a ``fattened'' string operator of width $\sim 2r$. 

By construction, $\widetilde{N}_j $ commute with both strong and weak symmetry operators ($S_f$ and $g_e$) when acting on $\rho$, except near its endpoints where it creates excitations in $\rho$. Namely, since $N_j$ is part of a loop stabilizer for $\rho^*$, there exists another string operator $N_j'$ such that $W_j = N_j N_j'$ forms a loop stabilizer: 
\begin{align}
W_j  = \figbox{4.0}{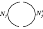}
\end{align}
Since $W_{j}$ is a strong symmetry of $\rho^*$, its conjugate $\widetilde{W_{j}}\equiv U^{\dagger} W_{j} U$ is a strong symmetry of $\rho$:
\begin{align}
\widetilde{W_{j}}\rho =  U^{\dagger} W_{j} U U^{\dagger} \rho^* U =
U^{\dagger} W_{j} \rho^* U =  U^{\dagger} \rho^* U  = \rho. \label{eq:strong_preserved}
\end{align}
Hence, $\widetilde{W_{j}}$ commutes with all strong and weak symmetries  when acting on $\rho$ (Fact.~\ref{fact:strong_commutation}). 
Since $\widetilde{W_{j}} = \widetilde{N}_j \widetilde{N}_{j}'$ and $\widetilde{N}_j,\widetilde{N}_{j}'$ overlap only at balls of radius $\sim r$, it follows that $\widetilde{N}_j'$ also commutes with strong and weak symmetry operators ($S_f,g_e$) away from its endpoints and therefore (at most) creates a pair of excitations in $\rho$. 

\subsubsection*{Fixed anyon charges}

We first show that $\widetilde{N}_j $ must create anyons $1$, $a$, or $a^2$ at its endpoints (i.e. charges $a^x$ and $a^{-x}$ for $x=0,1,2$):
\begin{align}
\figbox{4.0}{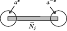}
\end{align}
That is, $\widetilde{N}_j $ creates excitations with \emph{fixed anyon charges} localized within balls of radius $\sim r$. 
Here, anyon charges are measured by eigenvalues of loop stabilizer operators $W_{\ell}$.

While one might expect this property intuitively, establishing it rigorously requires ruling out superpositions or any other unexpected forms of excitations. 
We establish this fact via a standard lightcone argument in Sec.~\ref{sec:anyon_charge}. 

\subsubsection*{Concentrating anyons}

Now consider the T-junction process for $\widetilde{N}_j$. On the one hand, the unitary equivalence between $\widetilde{N}_j$ and ${N}_j$ implies:
\begin{align}  
\Tr\Big[ \widetilde{N}_3^{\dagger} \widetilde{N}_2^{\dagger} \widetilde{N}_1^{\dagger}  \widetilde{N}_3 \widetilde{N}_2 \widetilde{N}_1 \rho\Big] = \Tr\Big[ N_3^{\dagger} N_2^{\dagger} N_1^{\dagger}  N_3 N_2 N_1 \rho^{*}\Big] = \omega^2. \label{eq:T-junction-conjugate}
\end{align}
Here we used the fact that $N_{j}$ are string operators for $a^*$ anyons in $\rho^*$, together with $U \rho U^{\dagger} = \rho^*$. On the other hand, $\rho$ only supports the anyons $a$ and $a^2$ anyons (with $\theta=\omega$). Thus, the T-junction process for $\widetilde{N}_{j}$ should instead yield $\omega$. Although this may seem intuitive, establishing this rigorously requires a careful argument. In particular, we need to show that the anyons created at an endpoint of $\widetilde{N}_{j}$ can be concentrated to a single excitation. 

To this end, we construct localized unitaries $U_{\text{left}}, U_{\text{right}}$, supported on balls of radius $\sim r$, that concentrate all endpoint excitations into a single anyon with a definite charge: 
\begin{align}
\figbox{4.0}{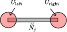} \ = \ \figbox{4.0}{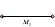} 
\end{align}
This construction follows from the uniqueness of Schmidt decomposition, together with certain decoupling properties, as discussed in Sec.~\ref{sec:concentration1}. 
Thus, on $\rho$, the action of $\widetilde{N}_j$ is equivalent to that of a Pauli string operator $M_{j}$ up to these local unitaries. 
In the canonical purification picture, this reads 
\begin{align}
(U_{\text{left}}U_{\text{right}} \widetilde{N}_j\otimes I) |\Psi_{\rho}\rangle = (M_j\otimes I) |\Psi_{\rho}\rangle. 
\end{align}

Finally, we show that the T-junction process is invariant under such local unitary operators. 
As shown in Sec.~\ref{sec:invariance}, we conclude 
\begin{align}  
\Tr\Big[ \widetilde{N}_3^{\dagger} \widetilde{N}_2^{\dagger} \widetilde{N}_1^{\dagger}  \widetilde{N}_3 \widetilde{N}_2 \widetilde{N}_1 \rho\Big] = \Tr\Big[ M_3^{\dagger} M_2^{\dagger} M_1^{\dagger}  M_3 M_2 M_1 \rho \Big]  = \omega. 
\end{align}
This contradiction rules out the existence of any finite-depth local unitary $U$ such that $U\rho U^{\dagger}=\rho^*$.

\subsection{Anyon charges are fixed}\label{sec:anyon_charge}

We begin by proving that a conjugated string operator $\widetilde{N}_j$ creates a pair of excitations with definite anyon charges at its endpoints. 

\begin{lemma}[Anyon charges]\label{lemma:anyon_charge}
Let $\widetilde{N}$ be a fattened string operator
obtained by conjugating a Pauli string with a local unitary circuit of depth $r$. Let $L$ and $R$ be balls of radius $2r$ containing the two endpoints of $\widetilde N$ %Choose slightly larger balls $L^+$ and $R^+$ around them 
and let $W_L$ and $W_R$ be loop stabilizers enclosing balls $L$ and $R$, respectively.
\begin{align}
 \figbox{4.0}{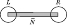} \ .
\end{align}
Specifically, for $\sigma = \widetilde{N} \rho \widetilde{N}^{\dagger}$, we have
\begin{align}
\Tr\big[W_{L} \sigma   \big] = \omega^j, \qquad
\Tr\big[ W_{R}\sigma   \big] = \omega^{-j}
\label{eq:anyon_decomposition}
\end{align}
for some $j=0,1,2$.
\end{lemma}

In other words, any LU-string operator $\widetilde{N}$ always creates a pair of anyons with \emph{fixed anyon charges} at endpoints, and the total charge is neutral. 
Namely, it does not create superpositions of anyonic excitations on $L$ and $R$.
The assumption that $\widetilde{N}$ is implementable by a finite-depth local unitary is crucial here, and thus the argument applies to abelian anyons only. 
Here, we focus on odd $d$ (in particular $d=3$), where loop stabilizers $W_{L}$ and $W_{R}$ suffice to characterize the anyon charge. 
For even $d$, the statement generalizes by using extended loop operators in the canonical purification (as discussed earlier).

\begin{proof}
Let $L^+$ and $R^+$ be balls of radius $\sim 4r$ that contain $L$ and $R$ respectively. Similarly, we introduce radius $\sim 3r$ balls $\tilde{L}$ and $\tilde{R}$ that contain $L$ and $R$ respectively.
For a sufficiently long string, $L^+$ and $R^+$ are decoupled: 
\begin{align}
\rho_{L^{+}R^{+}} = \rho_{L^{+}}\otimes \rho_{R^{+}}
\end{align}
since the mutual information between $L^+$ and $R^+$ vanishes in $\rho$.
As $\widetilde{N}$ has depth $r$, its action on $\widetilde{L}$ and $\widetilde{R}$ depend only on the degrees of freedom inside $L^+$ and $R^+$, respectively. Hence it cannot generate correlations between $\widetilde{L}$ and $\widetilde{R}$, so for $\sigma = \widetilde{N} \rho\widetilde{N}^{\dagger}$,
\begin{align}
\sigma_{\tilde{L}\tilde{R}} = \sigma_{\tilde{L}}\otimes \sigma_{\tilde{R}}.
\end{align}
Since $W_L$ and $W_R$ are supported on $\tilde{L}$ and $\tilde{R}$ respectively, we have 
\begin{align} \label{eq: factorize loop ends}
\Tr[ W_{ L} W_{ R} \sigma ] =\Tr[ W_{ L}\sigma ] \Tr[ W_{ R}\sigma ]  
\end{align}

Consider the loop stabilizer $W_{\mathrm{int}}$ connecting $W_{ L}$ and $W_{ R}$, such that $W_{\mathrm{int}}W_{ L}W_{ R}$ encloses $\widetilde{N}$:
\begin{align}
\figbox{4.0}{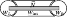} \ = \  \figbox{4.0}{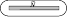}
\end{align}
Since $\widetilde{N}$ commutes with all stabilizer generators outside $L$ and $R$, it commutes with $W_{\mathrm{int}}$. %\TE{This is where having $W_L$ and $W_R$ strictly inside of $L$ and $R$ is an issue. We need $W_{\mathrm{int}}$ to be supported outside of $L$ and $R$. The resolution is clear though, we just need to have $W_L$ and $W_R$ to be supported in an annulus outside of $L$ and $R$. } \DL{Now the $W_L$ and $W_R$ are enclosing $L$ at its boundary. How do you think?} 
Hence, we have 
\begin{align}
\Tr[ W_{L} W_{ R} \sigma ] = \Tr[ W_{L} W_{ R} \widetilde{N} \rho\widetilde{N}^{\dagger} ]
= \Tr[ W_{L} W_{\mathrm{int}} W_{R} \widetilde{N} \rho\widetilde{N}^{\dagger} ].
\end{align}
Here, $W_{L} W_{\mathrm{int}} W_{R}$ is a large loop stabilizer enclosing the support of $\widetilde{N}$ and thus commutes with it. 
Hence we have 
\begin{align}
\Tr[ W_{L} W_{\mathrm{int}} W_{R} \widetilde{N} \rho\widetilde{N}^{\dagger} ] = 
\Tr[ \widetilde{N} W_{L} W_{\mathrm{int}} W_{R}\rho\widetilde{N}^{\dagger} ] =1.
\end{align}
Hence, we arrive at 
\begin{align}
\Tr[ W_{ L}\sigma ] \Tr[ W_{ R}\sigma ]  = \Tr[ W_{ L} W_{ R} \sigma ] = 1.
\label{eq:endpoint-charge-neutrality}
\end{align}

Since the operator norms of $W_{ L}$ and $W_{R}$ are bounded above by 1, $|\Tr[W_L\sigma]|$ and $|\Tr[W_R\sigma]|$ are bounded above by $1$; together with Eq.~\eqref{eq:endpoint-charge-neutrality}, this implies that both expectation values have unit modulus. 
Recalling that $W_L$ and $W_R$ have eigenvalues $\omega^j$, we obtain Eq.~\eqref{eq:anyon_decomposition} as desired.
% \YP{Slightly changed the wording, I think it is better to avoid ``maximum eigenvalue'' since the eigenvalues of $W_L,W_R$ are phases rather than real ordered numbers...}
% $\text{Tr}[W_{L} \sigma]$ and $\text{Tr}[W_R \sigma]$ should be same as the maximum eigenvalues of $W_{L}$ and $W_{R}$, respectively.  
% Recalling that $W_{L}$ and $W_{R}$ have eigenvalues $\omega^j$, we arrive at the desired result.
\end{proof}

\subsection{Concentrating anyons locally}\label{sec:concentration1}

The previous Lemma~\ref{lemma:anyon_charge} established that an LU-string operator $\widetilde{N}$ creates a pair of excitations with definite anyon charges inside balls $L$ and $R$ at its endpoints. 
Here, we show that these delocalized charges can be concentrated into point-like anyons by applying a local unitary $U_{LR} = U_{L}\otimes U_{R}$, supported only on $L$ and $R$. 

\begin{lemma}[Anyon concentration]\label{lemma_concentration}
Let $\widetilde{N}$ and $M$ be a pair of LU-string operators that create the same anyon charges inside balls $L$ and $R$ at their endpoints (as measured by loop stabilizers $W_L$ and $W_R$ respectively): 
\begin{align}
\figbox{4.0}{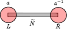} \ , \qquad 
\figbox{4.0}{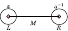}
\end{align}
Then, there exists a unitary operator $U_{LR}=U_{L}\otimes U_{R}$ such that 
\begin{align}
U_{LR}\widetilde{N} \rho = M \rho
\end{align}
or, equivalently in the canonical purification picture,
\begin{align}
(U_{LR}\widetilde{N}\otimes I) |\Psi_{\rho}\rangle = (M\otimes I) |\Psi_{\rho}\rangle.
\end{align}
\end{lemma}

Before starting the proof, let us emphasize why this statement is non-trivial.
It is relatively straightforward to show, e.g. using the Petz recovery map, that there exists a local quantum channel capable of concentrating anyons to single points. 
However, this is not sufficient for our purposes, as we must establish the existence of a local \emph{unitary} operator that achieves the same task.

At first glance, one might expect that applying standard quantum error correction, e.g. by measuring stabilizers and annihilating anyon pairs, would concentrate anyons to a single point. 
However, such a procedure naturally produces a quantum
channel rather than the unitary operation required here.
In the canonical purification picture, this issue is sharpened by the fact that the two states $(\widetilde{N}\otimes I)|\Psi_\rho\rangle$ and $(M\otimes I)|\Psi_\rho\rangle$ contain excitations not only on strong symmetry stabilizers $S_f\otimes I$, but also on weak symmetry stabilizers $g_e\otimes g_e^*$.
Measuring $g_e \otimes g_e^*$ necessarily involves the second-layer Hilbert space, and therefore does not correspond to a unitary supported solely on the first layer.

Below, we establish the existence of such a unitary operator by relying on the uniqueness of Schmidt decomposition, together with certain decoupling properties. 

\subsubsection*{Schmidt decomposition}

We begin with a standard fact. 

\begin{fact}[Schmidt decomposition]\label{lemma:Uhlmann}
Let $|\Phi\rangle$ and $|\tilde{\Phi}\rangle$ be pure states on a bipartition $A\otimes B$. 
If their marginals on $B$ coincide, 
\begin{align}
\Tr_{A}\big[ |\Phi\rangle \langle \Phi |  \big] = \Tr_{A}\big[ |\tilde{\Phi}\rangle \langle \tilde{\Phi} |  \big]
\end{align}
then there exists a unitary $U_A \otimes I_B$ such that
\begin{align}
(U_A \otimes I_B)  |\Phi\rangle =  |\tilde{\Phi}\rangle.
\end{align}
\end{fact}

Next, we show that the states created by $\widetilde{N}$ and $M$ agree on the complement of $L\cup R$ in the canonical purification picture (Fig.~\ref{fig_marginals}). 
Note that, since we are working with the canonical purification, $(LR)^{c}$ includes the second-layer Hilbert space.

\begin{figure}
\centering
\raisebox{\height}{\hspace{5pt}}\raisebox{-0.85\height}{\includegraphics[width=0.7\textwidth]{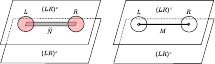}}
\hspace{10pt}
\caption{
In the canonical purification picture, the two states $|\tilde{\Phi}\rangle \equiv (\widetilde{N}\otimes I) |\Psi_{\rho}\rangle$ and $|\Phi \rangle \equiv (M\otimes I) |\Psi_{\rho}\rangle$ coincide on $(LR)^c$. 
}
\label{fig_marginals}
\end{figure}

\begin{lemma}[Marginals]\label{lemma:marginals}
With $|\tilde{\Phi}\rangle \equiv (\widetilde{N}\otimes I) |\Psi_{\rho}\rangle$, and $|\Phi \rangle \equiv (M\otimes I) |\Psi_{\rho}\rangle$, we have
\begin{align}
\Tr_{LR} \big( |\tilde{\Phi}\rangle \langle  \tilde{\Phi} | \big) = \Tr_{LR} \big( |\Phi\rangle \langle  \Phi  | \big). \label{eq:marginal}
\end{align}
\end{lemma}

\begin{proof}
Observe that $\Tr_{LR}(|\Phi \rangle \langle \Phi|)$ is a stabilizer mixed state -- namely, the maximally mixed state supported within the stabilizer subspace, defined by the stabilizer group $\mathcal{S}_{(LR)^c}$ consisting of all the stabilizers of $|\Phi \rangle$ contained within $(LR)^c$.
Our goal is to show that $\Tr_{LR}(|\tilde{\Phi} \rangle \langle \tilde{\Phi}|)$ is the same stabilizer mixed state.

We first show that $\Tr_{LR}(|\tilde{\Phi} \rangle \langle \tilde{\Phi}|)$ is supported in the stabilizer subspace of $\mathcal{S}_{(LR)^c}$. 
None of the local stabilizer generators $S_f\otimes I$, $I\otimes S_f^{*}$, and $g_e\otimes g_e^{*}$ supported on $(LR)^c$ are violated by  $\widetilde{N}\otimes I$ or $M\otimes I$.
% All local stabilizers, $S_f\otimes I$, $I\otimes S_f^{*}$, and $g_e\otimes g_e^{*}$ supported on $(LR)^c$, are unaffected by either $\widetilde{N}\otimes I$ or $M\otimes I$.
Therefore,  $\Tr_{LR}(|\tilde{\Phi} \rangle \langle \tilde{\Phi}|)$ has eigenvalue $+1$ under all such local stabilizers. 
The only non-trivial stabilizer on $(LR)^c$, which cannot be generated by local stabilizers on $(LR)^c$, is the loop stabilizer $W_{\partial L}$. Since $W_{\partial R}$ can be generated from $W_{\partial L}$ (their product being locally generated), it suffices to check $W_{\partial L}$. 
By assumption, $\widetilde{N}$ and $M$ create identical anyon charges inside $L$ and $R$, so both yield the same expectation values for $W_{\partial L}$ and $W_{\partial R}$. 
This proves that $\Tr_{LR}(|\tilde{\Phi} \rangle \langle \tilde{\Phi}|)$ is supported in the same stabilizer subspace of $\mathcal{S}_{(LR)^c}$. For even $d$, the same reasoning applies using the extended loop stabilizer. 

The remaining task is to prove that $\Tr_{LR}(|\tilde{\Phi} \rangle \langle \tilde{\Phi}|)$ is maximally mixed inside this subspace. 
For this purpose, we introduce slightly smaller balls $L^-$ and $R^-$ contained inside $L$ and $R$ respectively, such that the loop operators $W_{L^-}$ and $W_{R^-}$ still measure the same anyon charges for both $|\tilde{\Phi} \rangle$ and $|\Phi \rangle$. 
Such regions can always be chosen by enlarging $L$ and $R$ if necessary.
We denote the annular regions surrounding $L^-$ and $R^-$ by $L_{\Delta}$ and $R_{\Delta}$, respectively, so that
\begin{align}
L = L^- \cup L_{\Delta}, \qquad
R = R^- \cup R_{\Delta}.
\end{align}

Let $\mathcal{L}_{(L^{-}R^{-})^c}$ be the set of Pauli logical operators associated with the stabilizer code $\mathcal{S}_{(L^{-}R^{-})^c}$. 
Consider the quantum channel defined by twirling over these Pauli logical operators:
\begin{align}
\mathcal{Q}_{(L^{-}R^{-})^c } (\cdot) = \frac{1}{|\mathcal{L}_{(L^{-}R^{-})^c}|} \sum_{P \in \mathcal{L}_{(L^{-}R^{-})^c}} P (\cdot) P^{\dagger}. 
\end{align}
This quantum channel acts as a depolarizing channel within the stabilizer subspace of  $\mathcal{S}_{(L^{-}R^{-})^c}$, and therefore produces the maximally mixed state in that subspace.
Hence, we have 
\begin{align}
\mathcal{Q}_{(L^{-}R^{-})^c } \big( \Tr_{L^{-}R^{-}}(|\tilde{\Phi} \rangle \langle \tilde{\Phi}|) \big) 
= \Tr_{L^{-}R^{-}}(|\Phi \rangle \langle \Phi|). 
\end{align}

Next, we classify the logical operators in $\mathcal{L}_{(L^{-}R^{-})^c}$.
There are two types of logical operators.
The first type consists of weak-symmetry operators supported entirely on the second layer.\footnote{Specifically, the elements of $\mathcal{G}^*$ (see Sec.~\ref{sec:2.2}) lying below the regions $L$ and $R$ are weak-symmetry operators.}
The second type consists of logical operators supported on $L_{\Delta}$ and $R_{\Delta}$, corresponding to truncated stabilizer and weak-symmetry operators (see Sec.~\ref{sec:4-partite_additional}).
The first-type logical operators are supported entirely on the second layer, and therefore commute with both $M\otimes I$ and $\widetilde{N}\otimes I$. The second-type logical operators are supported within $L_{\Delta}R_{\Delta}$ and therefore disappear after tracing out $L_{\Delta}R_{\Delta}$.
Thus, inserting the logical twirl before tracing out $L_{\Delta}R_{\Delta}$ leaves the reduced state unchanged:
\begin{align}
 \Tr_{LR}(|\tilde{\Phi} \rangle \langle \tilde{\Phi}|) = 
 \Tr_{L_{\Delta}R_{\Delta}} \Tr_{L^- R^- }(|\tilde{\Phi}  \rangle \langle \tilde{\Phi}|) 
 =  \Tr_{L_{\Delta}R_{\Delta}}  \mathcal{Q}_{(L^{-}R^{-})^c } \big( \Tr_{L^- R^- }(|\tilde{\Phi}  \rangle \langle \tilde{\Phi}|) \big).
\end{align}
Using the previous identity for the twirled state, we obtain
\begin{align}
 \Tr_{LR}(|\tilde{\Phi} \rangle \langle \tilde{\Phi}|)=  \Tr_{L_{\Delta}R_{\Delta}}  \Tr_{L^- R^- }(|\Phi  \rangle \langle \Phi|)  =  \Tr_{LR}(|\Phi \rangle \langle \Phi|).
\end{align}

\end{proof}

\subsubsection*{Factorization of $U_{LR}$}

Lemma~\ref{lemma:marginals}, together with Fact~\ref{lemma:Uhlmann}, implies the existence of a local unitary $U_{LR}$. 
It remains to show that there exists a $U_{LR}$ that factorizes as $U_{LR}=U_{L} \otimes U_R$. 
This follows from a refinement of Fact~\ref{lemma:Uhlmann}, applied to four partitions. 

\begin{lemma}\label{lemma:four_partition}
Consider a pure state $|\Psi\rangle$ on $LK_{L}K_{R}R$ satisfying
\begin{align}
I(L: K_{R}R) = I(R: K_{L}L) = 0.
\end{align}
Then there exist isometries $V_{K_{L}}: \hat{L} \otimes \hat{K}_{L} \rightarrow K_{L} $ and $V_{K_{R}}: \hat{R} \otimes \hat{K}_{R} \rightarrow K_{R}$ such that 
\begin{align}
|\Psi\rangle = (I_{L} \otimes V_{K_{L}} \otimes V_{K_{R}} \otimes I_{R} ) |\psi\rangle_{L\hat{L}} \otimes |\psi\rangle_{\hat{K}_{L}\hat{K}_{R}}  \otimes |\psi\rangle_{R\hat{R}} 
\end{align}
for some states $ |\psi\rangle_{L\hat{L}}$, $ |\psi\rangle_{\hat{K}_{L}\hat{K}_{R}}$ and $ |\psi\rangle_{R\hat{R}}$, or equivalently 
\begin{align}
|\Psi\rangle = 
\figbox{4.0}{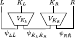} \ .
\end{align}
\end{lemma}

\begin{proof}
Since $I(L:K_R R)=0$, the reduced density matrix of $|\Psi\rangle$ on $LK_RR$ factorizes as $\rho_{L K_R R} = \rho_L \otimes \rho_{K_R R}$. One purification of this state is $|\Psi\rangle$. Another purification comes from purifying $\rho_L$ into $|\psi\rangle_{L\hat{L}}$ and $\rho_{K_R R}$ into $|\psi\rangle_{\hat{K}_L K_R R}$. 
Fact~\ref{lemma:Uhlmann} then implies the existence of an isometry 
$V_{K_L}:\hat{L} \otimes \hat{K}_L \to K_L$ such that
\begin{align}
|\Psi\rangle 
= (I_L \otimes V_{K_L} \otimes I_{K_R R}) 
\bigl( |\psi\rangle_{L\hat{L}} \otimes |\psi\rangle_{\hat{K}_L K_R R} \bigr).
\end{align}
Next, the condition $I(R:K_L L)=0$ implies $I(R:\hat{K}_L)=0$ in the state 
$|\psi\rangle_{\hat{K}_L K_R R}$.
Applying the same argument again, there exists an isometry 
$V_{K_R}:\hat{R} \otimes \hat{K}_R \to K_R$ such that
\begin{align}
|\psi\rangle_{\hat{K}_L K_R R}
= (I_{\hat{K}_L} \otimes V_{K_R} \otimes I_R)
\bigl( |\psi\rangle_{\hat{K}_L \hat{K}_R} \otimes |\psi\rangle_{R\hat{R}} \bigr).
\end{align}
Combining the two decompositions proves the claim.
\end{proof}

\begin{figure}
\centering
\raisebox{\height}{\hspace{5pt}}\raisebox{-0.85\height}{\includegraphics[width=0.35\textwidth]{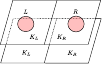}}
\hspace{10pt}
\caption{
The system is divided into four subsystems $L, K_{L}, K_{R}, R$ in the canonical purification picture. Here, $L$ ($R$) is decoupled from $K_{R}R$ ($K_{L}L$). 
}
\label{fig_four_regions}
\end{figure}
% \noindent 

Lemma~\ref{lemma:four_partition} is particularly useful for topologically ordered systems, since they naturally satisfy its assumptions. 
% For the purpose of showing that there exists a $U_{LR}$ that factorizes on $L$ and $R$, 
In particular, we consider a canonical purification on the four subsystems $L,R,K_{L},K_{R}$, as illustrated in Fig.~\ref{fig_four_regions}.
Here, $L$ and $R$ are balls on the top (first) layer while $K_{L}$ and $K_{R}$ occupy both layers such that $L\cup K_{L}$ and $R \cup K_{R}$ divide the system into two halves. 
In this setup, one can verify that $L$ is entangled only with $K_{L}$, and $R$ only with $K_{R}$:
\begin{align}
I(L:RK_R)=I(R:L K_L) = 0,\quad\text{for both $\ket{\Phi}$ and $\ket{\tilde\Phi}$}.
\end{align}
% so that $L$ is decoupled from $RK_{R}$ and $R$ is decoupled from $LK_{L}$. 

Given this decomposition into subsystems, we start by applying Lemma~\ref{lemma:four_partition} to $|\Phi\rangle=(M\otimes I) |\Psi_{\rho}\rangle$. This gives us
\begin{align}
V_{K_L}^\dagger V_{K_R}^\dagger|\Phi\rangle =  |\phi_l\rangle_{L\hat{L}} \otimes |\phi_k\rangle_{\hat{K}_{L}\hat{K}_{R}}  \otimes |\phi_r\rangle_{R\hat{R}}. 
\end{align}
To simplify the notation, we write this state as
\begin{align}
    |\chi\rangle = |\phi_l\rangle_{L\hat{L}} \otimes |\phi_k\rangle_{\hat{K}_{L}\hat{K}_{R}}  \otimes |\phi_r\rangle_{R\hat{R}}.
\end{align}
A similar decomposition holds for $|\tilde{\Phi}\rangle$, but with potentially different isometries $\tilde{V}_{K_L}, \tilde{V}_{K_R}$. %$V_{L}', V_{R}'$.
Our first goal is to show that, in fact, the same isometries $V_{K_L}, V_{K_R}$
%$V_{L}, V_{R}$ 
suffice for $|\tilde{\Phi}\rangle$. 

Applying $V_{K_L}^\dagger$ and $V_{K_R}^\dagger$ to $|\tilde{\Phi}\rangle$, we claim that: 
\begin{align}
V_{K_L}^{\dagger} V_{K_R}^{\dagger}|\tilde{\Phi}\rangle =  %(I_{L} \otimes V_{K_{L}} \otimes V_{K_{R}} 
%\otimes I_{R} ) 
 |\tilde{\phi}_{lr}\rangle_{L\hat{L}R\hat{R}} \otimes |\phi_k\rangle_{\hat{K}_{L}\hat{K}_{R}},
\end{align}
which we define as
\begin{align}\label{eq: tilde chi}
    |\tilde{\chi}\rangle = |\tilde{\phi}_{lr}\rangle_{L\hat{L}R\hat{R}} \otimes |\phi_k\rangle_{\hat{K}_{L}\hat{K}_{R}}.
\end{align}
To see that $|\tilde{\chi}\rangle$ admits this form, we first recall that the marginals of $|\Phi\rangle$ and $|\tilde{\Phi}\rangle$ agree on $K_LK_R$, i.e.,  $\rho_{K_LK_R}^\Phi = \rho_{K_LK_R}^{\tilde{\Phi}}$. Then, applying the isometries $V_{K_L}, V_{K_R}$ to both sides of the equation, we obtain $\rho_{\hat{L}\hat{K}_L\hat{K}_R\hat{R}}^\chi = \rho_{\hat{L}\hat{K}_L\hat{K}_R\hat{R}}^{\tilde{\chi}}$. Tracing out $\hat{L}\hat{R}$ gives us $\rho_{\hat{K}_L\hat{K}_R}^\chi = \rho_{\hat{K}_L\hat{K}_R}^{\tilde{\chi}}$. This shows us that $\rho_{\hat{K}_L\hat{K}_R}^{\tilde{\chi}}$ must be the density matrix for the pure state $|\phi_k\rangle_{\hat{K}_{L}\hat{K}_{R}}$. Thus, $|\tilde{\chi}\rangle$ must take the form in Eq.~\eqref{eq: tilde chi}.

We further need to show that $|\tilde{\chi}\rangle$ factorizes on $L\hat{L}$ and $R\hat{R}$. We accomplish this by showing that $S_{\tilde{\chi}}(L\hat{L})=0$. We begin by computing $S_{\tilde{\Phi}}(LK_L)$. Due to the invariance of this entropy under the isometry $V_{K_L}^\dagger$, we have
\begin{align} \label{eq: entropy 1}
    S_{\tilde{\Phi}}(LK_L) = S_{\tilde{\chi}}(L\hat{L}\hat{K}_L).
\end{align}
Given the factorization in Eq.~\eqref{eq: tilde chi}, the right-hand side decomposes as:
\begin{align} \label{eq: entropy 2}
    S_{\tilde{\chi}}(L\hat{L}\hat{K}_L) = S_{\tilde{\chi}}(L\hat{L})+S_{\phi_k}(\hat{K}_L)
\end{align}
Similarly, we can compute $S_{{\Phi}}(LK_L)$, which gives:
\begin{align} \label{eq: entropy 3}
    S_{{\Phi}}(LK_L) = S_{\phi_k}(\hat{K}_L).
\end{align}
Here, the factor $S_{{\chi}}(L\hat{L})$ vanishes, because it is a tensor product.
Putting together Eqs.~\eqref{eq: entropy 1}, \eqref{eq: entropy 2}, and \eqref{eq: entropy 3}, we have
\begin{align}
  S_{\tilde{\chi}}(L\hat{L}) =  S_{\tilde{\Phi}}(LK_L) -S_{{\Phi}}(LK_L).
\end{align}

Thus, it only remains to show that $S_{\tilde{\Phi}}(LK_L) = S_{{\Phi}}(LK_L)$. To do so, we consider the mutual information $I_\Phi(K_L:K_R)$:
\begin{eqs}
    I_\Phi(K_L:K_R) 
    &= S_\Phi(K_L) + S_\Phi(K_R)-S_\Phi(K_LK_R) \\
    &= S_\chi(\hat{L}\hat{K}_L) + S_\chi(\hat{K}_R\hat{R})-S_\chi(\hat{L}\hat{K}_L\hat{K}_R\hat{R}) \\
    &= \left[  S_{\phi_l}(\hat{L}) +S_{\phi_k}(\hat{K}_L) \right ] + \left[  S_{\phi_r}(\hat{R}) +S_{\phi_k}(\hat{K}_R) \right ] - \left[  S_{\phi_l}(\hat{L}) +S_{\phi_r}(\hat{R}) \right ] \\
    &= 2S_{\phi_k}(\hat{K}).
\end{eqs}
In the second equality, we used the invariance of the mutual information under the isometries $V^\dagger_{K_L}, V^\dagger_{K_R}$. In the third equality, we used the explicit form of $|\chi\rangle$. The last equality follows from the fact that $|\phi_k\rangle_{\hat{K}_L\hat{K}_R}$ is pure. 
Looking back at Eq.~\eqref{eq: entropy 3}, we can now write
\begin{align}
    S_{{\Phi}}(LK_L) = \frac12 I_\Phi(K_L:K_R).
\end{align}

The above equation follows from the structure established in Lemma~\ref{lemma:four_partition}, and that $\hat{L}$, $\hat{K_L}$, $\hat{R}$, $\hat{K_R}$ do not appear explicitly.
The same computation applies to $|\tilde\Phi\rangle$ (using the existence of isometries $\tilde{V}_{K_L}, \tilde{V}_{K_R}$ for $|\tilde{\Phi}\rangle$, which are a priori not required to coincide with $V_{K_L}, V_{K_R}$):
\begin{align}
     S_{{\tilde{\Phi}}}(LK_L) = \frac12 I_{\tilde{\Phi}}(K_L:K_R).
\end{align}
Finally, since the mutual information between $K_L$ and $K_R$ only depends on the reduced density matrix on $K_LK_R$, which is the same for $|\Phi\rangle$ and $|\tilde{\Phi}\rangle$, we arrive at
\begin{align}
    S_{{\Phi}}(LK_L) = S_{{\tilde{\Phi}}}(LK_L).
\end{align}
Thus, $S_{\tilde{\chi}}(L\hat{L})=0$, implying that $|\tilde{\chi}\rangle$ factorizes on $L\hat{L}$ and $R\hat{R}$:
\begin{align}
    |\tilde{\chi}\rangle = |\tilde{\phi}_l\rangle_{L\hat{L}} \otimes |\phi_k\rangle_{\hat{K}_{L}\hat{K}_{R}}  \otimes |\tilde{\phi}_r\rangle_{R\hat{R}}.
\end{align}

We are now able to argue that there is a unitary $U_{LR}$ that factorizes on $L$ and $R$. The states $|{\phi}_l\rangle_{L\hat{L}}$ and $|\tilde{\phi}_l\rangle_{L\hat{L}}$ have the same marginal on $\hat{L}$, as a consequence of $|\Phi\rangle$ and $|\tilde{\Phi}\rangle$ having the same marginal on $K_LK_R$. Therefore, Fact~\ref{lemma:Uhlmann} tells us that there exists a unitary $U_L$ that maps $|\tilde{\phi}_l\rangle_{L\hat{L}}$ to $|{\phi}_l\rangle_{L\hat{L}}$. Likewise, there is a unitary $U_R$ supported on $R$. Since the isometries are not supported on $L$ and $R$, we have our desired factorization of $U_{LR}$: 
\begin{align}
    U_L\otimes U_R |\tilde{\Phi}\rangle = |\Phi\rangle.
\end{align}

\subsection{Topological spin depends only on anyon type}\label{sec:invariance}

%We have established that LU-string operators create excitations with definite anyon total charges which can be concentrated into point-like anyons by a local unitary supported on balls localized on endpoints of the strings. 
Finally, we prove that the topological spin depends only on the anyon type. 
Here, we characterize the topological spin by the T-junction process. 
Let $R_0, R_1,R_2,R_3$ be balls of finite radius as shown in Fig.~\ref{fig_T_junction_balls}. 
We consider a triple of LU-transformed string operators $\widetilde{N}_1, \widetilde{N}_2, \widetilde{N}_3$ that transport excitations between $R_j$'s as follows:
\begin{align}
\widetilde{N}_1: R_1 \rightarrow R_0 \qquad  
\widetilde{N}_2: R_0 \rightarrow R_2 \qquad
\widetilde{N}_3: R_3 \rightarrow R_0.
\end{align}
Furthermore, we assume that these LU-string operators can be combined to form longer LU-string operators, namely 
\begin{align}
\widetilde{N}_2 \widetilde{N}_1: R_1 \rightarrow R_2 \qquad  
\widetilde{N}_1^{\dagger}\widetilde{N}_3: R_3 \rightarrow R_1 \qquad
\widetilde{N}_3^{\dagger}\widetilde{N}_2^{\dagger}: R_2 \rightarrow R_3.
\end{align}
Such a triple of LU-string operators $\widetilde{N}_1, \widetilde{N}_2, \widetilde{N}_3$ defines the topological spin:
\begin{align}
\theta(\widetilde{N}_1, \widetilde{N}_2, \widetilde{N}_3) \equiv \Tr\Big[ \widetilde{N}_3^{\dagger} \widetilde{N}_2^{\dagger} \widetilde{N}_1^{\dagger}  \widetilde{N}_3 \widetilde{N}_2 \widetilde{N}_1 \rho\Big]. 
\end{align}

\begin{lemma}[Topological spin]\label{lemma_spin}
Let $\widetilde{N}_j$ and $M_j$ ($j=1,2,3$) be triples of LU-string operators that create the same anyon charges inside balls $R_j$ (as measured by loop stabilizers $W_{R_j}$). Then, their topological spins match
\begin{align}
\theta(\widetilde{N}_1, \widetilde{N}_2, \widetilde{N}_3) = \theta(M_1, M_2, M_3).
\end{align}
\end{lemma}

Combining with Lemma~\ref{lemma:anyon_charge}, this lemma rigorously proves that the topological spin depends only on the anyon type. 
This establishes the contradiction with Eq.~\eqref{eq:T-junction-conjugate}, and thus proves the absence of local unitary $U$ achieving $U\rho U^{\dagger}=\rho^*$.
This completes the proof of Theorem~\ref{thm:LO-chirality}. 

\begin{proof}
Without loss of generality, we take $M_j$ to be Pauli strings. 
We use local unitaries $U_0,U_1,U_2,U_3$ on these balls to concentrate anyon excitations. 
First, using lemma~\ref{lemma_concentration}, we construct 
\begin{align}
\overline{N}_1 = U_{1}U_{0} \widetilde{N}_1, \quad \mbox{such that} \quad \overline{N}_1 \rho = M_1 \rho
\end{align}
by choosing $U_0, U_1$ supported on $R_0, R_1$. 
% Here, we adjust the overall $U(1)$ phase of $U_1$ so that $\overline{N}_1 \rho$ matches $M_1 \rho$ exactly, including the phase. 
Similarly, define 
\begin{align}
\overline{N}_3 = U_{3}U_{0} \widetilde{N}_3, \quad \mbox{such that} \quad \overline{N}_3 \rho = M_3 \rho
\end{align}
by using $U_3$ on $R_3$. 
Here, we can reuse the same $U_0$ because $\overline{N}_1$ and $\overline{N}_3$ create the same excitation pattern on $R_{0}$. 
We adjust the $U(1)$ phase of $U_3$ so that $\overline{N}_3 \rho$ and $M_3 \rho$ agree exactly. 
Next, note that $\widetilde{N}_2 \widetilde{N}_1$ is a string operator connecting $R_1$ and $R_2$, and creates no excitations in $R_0$.
Define
\begin{align}
\overline{N}_2 \overline{N}_1 = U_{2}U_{1} \widetilde{N}_2\widetilde{N}_1, \quad \overline{N}_2 = U_{2} \widetilde{N}_2 U_{0}^{\dagger} \quad \mbox{such that} \quad \overline{N}_2 \overline{N}_1 \rho = M_2 M_1 \rho. 
\label{eq:N2N1-phase-fixing}
\end{align}
Once again we adjust the $U(1)$ phase of $U_2$ to match $M_2 M_1 \rho$ exactly. 

\begin{figure}[h!]
\centering
\raisebox{\height}{\hspace{5pt}}\raisebox{-0.85\height}{\includegraphics[width=0.35\textwidth]{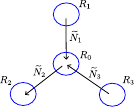}}
\hspace{10pt}
\caption{
Concentrating anyons in the T-junction.
}
\label{fig_T_junction_balls}
\end{figure}

% Recall that the string operator $\overline{N}_2 \overline{N}_1$ is deformable. 
% Namely, there exists another string operator $\overline{N}_{21}$ starting from $R_1$ and ending at $R_2$, with no overlap with $\overline{N}_3$ such that 
% \begin{align}
% \overline{N}_2 \overline{N}_1 \rho = \overline{N}_{21}\rho.
% \end{align}
% Using this, we have 
% \begin{eqs}
% \overline{N}_3 \overline{N}_2 \overline{N}_1 \rho &= \overline{N}_3 \overline{N}_{21}\rho 
% = \overline{N}_{21}\overline{N}_3 \rho 
% = \overline{N}_{21}M_3 \rho \\ 
% &=M_3  \overline{N}_{21}\rho = M_3   \overline{N}_2 \overline{N}_1\rho 
%  = M_3   M_2 M_1\rho. \label{eq:321}
%  \end{eqs}

Recall that the string operator $M_2 M_1$ is deformable. 
Namely, there exists another string operator $M_{21}$ starting from $R_1$ and ending at $R_2$, with no overlap with $\overline{N}_3$ such that 
\begin{align}
M_2M_1 \rho = M_{21}\rho.
\end{align}
Using this, we have 
\begin{eqs}
\overline{N}_3 \overline{N}_2 \overline{N}_1 \rho &= 
\overline{N}_3 M_2 M_1 \rho =
\overline{N}_3 M_{21}\rho 
= M_{21}\overline{N}_3 \rho 
= M_{21}M_3 \rho  
=M_3  M_{21}\rho \\
 &= M_3   M_2 M_1\rho. \label{eq:321}
 \end{eqs}

We also need to establish that 
\begin{align}
\overline{N}_1 \overline{N}_2 \overline{N}_3 \rho  
 = M_1   M_2 M_3\rho. \label{eq:123}
\end{align}
Using the same deformability argument for $M_2M_3$, it suffices to prove $\overline{N}_2 M_3 \rho = M_2 M_3\rho$. 
Observe that: 
\begin{equation}
    \begin{aligned}
\overline{N}_2 M_3 \rho  
 = M_2 M_3\rho \ 
 \quad&\text{if and only if}\quad \
M_2^{\dagger}  \overline{N}_2 M_3 \rho M_3^{\dagger} 
 = M_3\rho M_3^{\dagger}.
\label{eq:goal}
\end{aligned}
\end{equation}
The expressions above hold if and only if
\begin{align}
     M_2^{\dagger}  \overline{N}_2 [M_3\rho M_3^{\dagger}]_{\mathrm{supp}(M_2^\dagger \overline{N}_{2})} = [M_3\rho M_3^{\dagger}]_{\mathrm{supp}(M_2^\dagger \overline{N}_{2})}.
\end{align}
Here, we have reduced the state $M_3\rho M_3^{\dagger}$ to the geometric support of $M_{2}^{\dagger} \overline{N}_2$ and applied Lemma~\ref{lemma:LO_extension} below.
% Hence, our goal is to show that $M_2^{\dagger} \overline{N}_2$ is a strong symmetry unitary for $M_3 \rho M_3^{\dagger}$. 
On the other hand, applying the same equivalence to Eq.~\eqref{eq:N2N1-phase-fixing}, we get:
% One important implication of Eq.~\eqref{eq:N2N1-phase-fixing} is that $M_2^{\dagger} \overline{N}_2$ is a strong symmetry unitary for $M_1 \rho M_1^{\dagger}$: 
% \begin{align}
% M_2^{\dagger} \overline{N}_2 M_1 \rho M_1^{\dagger} = M_1 \rho  M_1^{\dagger}.
% \end{align} 
% Restricting this relation to the geometric support of $M_{2}^{\dagger} \overline{N}_2$, we obtain 
\begin{align}
M_2^{\dagger} \overline{N}_2 [M_1\rho M_1^\dagger]_{\mathrm{supp}(M_2^\dagger \overline{N}_{2})}
= [M_1\rho M_1^\dagger]_{\mathrm{supp}(M_2^\dagger \overline{N}_{2})}.
\end{align} 
Observing that 
\begin{align}
[M_1\rho M_1^\dagger]_{\mathrm{supp}(M_2^\dagger \overline{N}_{2})} = [M_3\rho M_3^\dagger]_{\mathrm{supp}(M_2^\dagger \overline{N}_{2})},
\end{align}
we have estabilished that $\overline{N}_2 M_3 \rho = M_2 M_3\rho$, hence Eq.~(\ref{eq:123}).
% Then, using lemma~\ref{lemma:LO_extension}, one can establish that $M_2^{\dagger} \overline{N}_2$ is also a strong symmetry unitary for $M_3 \rho M_3^{\dagger}$. 

Finally, using the identities in Eq.~\eqref{eq:321}~\eqref{eq:123} and the cyclicity of the trace, we have
\begin{align}
\Tr\Big[ \overline{N}_3^{\dagger} \overline{N}_2^{\dagger} \overline{N}_1^{\dagger}  \overline{N}_3 \overline{N}_2 \overline{N}_1 \rho\Big]  
=
\Tr\Big[ M_3^{\dagger} M_2^{\dagger} M_1^{\dagger}  M_3 M_2 M_1 \rho\Big] .
\end{align}
Here, we took the transpose conjugation of Eq.~\eqref{eq:123}. 
If we instead expand the $\overline{N}_j$ in terms of $\widetilde{N}_j$, we obtain
\begin{eqs}
\Tr\Big[ \overline{N}_3^{\dagger} \overline{N}_2^{\dagger} \overline{N}_1^{\dagger}  \overline{N}_3 \overline{N}_2 \overline{N}_1 \rho\Big]  
&=
\Tr\Big[  (\widetilde{N}_3^{\dagger} U_3^{\dagger} U_0^{\dagger}) (U_0\widetilde{N}_2^{\dagger} U_2^{\dagger}) (\widetilde{N}_1^{\dagger} U_1^{\dagger} U_0^{\dagger}) (U_3 U_0 \widetilde{N}_3)  (U_2 \widetilde{N}_2 U_0^{\dagger})(U_1 U_0 \widetilde{N}_1) \rho\Big]  
\\
&= 
\Tr\Big[ \widetilde{N}_3^{\dagger} \widetilde{N}_2^{\dagger} \widetilde{N}_1^{\dagger}  \widetilde{N}_3 \widetilde{N}_2 \widetilde{N}_1 \rho\Big].
\end{eqs}
Thus,
\begin{align}
\Tr\Big[ \widetilde{N}_3^{\dagger} \widetilde{N}_2^{\dagger} \widetilde{N}_1^{\dagger}  \widetilde{N}_3 \widetilde{N}_2 \widetilde{N}_1 \rho\Big] = 
\Tr\Big[ M_3^{\dagger} M_2^{\dagger} M_1^{\dagger} M_3 M_2 M_1 \rho\Big] = \omega.
\end{align}
\end{proof}

\subsection{LO generalization}\label{sec:LO_generalization}

The above proof of LU-chirality extends naturally to LO-chirality. 

We begin by showing that a strong symmetry unitary $Q_A$ of $\rho_{A}$ remains a strong symmetry for any of its extension state $\tau_{AR}$ (such that $\Tr_{R}(\tau_{AR}) = \rho_{A}$).

\begin{lemma}[Strong symmetry extension]\label{lemma:LO_extension}
Let a unitary operator $Q_A$ be a strong symmetry of $\rho_A$. Then, $Q_A \otimes I_{R}$ is a strong symmetry of $\tau_{AR}$:
\begin{align}
(Q_A \otimes I_{R})\tau_{AR} =\tau_{AR}.
\end{align}
\end{lemma}

\begin{proof}
Consider a canonical purification of $\tau_{AR}$ and Schmidt decompose it:
\begin{align}
|\sqrt{\tau}\rangle_{ARA'R'} = \sum_{|j\rangle \in \mathcal{H}(\rho_{A})} |j\rangle_{A} \otimes |\psi_j\rangle_{RA'R'}
\end{align}
where $\mathcal{H}(\rho_{A})$ denotes the support of $\rho_{A}$, $|j\rangle$ represents an orthonormal basis of this support, and the Schmidt coefficients are absorbed into the generally unnormalized states $|\psi_j\rangle$.
Recalling that $Q_A|j\rangle = |j\rangle$ for all $|j\rangle$ in $\mathcal{H}(\rho_{A})$, we find 
\begin{align}
(Q_A \otimes I_{R}\otimes I_{A'} \otimes I_{R'})|\sqrt{\tau}\rangle_{ARA'R'} = |\sqrt{\tau}\rangle_{ARA'R'}.
\end{align}
By tracing out $A'R'$, we obtain the desired result.
\end{proof}

Suppose that there exists a local operation achieving $\rho \rightarrow \rho^*$, namely
\begin{align}
U(\rho_A \otimes |0\rangle\langle 0|_R) U^{\dagger} = \sigma_{AR},\qquad 
\Tr_{R}(\sigma_{AR}) = \rho^*_{A}
\end{align}
where $U$ is a finite-depth local unitary and $R$ is an ancilla system initialized in the product state $|0\rangle\langle 0|_R$.  
(We assume that there are only $O(1)$ ancilla qubits per site.)

This Lemma implies that $U^{\dagger}(Q\otimes I)U$ is a strong symmetry of $\rho \otimes |0\rangle\langle 0|_R$, and thus $\rho \otimes |0\rangle\langle 0|_R$ would support string operators for both $\mathbb{Z}_3^{(1)}$ and $\mathbb{Z}_3^{(-1)}$ anyons. 
From this point, the remainder of the LU-chirality argument applies, completing the LO-chirality case. 

\subsection{Transversal Clifford for complex conjugation}
\label{sec3.6}
In the previous sections, we established that $\rho$ is LO-chiral if the anyon content is not mirror invariant, i.e. when $\mathcal{A}_d^{(k)}\not\simeq \mathcal{A}_d^{(-k)}$.
Here, we prove the converse:  $\rho$ is LO-non-chiral whenever the anyon content is mirror invariant, $\mathcal{A}_d^{(k)}\simeq \mathcal{A}_d^{(-k)}$.

\begin{lemma}[LO-non-chirality]
Consider the mixed state stabilizer model introduced in Definition~\ref{def:honeycomb-mixed} whose anyon content is mirror invariant $\mathcal{A}_d^{(k)}\simeq \mathcal{A}_d^{(-k)}$ with coprime $(d,k)$. 
Then there exists a depth-$1$ (transversal) Clifford transformation which transforms $\rho \rightarrow \rho^{*}$ for the $\mathbb{Z}_d^{(k)}$ stabilizer mixed states.
\end{lemma}

% \begin{proof} It suffices to focus on the $d=p^m$ cases with $p=1$ modulo $4$. From Lemma \ref{lemma:residue}, there exists $\lambda$ such that $\lambda^2 = -1$ modulo $d$. Thus, we can construct a single-qudit Clifford $U_j$ acting on a $j$th qudit by
% %For a qudit with prime Hilbert space dimension $p$, there exists a Clifford unitary $U$ satisfying the following:
% \begin{align}
% X_j \rightarrow X^\lambda_j \qquad Z_j \rightarrow Z^{-\lambda}_j \ ,
% \end{align}
% as the map preserves the commutation relation. Since the Clifford action maps, 
% \begin{align}
% X_jZ_j^k \rightarrow X_j^\lambda Z_j^{-k\lambda} = \omega^{k\lambda(\lambda-1)/2}(X_jZ_j^{-k})^\lambda \qquad \ X_jZ_j^{-k} \rightarrow X_j^\lambda Z_j^{k\lambda} = \omega^{-k\lambda(\lambda-1)/2}(X_jZ_j^k)^\lambda  
% \end{align}
% %Let $\lambda$ be an integer satisfying $\lambda^2 = -1$ modulo $p$. We then consider $U$ satisfying 
% %\begin{align}
% %X \rightarrow X^\lambda, \qquad Z \rightarrow Z^{-\lambda} \qquad XZ \rightarrow X^{\lambda}Z^{-\lambda}.
% %\end{align}
% the transversal Clifford $\overline{U} = \otimes_j U_j$ transforms $\rho$ to $\rho^*$.
% \end{proof}

\begin{proof} 
% It suffices to focus on the $d=p^m$ cases with $p=1$ modulo $4$. 
Define the map $U_j$ by
\begin{align}\label{eq:tranversalU}
X_j \rightarrow X^\lambda_j, \qquad Z_j \rightarrow Z^{-\lambda}_j \ .
\end{align}
The map is Clifford (i.e., invertible and preserves the commutation relations) if and only if $\lambda^2 = -1$ mod $d$. Such $\lambda$ always exists by Lemma \ref{lemma:residue}. 

We claim that the the transversal Clifford operator $\overline{U} \equiv \otimes_j U_j$ transforms stabilizers of the honeycomb model in Eq~\eqref{eq:stabilizer-generators-mixed_state} to stabilizers of $\rho^*$.
First, due to Eq.~\eqref{eq:tranversalU}, we have:
\begin{equation}
    S_{f_A} \rightarrow (S_{f_A}^*)^\lambda,\qquad
    S_{f_B} \rightarrow (S_{f_B}^*)^\lambda,
\end{equation}
so it remains to check $S_{f_C}$.
For this, we note that:
\begin{align}
X_jZ_j^k \rightarrow X_j^\lambda Z_j^{-k\lambda} = \omega^{k\lambda(\lambda-1)/2}(X_jZ_j^{-k})^\lambda, \qquad \ X_jZ_j^{-k} \rightarrow X_j^\lambda Z_j^{k\lambda} = \omega^{-k\lambda(\lambda-1)/2}(X_jZ_j^k)^\lambda.  
\end{align}
While there are extra factors, they cancel each other in $S_{f_C}$, leaving $S_{f_C} \rightarrow (S_{f_C}^*)^\lambda$.
This concludes the proof.
\end{proof}

\section{Four-partite chirality}
\label{sec4:four=partite-proof}

We have established that a stabilizer mixed state $\rho$ is LO-chiral if and only if its anyon content is not mirror invariant.
In this section, we extend the notion of LO-chirality to \emph{$n$-partite chirality}~\cite{vardhan2025chiralitymagicquantumcorrelations}. 

\begin{definition}
A mixed state $\rho$ is called \emph{$n$-partite chiral} under an $n$-partition into $A_1 A_2 \cdots  A_n$ if there does not exist a local quantum channel  
\begin{align}
\mathcal{Q} = \bigotimes_{j=1}^n \mathcal{Q}_{A_j}
\end{align}
such that $\mathcal{Q}(\rho)=\rho^*$.
\end{definition}

In this formulation, the quantum channel is required to act in a strictly transversal, factorized manner over the tensor product Hilbert space $\bigotimes_j \mathcal{H}_{A_j}$. 
At the same time, within each subset $A_j$ (which may contain multiple qubits), one allows completely general quantum channels, including those that may be highly complex or fine-tuned.

Conceptually, the question of $n$-partite chirality is more than an extension of the LO-chirality result. 
It directly probes a central goal of the broader research program, to characterize the static and dynamical properties of many-body systems using only the entanglement structure of their ground states. 
Concretely, it asks whether chirality, understood as an obstruction to complex conjugation, can be detected purely through $n$-partite entanglement. 
A previously proposed candidate is the modular commutator~\cite{PhysRevLett.128.176402}, which constitutes a three-partite entanglement measure. 
However, this quantity can be reproduced by trivial short-range entangled states, indicating that it fails to serve as a faithful diagnostic of chirality~\cite{Gass_2024}.
Moreover, for stabilizer states, the modular commutator vanishes identically.
These observations motivate a more general question of whether a genuine entanglement-based measure of chirality exists, and if so, what minimal degree of multipartiteness $n$ is required. 

In this section, we will determine whether the $\mathbb{Z}_{d}^{(k)}$ mixed state exhibits four-partite chirality. 
To this end, we consider a partition of the $\mathbb{Z}_d^{(k)}$ mixed state $\rho$ into five subsystems $A,B,C,D,E$, as depicted in Fig.~\ref{fig_5parties}, where $ABCD$ forms an open disk and $E$ denotes the complementary surrounding region.
In fact, the surrounding region $E$ may be empty in the subsequent discussions.
Although our arguments extend beyond specific lattice realizations, we focus here on the hexagonal lattice construction for concreteness.\footnote{
For instance, one may consider the mixed state $\rho_{ABCD}$ on a closed sphere with $E=\emptyset$. 
While a hexagonal lattice is not compatible with spherical topology, the construction of the $\mathbb{Z}_d^{(k)}$ mixed state can be generalized to any two-dimensional three-colorable lattice, following Ref.~\cite{lee2025chiralcolorcode}.
} 

In this section, we will establish that $\rho_{ABCD}$ is four-partite chiral if and only if its anyon content is not mirror invariant.

\begin{figure}
\centering
\raisebox{\height}{\hspace{5pt}}\raisebox{-0.85\height}{\includegraphics[width=0.3\textwidth]{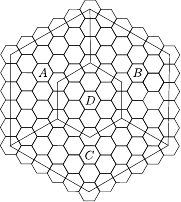}}
\hspace{10pt}
\caption{
Four-partite subsystems $ABCD$. The surrounding region $E$ may be empty. 
}
\label{fig_5parties}
\end{figure}

\begin{theorem}[Four-partite chiral]\label{thm:four-partite}
The $\mathbb{Z}_d^{(k)}$ mixed state $\rho_{ABCD}$ is four-partite LO-chiral if and only if $\mathcal{A}_d^{(k)}\not\simeq \mathcal{A}_d^{(-k)}$. 
\end{theorem}

Throughout this section, we will focus on the cases with odd $d$, as the extension to even $d$ cases is straightforward. 

It is important to emphasize that Theorem~\ref{thm:four-partite} applies only to the \emph{mixed-state} realization of the $\mathbb{Z}_{d}^{(k)}$ anyon theory. 
In particular, our proof does not currently extend to the \emph{pure-state} realization of the $\mathbb{Z}_{p^2}^{(1)}$ anyon theory (Definition \ref{def:Zp1pure}), although it does cover the mixed-state version of the same theory.
This stands in contrast to the LO-chirality result, which applies to both mixed and pure states. 
At present, it remains an open question whether the pure state realization $\rho_{ABCDE} = |\psi_{ABCDE}\rangle\langle \psi_{ABCDE}|$ for $\mathbb{Z}_{p^2}^{(1)}$ leads to a four-partite chiral mixed state. 
This is essentially due to the extra edge strong symmetry operators of order $p$. 
At the end of the section, we briefly comment on the technical obstruction and a possible route to overcoming it.

\subsubsection*{Proof sketch}

Below, we sketch the proof by focusing on local unitaries, as the extension to local channels follows directly from Lemma~\ref{lemma:LO_extension}. 
The argument proceeds analogously to the LO-chirality proof, but with an important complication due to the fact that there is no natural notion of micro-locality in the $n$-partite setting.
In particular, we find that one cannot directly define anyon charges as being enclosed within spatial regions, due to the subtleties associated with boundaries and tri-junctions between subsystems (see Fig.~\ref{fig_5parties}). 
The key idea to circumvent this issue is to introduce the notion of \emph{edge charges}. 
Namely, instead of associating charges with local spatial regions, we assign them to the total charge shared across the edges (boundaries) of neighboring subsystems. 
This enables us to consistently redistribute the charges associated with tri-junctions into edge charges. 
We then establish the conservation of edge charges, as well as their factorization properties. 
The topological spin can be extracted from the T-junction process as depicted in Fig~\ref{fig_4_partite} where each string operator originates from an edge separating neighboring subsystems.

\subsection{Canonical purification, weak symmetry, and string operators}\label{sec:4-partite_additional}

As in the analysis of LO-chirality, it is convenient to work with the canonical purification. 
A key difference here is that we consider the reduced density matrix on $ABCD$
\begin{align}
\rho_{ABCD} = \Tr_{E}(\rho_{ABCDE}) 
\end{align}
and hence must construct the canonical purification of the local state $\rho_{ABCD}$, rather than that of the global state $\rho_{ABCDE}$ (unless $E$ is taken to be empty).
This leads to a slight modification in the structure of strong and weak symmetry operators in the canonical purification (equivalently, stabilizer generators and logical operators for the stabilizer code $\mathcal{S}_{ABCD}$).
Although this does not affect the overall logic of the argument, we analyze these modifications explicitly.

The stabilizer group $\mathcal{S}_{ABCD}$ for $\rho_{ABCD}$ is generated solely by face stabilizers $S_f$ that are fully supported within the region $ABCD$. 
Two-body edge operators supported on $ABCD$ (as defined in Eq.~\eqref{eq:edge_operators}), which are supported entirely on $ABCD$, remain logical operators since they commute with all elements of $\mathcal{S}_{ABCD}$. 
However, tracing out subsystem $E$ removes all stabilizers that had support across the $ABCD|E$ cut. 
As a consequence, additional logical operators emerge, which should be regarded as weak symmetry operators for $\rho_{ABCD}$.

We now systematically enumerate all such additional weak symmetry (logical) operators, which can be classified as follows:

\begin{enumerate}[(i)]
\item \textbf{One-body logical operators} \\
At certain boundary qudits, single-site Pauli operators commute with all stabilizers in $\mathcal{S}_{ABCD}$ and thus also serve as logical operators:
\begin{align}
\figbox{3.0}{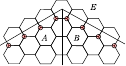}\ .
\end{align}
These can be interpreted as truncations of the two-body logical operators on the edges.
Note that the boundary sites are acted on by only a single face stabilizer. 
% Such qudits thus support one-body logical operators localized at the boundary.  

\item \textbf{Truncated face stabilizers on edges} \\
Face stabilizers $S_f$, supported across $ABCD$ and $E$, become truncated when restricted to $ABCD$:\footnote{
In the canonical purification picture, the truncated face stabilizer comes from an $S_f \otimes I$ term that intersected the $ABCD |E$ cut.
}
\begin{align}
\figbox{3.0}{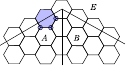}\ .
\end{align}
By construction, these truncated operators commute with all stabilizers in $\mathcal{S}_{ABCD}$, and therefore act as logical operators. However, they can be generated by multiplying one-body logical operators from (i) and two-body edge operators from Eq.~\eqref{eq:edge_operators}. They therefore \emph{do not} introduce additional independent logical degrees of freedom.
%Moreover, neighboring truncated operators fail to commute with each other, confirming they generate non-trivial logical degrees of freedom. 

\item \textbf{Truncated corner stabilizer} \\
Near the tri-junction corners of the $ABCD$ region, one finds truncated stabilizer of mixed support:
\begin{align}
\figbox{3.0}{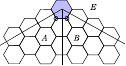}\ .
\end{align}
These corner operators can be generated by multiplying one-body logical operators from (i) and two-body edge operators from Eq.~\eqref{eq:edge_operators}.
They therefore \emph{do not} introduce additional independent logical degrees of freedom.  
\end{enumerate}

Thus, boundary operators in (i) and two-body edge operators from Eq.~\eqref{eq:edge_operators} generate all the additional weak symmetry (logical) operators for $\rho_{ABCD}$. In summary, we have
%Importantly, these additional boundary operators can be realized locally within each subsystem $A$, $B$, or $C$, as the corner operators from iii) can be generated from two-body edge operators and single-body logical operators from ii). 

\begin{observation}\label{obs:additional}
Logical operators of $\mathcal{S}_{ABCD}$ can be generated by two-body edge operators from Eq.~\eqref{eq:edge_operators}, as well as the additional single-site boundary operators in (i).
% that are locally supported within each subsystem $A$, $B$, or $C$. 
\end{observation}

The fact that additional logical operators can be generated locally within subsystems $A$, $B$, and $C$ allows us to proceed analogously to the LO-chirality proof.
For notational simplicity, in what follows, we will write 
\begin{align}
\rho \equiv \rho_{ABCD}. 
\end{align}
This convention will be used throughout the remaining discussion.

\begin{figure}
\centering
\raisebox{\height}{\hspace{5pt}}\raisebox{-0.85\height}{\includegraphics[width=0.3\textwidth]{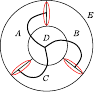}}
\hspace{10pt}
\caption{
A T-junction process in four-partite ($ABCD$) chiral mixed state. 
}
\label{fig_4_partite}
\end{figure}

\subsubsection*{Edge-local string operators}

As alluded to earlier, our proof proceeds by characterizing anyonic excitations associated with edges. 
To gain insight, let us study a string operator $N_A$ for $\rho^*$ that starts on the $AB$ edge and ends on the $AD$ edge, as shown in Fig.~\ref{fig_4_partite_string_a}.
This operator $N_A$ can be obtained by truncating a loop stabilizer $N$ that encircles the $A,D,B$ subsystems, around their tri-junction.
Since $N_A$ is a segment of a strong symmetry loop stabilizer for $\rho^*$ that does not cross the $AC$ edge, it commutes (when acting on $\rho^*$) with all strong and weak symmetry operators $S_f$ and $g_e$ supported entirely on $A$ as well as those across the $AC$ edge.\footnote{In the canonical purification picture $|\Psi_{\rho^*}\rangle$, this means that $N_A \otimes I$ commutes with $S_f\otimes I$ and $g_e\otimes g_e^*$ supported on $AC \otimes A'C'$. 
Consequently, $N_{A}\otimes I $ can only violate strong or weak symmetry operators located on the $AD$ and $AB$ edges in $|\Psi_{\rho^*}\rangle$. 

}

This motivates the following notion of edge-local string operators. 

\begin{definition}[Edge-local string operator]
An operator $N_{A}$ is called a string-like operator connecting edges $AB$ and $AD$ if it commutes (when acting on the state) with all strong and weak symmetry operators except those supported on $AB$ and $AD$. 
\label{def:Edge-local string}
\end{definition}

Now, for contradiction, suppose that there exists a unitary of the form
\begin{align}
U = U_A \otimes U_B \otimes U_C \otimes U_D \quad \text{such that} \quad U\rho U^{\dagger} = \rho^* .
\end{align}
We emphasize that each of $U_A, U_B, U_C, U_D$ is not necessarily a finite-depth unitary. 
For a loop stabilizer $N$ encircling the $A,D,B$ subsystems around their tri-junction, define 
\begin{align}
\widetilde{N} \equiv U^{\dagger} N U = \widetilde{N}_A \widetilde{N}_D \widetilde{N}_B, \qquad \widetilde{N}_R \equiv   U^{\dagger}_{R} N_R U_{R} 
\label{eq:unitary-on-regions}
\end{align}
where $R\in \{A,D,B\}$ as shown in Fig.~\ref{fig_4_partite_string}.
Here, the conjugated operator $\widetilde{N}_A$ may be spread over the whole region $A$ as graphically emphasized in Fig.~\ref{fig_4_partite_string_b}. 
Yet, one can still deduce key properties of $\widetilde{N}_A$ and conclude that $\widetilde{N}_A$ essentially acts like a string operator for $\rho$. 
Since $N$ is a strong symmetry operator for $\rho^*$, $\widetilde{N}$ is a strong symmetry operator for $\rho$. 
In particular, $\widetilde{N}$ commutes with all weak symmetry operators when acting on $\rho$.
It is worth emphasizing again that weak symmetry operators here are defined with respect to the reduced mixed state $\rho_{ABCD}$, rather than the global mixed state $\rho_{ABCDE}$. 
In particular, they include the additional boundary logical operators identified in the previous subsection. 

Since $\widetilde{N}=\widetilde{N}_A\widetilde{N}_D\widetilde{N}_B$, the individual operator $\widetilde{N}_A$ can create violations only on specific edges. 
Observe that, in the canonical purification, $\widetilde{N}_A\otimes I$ commutes with the stabilizers $g\otimes g^*$ corresponding to the additional boundary weak symmetry identified in Observation~\ref{obs:additional}.
This suggests that $\widetilde{N}_A\otimes I$ can violate only strong symmetry stabilizers and weak symmetry stabilizers of the form $g\otimes g^*$ near the edges $AD$ and $AB$, analogous to the string operator $N_A$ for $\rho^*$.\footnote{
Note that $\widetilde{N}_A\otimes I$ may violate the stabilizer at the tri-junction of $A$, $B$, and $D$. We return to this issue in the next subsection.
}

Our discussion can be summarized as follows. 

\begin{observation} \label{obs:edgelocal}
Given a string operator $N_A$ for $\rho^*$ connecting $AB$ and $AD$ edges, its conjugated operator $\widetilde{N}_A = U_A^{\dagger} N_A U_{A}$ also acts effectively as a string operator connecting $AB$ and $AD$ edges for $\rho$. 
In particular, $\widetilde{N}_{A}$ can only violate strong or weak symmetry operators located on the $AD$ and $AB$ edges. 
\end{observation}

% \begin{figure}
% \centering
% \raisebox{\height}{\hspace{5pt}}\raisebox{-0.85\height}{\includegraphics[width=0.45\textwidth]{fig_4_partite_string}}
% \hspace{10pt}
% \caption{
% A string operator $N_A$ for $\rho^*$, and the corresponding operator $\widetilde{N}_A$ for $\rho$. Although $\widetilde{N}_A$ may be supported globally on $A$ since there is no micro-locality constraint on $U_{A}$, it creates excitations only on $AB$ and $AD$ edges.\YP{Coming back to this.}}
% \label{fig_4_partite_string}
% \end{figure} 

\begin{figure}
\centering
\begin{subfigure}{0.22\textwidth}
    \centering
    \includegraphics[width=\textwidth]{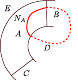}
    \caption{}
    \label{fig_4_partite_string_a}
\end{subfigure}
\hspace{10pt}
\begin{subfigure}{0.22\textwidth}
    \centering
    \includegraphics[width=\textwidth]{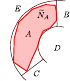}
    \caption{}
    \label{fig_4_partite_string_b}
\end{subfigure}
\caption{
A string operator $N_A$ for $\rho^*$, and the corresponding operator $\widetilde{N}_A$ for $\rho$. 
Although $\widetilde{N}_A$ may be supported globally on $A$ since there is no micro-locality constraint on $U_A$, it creates excitations only on the $AB$ and $AD$ edges.}
\label{fig_4_partite_string}
\end{figure}

\subsection{Edge charges}

We now introduce the notion of \emph{edge charges}. 
A natural first attempt is to define an edge charge by multiplying face stabilizers along a given edge, for example along the $AB$ edge:
\begin{align}
W_{AB}^{\text{naive}} &= \prod_f S_f = \figbox{3.5}{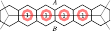}
\end{align}
where the ``$1$'' labels in the diagram indicate that each $S_f$ enters with exponent $1$. 

However, the definition of $W_{AB}^{\text{naive}}$ is problematic, since it does not include stabilizers located at the tri-junctions:
\begin{align}
\figbox{3.5}{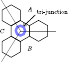} \ .
\end{align}
Namely, under the action of a local unitary $U_A$, such tri-junction stabilizers may also be violated. 
As a result, one cannot distinguish a charge localized on the $AB$ edge from contributions near the tri-junction.
Hence, $W_{AB}^{\text{naive}}$ alone is insufficient to establish an analogue of Lemma~\ref{lemma:anyon_charge}.

%\newtext{Thus, one would have to choose how to assign such tri-junction-supported operators to the neighboring edges. Such a choice would not guarantee that the charge is conserved independently on each edge, as is needed to establish an analogue of Lemma~\ref{lemma:anyon_charge}.} As a result, \newtext{one cannot distinguish a charge localized on the $AB$ edge from contributions near the tri-junction,}
%preventing us from establishing an analogue of Lemma~\ref{lemma:anyon_charge} using $W_{AB}^{\text{naive}}$ alone. 
%\TE{I do not know what is meant by mixing here. Maybe alternatively, we can say: However, this definition is problematic, since the tri-junction stabilizer is unaccounted for. We could pick an edge of the partition and associate it to that edge, but then we would have no guarantee that the charge is conserved at each edge, as is needed to establish an analogue of Lemma~\ref{lemma:anyon_charge}.}\YP{I agree and thanks for the suggestion. I changed it, what do you think?}

\subsubsection*{Decorated edge charges}

To avoid this issue, we introduce a \emph{decorated edge charge} operator.
From this point onward, we work in the canonical purification picture. We define
\begin{align}
W_{AB} \equiv S_e^{-1} S_{e'}^{-1} \prod_f S_f \otimes I = \figbox{3.5}{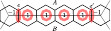}\ 
\end{align}
where the ``$-1$'' labels indicate the inverse powers of the weak symmetry stabilizers $S_e = g_e \otimes g_e^*$. 
In other words, 
\begin{align}
W_{AB} = S_e^{-1} S_{e'}^{-1}(W_{AB}^{\text{naive}} \otimes I). 
\end{align}
In the canonical purification picture, $W_{AB}$ acts on \emph{both layers}, since $S_e^{-1}$ and $S_{e'}^{-1}$ involve $g_e \otimes g_e^*$ terms. 
For all bipartite edges ($AB,AC,AD,BC,BD,CD$), we can define edge charges with decorations by weak symmetry stabilizers $S_e = g_e \otimes g_e^*$ that cross the corresponding edges, see Fig.~\ref{fig_4_partite_edge_charge} for illustration. We now consider excitations created by $\widetilde{N}_A $, a ``string-like'' operator connecting $AB$ and $AD$ edges. 
% {\newtext{The decoration removes the ambiguity of $W_{AB}^{\text{naive}}$ by making the edge-charge operator deformable near the adjacent tri-junctions, in the sense made precise in the lemma below.}
% The decoration removes the ambiguity of $W_{AB}^{\text{naive}}$ by making the edge-charge operator deformable near the adjacent tri-junctions. We make this statement precise in Lemma below.
% The decoration removes the ambiguity of $W_{AB}^{\text{naive}}$ by making the edge-charge operator deformable near the adjacent tri-junctions, in the sense that its local representative near a junction can be changed by multiplying stabilizers without changing its action on the relevant excited state. We make this statement precise in the lemma below.

\begin{figure}
\centering
\raisebox{\height}{\hspace{5pt}}\raisebox{-0.85\height}{\includegraphics[width=0.35\textwidth]{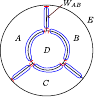}}
\hspace{10pt}
\caption{
Decorated construction of edge charges $W_{AB},W_{AC},\cdots$. 
}
\label{fig_4_partite_edge_charge}
\end{figure}

\begin{lemma}\label{lemma:edge_junction}
Given a string-like operator $\widetilde{N}_{A}$ for $\rho$ connecting $AB$ and $AD$ edges, in the canonical purification, we have 
% \begin{align}
% \Tr\big[ (S_{e_{AB}} S_{e_{AD}} S_{f}) \sigma \big] =1,
% \end{align}
% where $ \sigma = \widetilde{{N}}_A \rho \widetilde{{N}}_A^{\dagger}$, $S_{f}$ denotes the tri-junction stabilizer, and $S_{e_{AB}}$ and $S_{e_{AD}}$ represent the weak symmetry stabilizers crossing the $AB$ and $AD$ edges adjacent to the tri-junction, respectively. \TE{Sorry, what is meant by $S_{e_{AB}}$ and $S_{e_{AD}}$ here? These are defined for the canonical purification. All we need is the statement below in terms of the canonical purification (with $S_f$ replaced by $S_f \otimes I$ for consistency) -- or am I missing something? Since the section begins by saying that we will be working in the purification, we might as well state the lemma in terms of the purification anyways.}
% \ZL{I agree with the comment. I don't see why two equations are equivalent. }\YP{I agree, I removed Eq.136 and wrote the statement of the Lemma in canonical purification as Tyler suggested (see colored text).}
% Equivalently, in the canonical purification, we have 
\begin{align}
\langle S_{e_{AB}} S_{e_{AD}} (S_{f} \otimes I) \rangle =1 \quad \text{for} \quad  \mbox{$(\widetilde{N}_A \otimes I)|\Psi_{\rho}\rangle$ }.
\end{align}
where $S_{e_{AB}}$ and $S_{e_{AD}}$ represent the weak symmetry stabilizers crossing the $AB$ and $AD$ edges adjacent to the tri-junction, respectively.
\end{lemma}
% \newtext{Let $\widetilde{N}_{A}$ be a string-like operator for $\rho$ connecting the $AB$ and $AD$ edges, and define
% \[
% |\Phi_A\rangle \equiv (\widetilde{N}_{A}\otimes I)|\Psi_{\rho}\rangle .
% \]
% Then
% \begin{align}
% \langle \Phi_A|
% S_{e_{AB}} S_{e_{AD}} (S_{f}\otimes I)
% |\Phi_A\rangle =1 .
% \end{align}
% Here, $S_{f}\otimes I$ denotes the first-layer tri-junction stabilizer, while
% \[
% S_{e_{AB}}=g_{e_{AB}}\otimes g_{e_{AB}}^*,
% \qquad
% S_{e_{AD}}=g_{e_{AD}}\otimes g_{e_{AD}}^*
% \]
% are the weak-symmetry stabilizers in the canonical purification crossing the $AB$ and $AD$ edges adjacent to the tri-junction, respectively.}

It is useful to graphically depict the statement of this lemma:
\begin{align}
\figbox{3.5}{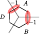}  \simeq \figbox{3.5}{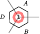}  \label{eq:tri-junction_decomposition}
\end{align}
where two operators have identical action on the excited state $(\widetilde{N}_A\otimes I) |\Psi_{\rho}\rangle$, with $\simeq$ denoting the equivalence on this state. 
This lemma essentially says that the tri-junction charge $S_f \otimes I$ can be split into two edge operators $S_{e_{AB}}^{-1}$ and $S_{e_{AD}}^{-1}$ without any leakage, when acting on the purified state $(\widetilde{N}_A\otimes I) |\Psi_{\rho}\rangle$. 
% \newtext{excited purified state $|\Phi_A\rangle$, with $\simeq$ denoting equivalence on this state. Equivalently,
% \[
% (S_f\otimes I)|\Phi_A\rangle
% =
% S_{e_{AD}}^{-1}S_{e_{AB}}^{-1}|\Phi_A\rangle .
% \]
% Thus, the tri-junction charge $S_f\otimes I$ can be absorbed into the two adjacent edge operators $S_{e_{AB}}^{-1}$ and $S_{e_{AD}}^{-1}$ without any leakage.}

% excited state $(\widetilde{N}_A\otimes I) |\Psi_{\rho}\rangle$, with $\simeq$ denoting the equivalence on this state. 
% This lemma essentially says that the tri-junction charge $S_f \otimes I$ can be split into two edge operators $S_{e_{AB}}^{-1}$ and $S_{e_{AD}}^{-1}$ without any leakage, when acting on the purified state $(\widetilde{N}_A\otimes I) |\Psi_{\rho}\rangle$. 

\begin{proof}
Examining the tri-junction hexagon gives the following redundancy condition for stabilizer generators:
\begin{align}
\prod_{e \in f} S_e^{-1} = \figbox{3.5}{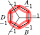} = S_{f}\otimes S_f^*. \label{eq:local_redundancy}
\end{align}
In other words, a product of eight stabilizer generators surrounding a hexagonal prism is an identity:
\begin{align}
\figbox{3.5}{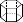} \  = I. \label{eq:local_redundancy_figure}
\end{align}
Since $\widetilde{N}_A\otimes I$ is a string-like operator connecting $AB$ and $AD$ edges, it commutes with all strong and weak symmetry operators supported on $A$.
Namely, it commutes with the operator
\begin{align}
\figbox{3.5}{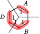}.
\end{align}
Also, since $\widetilde{N}_A\otimes I$ is a single-layer operator, it commutes with $(I\otimes S_f)$. 
Hence, it follows that 
two operators
\begin{align}
\figbox{3.5}{fig_tri_junction_two_edge.pdf}  \simeq \figbox{3.5}{fig_tri_junction_face.pdf} 
\end{align}
have identical action. 
\end{proof}

\subsubsection*{Conservation of edge charges}

We now claim that the decorated edge charges satisfy a conservation law analogous to the anyon charge conservation.

\begin{lemma}
Given a string-like operator $\widetilde{N}_{A}$ for $\rho$ connecting $AB$ and $AD$ edges, we have 
\begin{align}
\langle W_{AD}W_{AB} \otimes I \rangle =1 \quad \text{for} \quad \mbox{$(\widetilde{N}_A \otimes I)|\Psi_{\rho}\rangle$ }.\label{eq:partite_conservation}
\end{align}
\end{lemma}

This expresses the conservation of total edge charge, analogous to the LO-chirality case.

\begin{proof}
This lemma can be verified analogously to Lemma~\ref{lemma:anyon_charge} where we deform $W_{AD}W_{AB}\otimes I$ by multiplying strong and weak stabilizer generators that commute with $\widetilde{N}_A \otimes I$. 
First, since $W_{AC}$ commutes with $\widetilde{N}_A \otimes I$ when acting on $|\Psi_{\rho}\rangle$, we can perform the following deformation:
\begin{align}
\figbox{3.5}{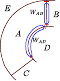} \ \simeq \ \figbox{3.5}{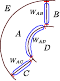} \ . 
\end{align}
Here, the operators shown in blue act only on the first layer while the operators shown in red act on both layers.

Next, the corner decorations can be deformed into the tri-junction stabilizers using Eq.~\eqref{eq:tri-junction_decomposition}, giving 
\begin{align}
\figbox{3.5}{fig_4_partite_deformation2.pdf} \ \simeq \ \figbox{3.5}{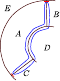} \ .
\end{align}

Since $\widetilde{N}_A$ commutes with all the face stabilizer operators on $A$, multiplying them yields:
\begin{align}
\figbox{3.5}{fig_4_partite_deformation3.pdf} \ \simeq \ \figbox{3.5}{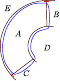} \ , 
\end{align}
producing a closed loop stabilizer decorated at the tri-junction corners ($ABE$ and $ACE$).

This operator can be further deformed into the second layer using the boundary weak symmetry stabilizers: 
\begin{align}
\figbox{3.5}{fig_4_partite_deformation4.pdf} \ \simeq \ \figbox{3.5}{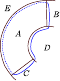} \ .
\end{align}
Here the dotted blue operators now act in the second layer while the solid blue operators still act on the first layer. 
Such deformation is possible because the boundary logical operators are truncated face stabilizers along the $AE$ boundary, as discussed in Sec.~\ref{sec:4-partite_additional}, namely those identified in (ii). 
Deformations near at the tri-junction corners (those involving the weak symmetry operators shown in red) can be performed by using the one-body boundary logical operators identified in (i). 

The resulting loop operator no longer overlaps with $\widetilde{N}_A\otimes I$, and thus Eq.~\eqref{eq:partite_conservation} follows. 
\end{proof}

\subsection{Edge charges are fixed}

We have established that edge charges $W_{AB},W_{AC},\cdots$ are indeed conserved when we focus on excitations created by string-like operators. 
We now establish that edge charges $W_{AB},W_{AC},\cdots$ are fixed, meaning that the excited states must be in eigenstates of edge charge operators. 

\begin{lemma}\label{lemma:charge_factorization_partite}
Consider a string-like operator $\widetilde{N}_{A}$ for $\rho$ connecting $AB$ and $AD$ edges, which is a part of a strong symmetry operator $\widetilde{N}$. 
We then have 
\begin{align}
 \langle W_{AB}\otimes I\rangle \langle W_{AD}\otimes I\rangle = 1 \quad \text{for} \quad \mbox{$(\widetilde{N}_A \otimes I)|\Psi_{\rho}\rangle$ }.
\end{align}
\end{lemma}

Consider transporting edge charges along a loop-like process decomposed into four string operators (Fig.~\ref{fig_4_partite_LR}):
\begin{align}
W_{AD} 
\underset{\widetilde{N}_A}{\longrightarrow} W_{AB}  
\underset{\widetilde{N}_B}{\longrightarrow} W_{BC}  
\underset{\widetilde{N}_C}{\longrightarrow} W_{CD}  
\underset{\widetilde{N}_D}{\longrightarrow} W_{AD}
\end{align}
Since $\widetilde{N}=\widetilde{N}_A\widetilde{N}_B\widetilde{N}_C\widetilde{N}_D$ is a strong symmetry, $\widetilde{N}_A$ and $\widetilde{N}_D^{\dagger}\widetilde{N}_C^{\dagger}\widetilde{N}_B^{\dagger}$ must generate the same state, creating excitations only on $W_{AD}$ and $W_{AB}$.

\begin{figure}
\centering
\raisebox{\height}{\hspace{5pt}}\raisebox{-0.85\height}{\includegraphics[width=0.35\textwidth]{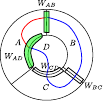}}
\hspace{10pt}
\caption{Edge charges $W_{AB}, W_{AD}$ can be created by $\widetilde{N_A}\otimes I$, as well as by $\widetilde{N}_B^{\dagger}\widetilde{N}_C^{\dagger}\widetilde{N}_D^{\dagger}\otimes I$.
}
\label{fig_4_partite_LR}
\end{figure}

One might attempt to argue that $\langle W_{AD} W_{AB} \rangle$ should factorize as $\langle W_{AD}\rangle \langle W_{AB} \rangle$ from a naive lightcone argument. 
The intuition is that each string operator only couples neighboring edge charges, so a local unitary of the form $\widetilde{N}_B \otimes \widetilde{N}_C \otimes \widetilde{N}_D$ should not generate correlations between $W_{AD}$ and $W_{AB}$.
However, this reasoning fails because the subsystems supporting neighboring edge charges are not fully decoupled. 
Specifically, if $R_{AD}$ and $R_{AB}$ are the regions supporting $W_{AD}$ and $W_{AB}$ (including second-layer qudits near the tri-junction), then a weak symmetry stabilizer $S_e = g_e\otimes g_e^*$ acts jointly on $R_{AD}$ and $R_{AB}$. This introduces correlations such that 
\begin{align}
\rho_{R_{AD}R_{AB}}\not=\rho_{R_{AD}}\otimes \rho_{R_{AB}}.
\end{align}

A hint for how this issue can be resolved comes from the following physical intuition. 
Although $S_e$ creates correlations between $R_{AD}$ and $R_{AB}$, these correlations are never ``utilized''.
Since $S_e$ is supported locally on $A\otimes A'$, it is naturally regarded as a local stabilizer constraint rather than an independent edge-charge degree of freedom. 
% Indeed, $\widetilde{N}_A \otimes I$ commutes with $S_e$, since the latter is supported locally on $A\otimes A'$.
Thus, while $S_e$ prevents strict factorization of the reduced density matrices, it does not obstruct factorization of the expectation values of edge-charge operators in the excited state created by $\widetilde{N}_A \otimes I$.
 
To formalize this intuition, we embed the excitation Hilbert space isometrically into a Hilbert space where the underlying state is initially factorized. 
Define the stabilizer code %\DL{I think we may need to include the boundary logical operators in $\mathcal{S}_{\text{code}}$ or mention $A$, $B$, $C$, $D$ are partitioning closed sphere?} \TE{We are back to working with the canonical purification. Those single-site logical operators don't all appear here, because they can fail to commute with the stabilizers $g_e \otimes g_e^*$.} \TE{We do need to include some truncations of the stabilizers of the purification, though. These are missing in the equation below. } \DL{I think we add them in a form of $g_e \otimes g_e^*$. In this form, all of boundary weak symmetry stabilizer commute (single-site logicals too).}
\begin{equation}
\mathcal{S}_{\text{code}} \equiv \Big\langle \ \text{$S_f \otimes I$ on $A,B,C,D$},  
\ \text{$g_e \otimes g_e^*$ on $AA',BB',CC',DD'$},
\ \text{all of $I \otimes S_f^*$} \ \Big\rangle. 
\label{eq:code}
\end{equation}
Namely, we only keep $S_f\otimes I$ that are supported locally in each subsystem $A$, $B$, $C$, or $D$ (and $g_e \otimes g_e^*$ in $AA'$, $BB'$, $CC'$, or $DD'$). 
Here, the generators of the form $g_e \otimes g_e^*$ include stabilizers associated to the boundary weak symmetry operators of the mixed state, as they are locally supported on each subsystem $A$, $B$, $C$, or $D$. 
%\TE{They should be included, but these are different from the boundary logical operators that were initially introduced, since here we are working with the reduced density matrix of the canonical purification. We also cannot add all of them, since they are not mutually commuting. The ones that are supported on a single-site in the top and bottom layer should be sufficient, but it would be good to do the counting to be  sure.} \DL{Here, again I think we include all boundary weak symmetry stabilizers in the form of $g_e \otimes g_e^*$. Maybe it would be better to call them boundary weak symmetry ``stabilizer" instead of boundary weak symmetry logical operator?} \TE{I changed it. Feel free to remove these comments.}
By construction, the conjugated string operators $\widetilde{N}_A, \widetilde{N}_B, \widetilde{N}_C, \widetilde{N}_D$ commute with all generators of $\mathcal{S}_{\text{code}} $. 
Hence, they preserve the code subspace and act as logical unitary operators within it. 
% We will establish the factorization of edge charges in this excitation Hilbert space, characterized as the code subspace of $\mathcal{S}_{\text{code}} $. 

\subsubsection*{Number of logical qudits}

To characterize the code $\mathcal{S}_{\text{code}}$ (equivalently, the excitation Hilbert space), we begin by counting the number of logical qudits it supports. 

\begin{lemma}
Let $m_{AB}, m_{AC}, \cdots$ be the number of faces cut by the $AB, AC, \cdots$ edges. 
Let
\begin{align}
N = (m_{AB}+1) +  (m_{AC}+1) + (m_{AD}+1) +  (m_{BC}+1) + (m_{BD}+1) +  (m_{CD}+1). 
\end{align}
Then, the code supports 
\begin{align}
k = N-1
\end{align}
logical qudits. 
\end{lemma}

\begin{proof}
The code $\mathcal{S}_{\text{code}}$ is obtained by removing those stabilizer generators of $|\Psi_{\mathcal{S}}\rangle$ which are cut by the $AB, AC,\cdots$ edges. 
The number of logical qudits is therefore determined by the number of independent stabilizer generators that are removed. 
(Here, it suffices to consider stabilizers supported along edges, since tri-junction stabilizers can be generated from weak symmetry operators, as discussed in Lemma~\ref{lemma:edge_junction}.)

We begin by analyzing a single edge, say $AB$.
The potentially violated stabilizers consist of face stabilizers along the edge, together with weak symmetry stabilizers:
\begin{align}
S_{f_1}^{(AB)}\otimes I ,\cdots, S_{f_{m_{AB}}}^{(AB)}\otimes I, S_{e_0}^{(AB)}, \cdots, S_{e_{m_{AB}}}^{(AB)} \label{eq:edge_stabilizer}
\end{align}
where $m_{AB}$ represents the number of faces cut by the $AB$ edge. 
Graphically, these operators are illustrated as
\begin{align}
\figbox{3.5}{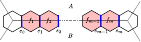}. 
\end{align}

However, these operators are not all independent. 
Using the local redundancy condition (see Eq.~\eqref{eq:local_redundancy} and  \eqref{eq:local_redundancy_figure}), we have
\begin{align}
S_{e_{j-1}}S_{f_j}S_{e_{j}} = \figbox{3.5}{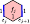}  \simeq I
\end{align}
within the codeword subspace (equivalently, the excitation Hilbert space).
As a result, $m_{AB}$ face stabilizers are actually redundant since they can all generated from the $m_{AB}+1$ weak symmetry stabilizers:
\begin{align}
\{ \ S_{e_0}^{(AB)}, \cdots, S_{e_{m_{AB}}}^{(AB)} \ \} \qquad \figbox{3.5}{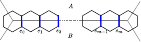} \ .
\end{align}

In addition, there is a  global constraint corresponding to conservation of total edge charge: %\TE{I don't understand this global constraint.} \DL{I think this global constraint is evaluated in the subspace stabilized by $\mathcal{S}_{\text{code}}$. Explicitly, multiple of all face stabilizer on top and bottom layers becomes identity. (We replace face stabilizers at the junction to product of gauge operators at edge charges)} \TE{Yes, we will just need to choose the boundary stabilizers correctly so that this is an identity.}

% \ZL{I changed the ``label" to $\mathcal{E}$. We can easily change it back by redefine the command \textbackslash Edge.}
\begin{align}
\prod_{\Edge}W_{\Edge} \simeq I, \qquad \Edge \in \{ AB, AC, AD, BC, BD, CD \} .
\end{align}
Each edge charge operator $W_{\Edge}$ can be expressed in terms of weak symmetry stabilizers. For example, 
\begin{align}
W_{AB} \simeq \prod_{j=0}^{m_{AB}} S_{e_j}^{-2}.
\end{align}
Therefore, the conservation law becomes\footnote{
Here, $d$ is assumed to be an odd integer. 
For even $d$, this instead yields a constraint of order $d/2$. 
} 
\begin{align}
\prod_{\Edge} \prod_{j=0}^{m_{\Edge}} S^{-2}_{e_j} \simeq I. 
\end{align}

Summing over all edges, we obtain
\begin{align}
N = (m_{AB}+1) +  (m_{AC}+1) + (m_{AD}+1) +  (m_{BC}+1) + (m_{BD}+1) +  (m_{CD}+1),
\end{align}
Taking into account the single global redundancy, we conclude that the code $\mathcal{S}_{\text{code}}$ supports
\begin{align}
k = N-1
\end{align}
logical qudits. 
\end{proof}

\subsubsection*{Logical operators}

It is convenient to identify the excited states code space $\mathcal{H}_{\text{code}}$, defined by $\mathcal{S}_{\text{code}}$, as a subspace of $\mathcal{H}_{\text{ext}}$ consisting of $N = k+1$ (logical) qudits and subject to a single global $\mathbb{Z}_d$ constraint. 
This embedding can be made explicit by constructing a complete set of Pauli logical operators.

The excitation code Hilbert space $\mathcal{H}_{\text{ext}}$ naturally decomposes into six subsystems, one associated with each bipartite edge:
\begin{align}
\mathcal{H}_{\text{ext}} = \bigotimes_{\Edge\in \{ AB,AC, \cdots \}} \mathcal{H}^{(\Edge)}, \qquad \dim \mathcal{H}^{(\Edge)} = d^{m_{\Edge}+1}.
\end{align}
In this new excitation code Hilbert space, logical qudits are located at the interfaces between sub-regions and the initial state is a product state. 
Logical $Z$ operators are identified with weak symmetry stabilizers: 
\begin{align}
S_{e^{(\Edge)}_j } = \figbox{3.5}{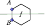} \ \rightarrow  \ \overline{Z_j}^{(\Edge)}.  \label{eq:Z} 
\end{align}
The global constraint, $\prod_{\Edge} \prod_{j=0}^{m_{\Edge}} S^{-2}_{e_j} \simeq I$, is realized as a global $\mathbb{Z}_d$ symmetry:
\begin{align}
\prod_{\Edge} \prod_{j=0}^{m_{\Edge}} \overline{Z_j}^{(\Edge)} |\psi\rangle = |\psi\rangle, \qquad |\psi\rangle \in \mathcal{H}_{\text{ext}}\label{eq:global_symmetry}
\end{align}
Furthermore, each edge charge $W_{\Edge}$ corresponds to the symmetry action localized on that edge:
\begin{align}
W_{\Edge} \rightarrow \prod_j \overline{Z_j^{-2}}^{(\Edge)}.
\end{align}
Hence, the excitation code subspace $\mathcal{H}_{\text{ext}}$ can be viewed as a $\mathbb{Z}_d$-symmetric Hilbert space decomposed into six subsystems, where the local $\mathbb{Z}_d$ charges are characterized by the edge charge operators $W_{\Edge}$.

Next, we specify the logical $X$ operators. 
Because of the constraint in Eq.~\eqref{eq:global_symmetry}, admissible logical $X$ operators must commute with the global $Z$ product. 
These are generated by operators of the form $X\otimes X^{\dagger}$, which transport anyon charges along or between edges. 
Here, truncated segments of face stabilizers on the $AB$ edge act as logical operators. 
Specifically, by examining the commutation relations with logical $Z$ operators, one can deduce the following identifications:
\begin{align}
\figbox{3.5}{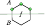} \ \rightarrow \ \overline{X_j}^{(\text{AB})}
\overline{X_{j+1}^{\dagger}Z_{j+1}^{\dagger}}^{(\text{AB})} \qquad 
\figbox{3.5}{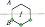} \  \rightarrow \ 
\overline{X_{j}^{\dagger}Z_j^{\dagger}}^{(\text{AB})}
\overline{X_{j+1} }^{(\text{AB})}
\label{eq:XXZ}
\end{align}
which shift charge by one site along the edge with a dressing by logical $Z$ operators.  
These logical operators are locally symmetric (i.e. neutral) with respect to the edge charge, commuting with $W_{AB}$. 

Here, the decoration by a logical $Z$ operator is essential, reflecting the fact that truncated face stabilizers supported within the same subsystem (e.g., $A$) do not, in general, commute with one another. 
This also ensures that truncated face stabilizers supported on the different subsystems ($A$ and $B$) commute with each other.

Together with logical $Z$ operators, these truncated face logical operators form a generating set for symmetric operators on each edge Hilbert space $\mathcal{H}^{(\Edge)}$.
Namely, for an arbitrary logical Pauli operator $O^{(\Edge)}$ supported on a particular edge, we have 
\begin{align}
\overline{O^{(\Edge)}} \in 
\left\langle \figbox{3.5}{fig_XX_logical.pdf} , \figbox{3.5}{fig_XX_logical2.pdf},  \figbox{3.5}{fig_Z_logical.pdf} \right\rangle,
\label{Eq:edgelogicalgenerator}
\end{align}
meaning that $\overline{O^{(\Edge)}}$ can be expressed as a product of truncated face stabilizers and weak stabilizers. 
These operators therefore generate the full algebra of locally symmetric operators on each edge. 

We have constructed a generating set for symmetric operators localized on each edge.
We also need to consider operators that are globally symmetric operators, but are not locally symmetric on each edge.
A representative example is
\begin{align}
\overline{X_j}^{(\text{AB})}\overline{X_{j'}^\dagger}^{(\text{AD})}.
\end{align}
Such operators can be constructed from Pauli strings connecting the $AB$ and $AD$ edges.
Together, these $X$ and $Z$ operators provide a complete set of logical Pauli operators, and hence specify the embedding of the excitation space into the $N=k+1$-qudit Hilbert space $\mathcal{H}_{\text{ext}}$.

\begin{observation}
The code Hilbert space decomposes into six subsystems with a global $\mathbb{Z}_d$ symmetry:
\begin{align}
\mathcal{H}_{\text{ext}} = \bigotimes_{\Edge\in \{ AB,AC, \cdots \}} \mathcal{H}^{(\Edge)}, \qquad U_{\mathrm{sym}}=\prod_{\Edge} \prod_{j=0}^{m_{\Edge}} \overline{Z_j}^{(\Edge)}.
\end{align}
Each edge charge corresponds to a local $\mathbb{Z}_d$ symmetry generator
\begin{align}
W_{\Edge} = \prod_{j=0}^{m_{\Edge}} \overline{Z_j^{-2}}^{(\Edge)}.
\end{align}
Moreover, the algebra of locally symmetric operators supported on $\mathcal{H}^{(\Edge)}$ is generated by operators localized near the corresponding edge.
\end{observation}

\subsubsection*{Locality of logical operators}

With the above definitions of $\mathcal{H}_{\text{code}}$ and $\mathcal{H}_{\text{ext}}$ in mind, let us revisit the edge-local string operators (Def.~\eqref{def:Edge-local string}) $\widetilde{N}_{A}$.
Since $\widetilde{N}_{A}$ is a logical operator of $\mathcal{S}_{\text{code}}$, it preserves the code subspace and defines a map 
\begin{align}
\widetilde{N}_{A}: \mathcal{H}_{\text{code}} \to \mathcal{H}_{\text{code}}.
\end{align}
Through the isometric embedding of $\mathcal{H}_{\text{code}}$ into $\mathcal{H}_{\text{ext}}$, we can represent it by a symmetric operator on $\mathcal{H}_{\text{ext}}$, which we denote by $\overline{\widetilde{N}_{A}}$.
We now claim that $\overline{\widetilde{N}_{A}}$ can always be chosen to act non-trivially only on $\mathcal{H}^{(AB)} \otimes \mathcal{H}^{(AC)} \otimes \mathcal{H}^{(AD)}$.
In other words, the locality of $\widetilde{N}_{A}$ in the original Hilbert space is preserved in $\mathcal{H}_{\text{code}}$ where it may act non-trivially only on the edges which $\widetilde{N}_{A}$ may touch. 

For ease of notation, let
\begin{align}
\mathcal{H}_1=\mathcal{H}^{(AB)} \otimes \mathcal{H}^{(AC)} \otimes \mathcal{H}^{(AD)}, \qquad \overline{S}_1=W_{AB}\otimes W_{AC} \otimes W_{AD},
\end{align}
and define $\mathcal{H}_2$ and $\overline{S}_2$ analogously for the edges $\{BC, BD, CD\}$.
The key observation is that every symmetric operator on $\mathcal{H}_2$ admits a realization by a physical operator supported entirely on the subsystem $BCD$. 
More concretely, the algebra of symmetric operators on $\mathcal{H}_2$ is generated by operators of the form $\overline{Z_j}^{\Edge}$ and $\overline{X_j}^{\Edge} \overline{X_{j'}^\dagger}^{\Edge'}$, where $\Edge, \Edge'\in\{BC,BD,CD\}$, which correspond to physical operators supported on $BCD$ as discussed above.  
Since $\widetilde{N}_{A}$ is supported entirely on $A$, it commutes with all physical operators supported on $BCD$. Therefore, its representative $\overline{\widetilde{N}_{A}}$ commutes with every symmetric operator acting on $\mathcal{H}_2$.

Such commutativity suggests that there exists a symmetric operator acting solely on $\mathcal{H}_a$ whose action matches that of $\overline{\widetilde{N}_{A}}$. 
(Note that this statement holds even if $\overline{\widetilde{N}_{A}}$ is a non-Pauli operator.)
To see this, decompose $\mathcal{H}_{1}$ and $\mathcal{H}_2$ into eigenspaces of $\overline{S}_{1}$ and $\overline{S}_2$ respectively. Then 
\begin{equation}
    \mathcal{H}_{\text{code}} =\bigoplus_{q\in\mathbb{Z}_d}\mathcal{H}_1^{(q)} \otimes \mathcal{H}_2^{(-q)}. 
\end{equation}
The algebra of symmetric operators on $\mathcal{H}_2$ can be identified with
\begin{align}
\mathcal M_2 = 
\bigoplus_q
I_{\mathcal H_1^{(q)}}
\otimes
\mathcal B(\mathcal H_2^{(-q)}).
\end{align}
Since $\overline{\widetilde{N}_{A}}$ commutes with all such operators, it belongs to the commutant of $\mathcal{M}_2$ inside $\mathcal{B}(\mathcal{H}_{\text{code}})$:
\begin{equation}
    \mathcal{M}_2' = \bigoplus_q \mathcal{B}(\mathcal{H}_1^{(q)})\otimes I_{\mathcal{H}_2^{(-q)}}.
\end{equation}
This algebra is naturally identified with the algebra of symmetric operators acting on $\mathcal{H}_1$.
Hence, $\overline{\widetilde{N}_{A}}$ admits a representation $\overline{\widetilde{N}_{A}}$ supported on $\mathcal{H}_1$ only.

Moreover, Observation \ref{obs:edgelocal} implies that $\overline{\widetilde{N}_{A}}$ cannot create excitations on the $AC$ edge. 
Consequently, if
\begin{align}
\overline{Z}_j^{(AC)}=1 \qquad \text{for all } j,
\end{align}
then the same relations continue to hold after acting with $\overline{\widetilde{N}_{A}}$.
Therefore, upon restricting to the subspace of $\mathcal{H}_{\text{ext}}$ satisfying $\overline{Z}_j^{(AC)}=1$ for all $j$, the operator $\overline{\widetilde{N}_{A}}$ acts effectively only on
\begin{align}
\mathcal{H}^{(AB)}\otimes \mathcal{H}^{(AD)}.
\end{align}
In this sense, $\overline{\widetilde{N}_{A}}$ is effectively a two-local operator on the edge decomposition. 

Similar conclusions hold for the other edge-local string operators, such as $\widetilde{N}_{B}$, $\widetilde{N}_{C}$, and $\widetilde{N}_{D}$, as well as their conjugates.

\subsubsection*{Factorization}

We are now ready to establish the factorization of edge charges (Lemma~\ref{lemma:charge_factorization_partite}). 

\begin{proof}
The canonical purification $|\Psi_{\mathcal{S}}\rangle$ corresponds to the trivial code state 
\begin{align}
|\psi\rangle = |0\rangle^{\otimes N} \in \mathcal{H}_{\text{ext}},
\end{align}
which is the simultaneous $+1$ eigenstate of all stabilizers. 
Crucially, this state factorizes across the six edge subsystems, and therefore exhibits no correlations between distinct edges.
(This makes precise the intuition that correlations between, e.g., $R_{AB}$ and $R_{AD}$ are not ``utilized''.) 

In the excitation code Hilbert space $\mathcal{H}_{\text{ext}}$, we have showed that the conjugated string operators (e.g., $\widetilde{N}_{A}$) act as 2-local logical operators on the edge decomposition, when a relevant edge (e.g., edge $AC$ for $\widetilde{N}_{A}$) is unexcited.
Now we apply a standard lightcone argument to establish the factorization.
First note that $\tilde{N}_A\ket{\Psi_\rho}=\tilde{N}_C^\dagger \tilde{N}_B^\dagger \tilde{N}_D^\dagger\ket{\Psi_\rho}$.
On $\mathcal{H}_{\text{ext}}$, this becomes:
\begin{equation}
\overline{\tilde{N}_A}\ket{\psi}=\overline{\tilde{N}_C^\dagger} \overline{\tilde{N}_B^\dagger} \overline{\tilde{N}_D^\dagger}\ket{\psi}.
\end{equation}
In the state $\overline{\tilde{N}_B^\dagger} \overline{\tilde{N}_D^\dagger}\ket{\psi}$, the edge $AB$ can only be entangled with $BD$, and the edge $AD$ can only be entangled with $CD$.

After applying $\overline{\tilde{N}_C^\dagger}$, $AB$ and $AD$ remain unentangled. (Here, note that edge $AC$ remains in the state $\ket{0}^{\otimes(m_{AC}+1)}$ throughout this process. Similar observations apply for edge $BD$).
Therefore,  
\begin{align}
\langle W_{AB} W_{AD}\rangle =  \langle W_{AB}\rangle \langle W_{AD}\rangle 
\end{align}
which establishes the factorization of edge charges.
\end{proof}

\subsubsection*{Concentrating edge charges}

With the factorization of edge charges established, the final step in the proof of Theorem~\ref{thm:four-partite} is to show that edge charges can be concentrated into point-like edge anyons by applying suitable local unitaries, without changing the corresponding topological spin.

Consider the string operator $\widetilde{N}_A$ illustrated in Fig.~\ref{fig_4_partite_string}, which creates edge charges $W_{AB}$ and $W_{AD}$. 
We have established that the excited state $(\widetilde{N}_A \otimes I) |\Psi_{\rho}\rangle$ can be represented as a product state in the code Hilbert space $\mathcal{H}_{\text{ext}}$:
\begin{align}
(\widetilde{N}_A \otimes I) |\Psi_{\rho}\rangle \rightarrow  |\psi\rangle_{AB} \otimes |\phi\rangle_{AD} \otimes |0\rangle_{AC} \otimes |0\rangle_{BC} \otimes |0\rangle_{BD} \otimes |0\rangle_{CD}.
\end{align}
This state satisfies $\langle W_{AB}\otimes I \rangle = \omega^j$ and $\langle W_{AD}\otimes I \rangle = \omega^{-j}$ for some $j$. 

Now consider a Pauli string operator $M_A$ that creates the same pair of edge charges:
\begin{align}
(M_A \otimes I) |\Psi_{\rho}\rangle \rightarrow  |\psi'\rangle_{AB} \otimes |\phi'\rangle_{AD} \otimes |0\rangle_{AC} \otimes |0\rangle_{BC} \otimes |0\rangle_{BD} \otimes |0\rangle_{CD}.
\end{align}
Since $|\psi\rangle$ and $|\psi'\rangle$ have the same conserved $\mathbb{Z}_d$ charge, there exists a unitary operator $V_{AB}$ such that
\begin{align}
|\psi'\rangle = V_{AB} |\psi\rangle
\end{align}
within the excitation code Hilbert space $\mathcal{H}_{\text{ext}}$. 
We then consider the corresponding logical operator $\widetilde{V}_{AB}$, acting as $V_{AB}$ within the code Hilbert space. 
Because the logical operators living on each edge form a complete operator basis for all unitaries that commute with the edge charge operator, $\widetilde{V}_{AB}$ can be expanded as a polynomial in these generators. 
Consequently, $\widetilde{V}_{AB}$ acts locally along the $AB$ edge and maps the extended edge charge into a point-like anyonic excitation localized near the edge. 

Moreover, this edge-local modification by $\widetilde{V}_{AB}$ commutes with other string operators such as $\widetilde{N}_B$ and $\widetilde{N}_C$, as depicted in Fig.~\ref{fig_4_partite}, and therefore does not affect the topological spin.
Repeating the same argument for $\widetilde{N}_B$ and $\widetilde{N}_C$, we establish that the topological spin depends only on the edge anyon type. 
The remainder of the argument then proceeds identically to the LO-chirality case.
This completes the proof of Theorem~\ref{thm:four-partite}.

\subsection{Three-partite non-chirality}

% \ZL{My comments in this subsection are just for clarity; I think the argument is sound and I understand it's a sketch anyway. Given the timeframe, we can simply proceed with the current version. Alternatively, I can try to write a slightly expanded version.}

We have established that $\rho$ is four-partite chiral whenever its anyon content is not mirror invariant. 
A natural question is whether $\rho$ can exhibit three-partite chirality. 
Here, we show that the state is three-partite non-chiral with respected to a class of natural tri-partition Fig.~\ref{fig_3_partite}.

\begin{theorem}\label{thm:tri-partite}
The $\mathbb{Z}_d^{(k)}$ mixed state $\rho$, as defined in Definition \ref{def:honeycombmodel}, is not three-partite chiral for the $Y$-shaped partition in Fig.~\ref{fig_3_partite}.
\end{theorem}
% \YP{Maybe adding the qualifier ``For any connected disk-like $Y$-partition of the form in Fig.~\ref{fig_3_partite}...'' }
This theorem essentially implies that tri-partite entanglement measures cannot detect chirality in stabilizer mixed states. 

Below we present a proof sketch of this result. 
The basic strategy is to show that $\rho_{ABC}$ decomposes into EPR entanglement or GHZ-type classical correlation, neither of which can be three-partite chiral. \footnote{This fact can be also verified by considering the excitation Hilbert space with respect to the tripartition $ABC$.}
Recall that an EPR pair can be characterized by the following stabilizer generators:
\begin{align}
X\otimes X, \quad Z\otimes Z^{\dagger} \label{eq:EPR}
\end{align}
whose single-qudit restrictions do not commute, while the full operators mutually commute.
We will show that stabilizer generators with support across multiple subsystems can be organized into sets that satisfy the above commutation relations.

Let us focus on face stabilizers $S_{f}$ supported along the edge $AB$. 
We first assume that there are an \emph{odd number} of face stabilizers $S_{f}$ on the edge.
Let $T_{0}$ denote the stabilizer at the tri-junction (see Fig.~\ref{fig_3_partite}).
Along the $AB$ edge, face stabilizers can be arranged into pairs where each column (except the last) has the commutation structure of an EPR pair:
\begin{align}
\begin{bmatrix}
S_1^{(AB)}, & S_1^{(AB)} S_3^{(AB)},& S_1^{(AB)} S_3^{(AB)} S_5^{(AB)}, & \cdots, & S_1^{(AB)}S_3^{(AB)}\cdots S_{m_{AB}}^{(AB)} \\
S_2^{(AB)}, & S_4^{(AB)},& S_6^{(AB)}, & \cdots , & T_0  \\
\end{bmatrix} \label{eq:JW}
\end{align}
For instance, in the first column, the pair $S_1^{(AB)}$ and $S_2^{(AB)}$ has the property that their restrictions to each subsystem ($A$ and $B$) do not commute, while the full operators commute. 
The above arrangement is chosen so that different columns correspond to independent EPR-like commutation relations.
This construction can be viewed as a qudit generalization of the Jordan–Wigner transformation, which organizes non-commuting operators into mutually commuting pairs.
It follows that there are $\frac{m_{AB}-1}{2}$ independent EPR pairs along the $AB$ edge.
The same argument applies to the $AC$ and $BC$ edges. 

The final column generates classical GHZ-like correlation among $A,B,C$.
Namely, these stabilizer generators on the $AB$, $BC$, and $AC$ edges, and the tri-junction stabilizer $T_0$, can be transformed into the following set of the GHZ-state stabilizer generators
\begin{align}
\figbox{3.5}{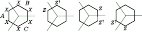}
\end{align}
after applying appropriate three-partite Clifford unitary. 
This set of stabilizers cannot generate chirality as the stabilizer group generated by these four stabilizers is invariant under complex conjugation. 
Note that these GHZ-like entanglement are fully decoupled from other EPR-like entanglement. 

Next, assume that there are an \emph{even number} of face stabilizers on the $AB$ edge while other edges have odd numbers of face stabilizers.
We then have
\begin{align}
\begin{bmatrix}
S_1^{(AB)}, & S_1^{(AB)} S_3^{(AB)},&  \cdots, & S_1^{(AB)}S_3^{(AB)}\cdots S_{m_{AB}-1}^{(AB)}, &  S_1^{(AB)}S_3^{(AB)}\cdots S_{m_{AB}-1}^{(AB)} T_0 \\
S_2^{(AB)}, & S_4^{(AB)}, & \cdots , & S_{m_{AB}}^{(AB)},&   \emptyset \\
\end{bmatrix} 
\end{align}
where the tri-junction $T_0$ is decorated by face stabilizers so that the decorated $T_0$ is decoupled from other EPR-like entanglement. 
In this case, a stabilizer across the $AB$ edge in the GHZ-state stabilizer set is absent:
\begin{align}
\figbox{3.5}{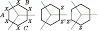}.
\end{align}
This set of stabilizers cannot generate chirality. 
The same argument applies to the cases where the $AC$ and $AB$ edges have even numbers of face stabilizers, leading to Theorem~\ref{thm:tri-partite}.

\begin{figure}[h!]
\centering
\raisebox{\height}{\hspace{5pt}}\raisebox{-0.85\height}{\includegraphics[width=0.35\textwidth]{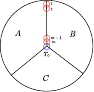}}
\hspace{10pt}
\caption{Partition into three subsystems with the tri-junction stabilizer $T_0$ at the center.  
}
\label{fig_3_partite}
\end{figure}

\subsection{Toward four-partite chirality in pure state realizations}

So far, we have focused on establishing four-partite chirality (and three-partite non-chirality) for the \emph{mixed state} realizations of the $\mathbb{Z}_{d}^{(1)}$ models. 
A natural question is whether similar characterizations hold for \emph{pure state} realizations of the $\mathbb{Z}_{p^{2k}}^{(1)}$ anyon theories. 

Given that the LO-chirality proof extends straightforwardly to pure-state realizations, one might expect that the four-partite chirality result also extends. 
However, the situation is more subtle.
In this subsection, we point out an important obstruction in the $\mathbb{Z}_{p^{2k}}^{(1)}$ models that prevents a straightforward extension of our proof to pure states.

\subsubsection*{$\mathbb{Z}_p$ v.s. $\mathbb{Z}_{p^2}$ charges}

To recap, we consider the pure-state realization $|\psi_{ABCD}\rangle$ of the $\mathbb{Z}_{p^{2}}^{(1)}$ anyons with $p=3$ (mod $4$) where additional edge stabilizers are introduced. 
These edge stabilizers are constructed by taking the $p$th power of weak symmetry edge operators in the mixed-state realization. 
Due to this modification, the definition of edge charge operators changes accordingly:
\begin{align}
W_{AB} \equiv g_e^{-p} g_{e'}^{-p} \prod_f S_f^p \otimes I = \figbox{3.5}{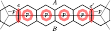}\ \end{align}
where the $p$th powers of face stabilizers are used to match the modified edge stabilizers.
The key point is that these edge operators now characterize $\mathbb{Z}_p$ charges rather than $\mathbb{Z}_{p^2}$ charges.
In other words, edge charge operators detect only the charge modulo $p$, rather than the full $\mathbb{Z}_{p^2}$ charge.

Given these order-$p$ edge charges, one can still establish a form of edge charge conservation.
This reflects the fact that charge can leak through tri-junctions, while edge operators are sensitive only to its mod-$p$ component.
However, this loss of information is precisely what obstructs our previous argument.
Since edge charges no longer fully resolve the anyon type, the factorization and concentration arguments used in the mixed-state case no longer apply directly.
As a result, our proof does not currently extend to pure-state realizations.

% \YP{Please check if this change is ok.}
One possible route to overcoming this obstruction is to exploit the remaining spin information contained in the mod-$p$ charge.
Namely, although the charge modulo $p$ does not determine the full topological spin, it determines the reduced quantity
\begin{align}
\vartheta(\bar a):=\theta(a)^p
=
\exp\left(\frac{2\pi i \bar a^2}{p}\right),
\qquad
\bar a := a \quad (\text{mod $p$}).
\end{align}
This reduced quantity is still not invariant under complex conjugation for $p=3$ (mod $4$).
This suggests that knowledge of the charge modulo $p$ may already be sufficient to detect the obstruction to complex conjugation.
At present, however, it remains an open problem whether this idea can be developed into a complete proof.

\subsubsection*{Three-partite chiral mixed state under Clifford}

Related to the aforementioned problem, an interesting connection arises between pure-state and mixed-state chirality. 
We begin by recalling the following result~\cite{vardhan2025chiralitymagicquantumcorrelations}.

\begin{lemma}
An $n$-partite pure state $|\psi_{A_1 \cdots A_n}\rangle$ is $n$-partite LO-chiral if and only if its $(n-1)$-partite reduced density state $\rho_{A_2\cdots A_n}$ is $(n-1)$-partite LO-chiral. 
\end{lemma}

This lemma provides a useful perspective.
Consider the pure-state realization of $\mathbb{Z}_{p^2}^{(1)}$, partitioned into four subsystems $A,B,C,D$.
The question of whether this state is four-partite chiral is therefore equivalent to asking whether its reduced density matrix on $A,B,C$ is three-partite chiral.

While we do not resolve this question in full generality, progress can be made by restricting to Clifford operations.
Specifically, we ask whether four-partite chirality of $|\psi_{ABCD}\rangle$ (equivalently, three-partite chirality of $\rho_{ABC}$) can be established under tensor-product Clifford operations.
To this end, we construct a three-qudit mixed state $\rho_{123}$ for $d=p^2$ with $p \equiv 3 \ (\mathrm{mod}\ 4)$.
We then prove that it is not possible to bring $\rho_{123}$ into $\rho_{123}^*$ by any tensor-product Clifford operations.

\begin{lemma}
Consider the three-qudit $d=p^2$ mixed state stabilized by:
\begin{align}
S_{123} &= XZ \otimes XZ \otimes XZ \\
S_{12} &= X^p \otimes Z^p \otimes I \\
S_{23} &= I \otimes X^p \otimes Z^p  \\
S_{31} &=Z^p\otimes I \otimes  X^p  
\end{align}
This state is three-partite chiral under Clifford if $p=3 \text{~mod~} 4$. 
\end{lemma}

This construction is motivated by the small-system limit, and it is notable that the same number-theoretic condition appears here. 
This may hint that the corresponding pure-state statement also holds, although this remains open.\footnote{ 
This construction is closely related to additional algebraic structure specific to $\mathbb{Z}_{p^2}$ stabilizers. 
Tripartite entanglement structures in pure $\mathbb{Z}_{p^2}$ stabilizer states have been recently studied in Ref.~\cite{wong2025}, which proved that the states contain GHZ- or EPR-like entanglement only.
However, in our understanding, it remains open whether such reductions can be performed by Clifford operations or require non-Clifford operations.
}

\begin{proof}
A single-qudit Clifford operation can be characterized (up to phases) as follows:
\begin{align}
X_j \longrightarrow X_j^{a_j} Z_j^{b_j}\\
Z_j \longrightarrow X_j^{c_j} Z_j^{d_j}
\end{align}
for $j=1,2,3$, with 
\begin{align}
a_j d_j - b_j c_j = 1 \qquad \text{mod $p^2$}.
\end{align}

Suppose that a tensor-product Clifford operation implementing $\rho \rightarrow \rho^*$ exists.
The two-body stabilizer $S_{12}$ is transformed (up to phases) as
\begin{align}
S_{12} = (X_1)^p \otimes (Z_2)^p \longrightarrow (X_{1})^{p a_1} (Z_{1})^{p b_1} \otimes (X_2)^{p c_2} (Z_2)^{p d_2}.
\end{align}
The resulting operator must be a product of $S_{12}^*, S_{13}^*, S_{23}^*$ and $S_{123}^*$.
Since it supports only on 1 and 2, it must be $S_{12}^*$ or its powers (note that $S_{123}^p=S_{12}S_{23}S_{13}$).
Similarly conclusions hold for $S_{23}$ and $S_{31}$.
This implies
\begin{align}\label{eq:p2qudit-1}
p b_j = p c_j = 0 \qquad \text{mod $p^2$}, \qquad j=1,2,3,
\end{align}
and
\begin{align}\label{eq:p2qudit-2}
p a_1 = - p d_2,\quad
p a_2 = - p d_3,\quad
p a_3 = - p d_1 \qquad \text{mod $p^2$}.
\end{align}

Equation~\eqref{eq:p2qudit-1} implies
\begin{align}
b_j = c_j = 0 \qquad \text{mod $p$}, \qquad j=1,2,3,
\end{align}
and hence
\begin{align}
a_j d_j = 1 \qquad \text{mod $p^2$}, \qquad j=1,2,3.
\end{align}
In particular, $d_j = (a_j)^{-1}$ mod $p^2$, and therefore also mod $p$.

Equation~\eqref{eq:p2qudit-2} implies
\begin{align}
a_1 = - d_2,\quad
a_2 = - d_3,\quad
a_3 = - d_1 \qquad \text{mod $p$}.
\end{align}
Combining these relations yields
\begin{align}
a_1 = -d_2 = -a_2^{-1} = d_3^{-1} = a_3 = -d_1 = -a_1^{-1} \qquad \text{mod $p$},
\end{align}
which implies
\begin{align}
a_1^2 = a_2^2 = a_3^2 = -1 \qquad \text{mod $p$}.
\end{align}
However, this equation has no solution when $p \equiv 3 \ (\mathrm{mod}\ 4)$.
\end{proof}

One interesting feature of this result is that three-partite chirality under Clifford can be established without including $S_{123}$, as seen in the proof. 
Since the remaining stabilizer generators are two-local, this suggests that chirality can arise from purely bipartite entanglement, at least within the Clifford setting. 

\section{Imaginarity of wavefunctions}
\label{sec5}

So far, we have examined LO-chirality and $n$-partite chirality by asking whether a stabilizer mixed state $\rho$ can be transformed into its complex conjugate $\rho^*$ via local quantum channels:
\begin{align}
\rho \xlongrightarrow[?]{LO}\rho^* .
\end{align}
In this section, we turn to a related, but conceptually distinct question: whether $\rho$ can be represented using only real-valued matrix elements in some tensor product basis
\begin{align}
\rho \xlongrightarrow[?]{LU} \text{real \ } \sigma .
\end{align}
While LO-chirality concerns the relation between a state $\rho$ and its complex conjugate $\rho^*$, the question here concerns whether the many-body state itself can be made entirely real by local operations or not.
Here, we restrict ourselves to local unitaries (LU) instead of local channels (LO). 
This is because a complete depolarizing channel transforms any state into the maximally mixed state, which is real-valued. 
Therefore, this question is non-trivial only for LU transformations. 
\footnote{Alternatively, one may formulate the notion using two-way LO connectivity between $\rho$ and a real-valued state $\sigma$.}

\begin{definition}
A density matrix $\rho$ is said to be \emph{LU-real} when there exists a finite-depth local unitary $U$ which transforms $\rho$ such that
\begin{align}
\sigma = U\rho U^{\dagger}, \qquad \sigma = \sigma^*
\end{align}
in some tensor product basis. 
If no such local unitary exists, $\rho$ is said to be \emph{LU-imaginary}.
\end{definition}

To gain intuition for this notion, let us begin with some prototypical examples. 
Consider the $\mathbb{Z}_d$ toric code, defined by vertex $X$-type stabilizers $A_v$ and plaquette $Z$-type stabilizers $B_p$. 
One of its ground states can be written as
\begin{align}
|\psi_{\text{t.c.}}\rangle \propto \prod_v (I + A_v) |0\rangle^{\otimes n} 
= \sum_{\gamma \in \text{loops}} |\gamma\rangle,
\end{align}
which is a superposition of all closed loop configurations in the dual lattice. 
In the computational basis $|j\rangle$ ($j=0,\dots,d-1$), all amplitudes are non-negative real numbers. 
Thus, the toric code ground state provides an example of a topologically ordered state that is manifestly LU-real.

A contrasting example is the double semion model~\cite{Levin_2012}. 
In this case, a ground state takes the form
\begin{align}
|\psi\rangle \propto \sum_{\gamma \in \text{loops}} (-1)^{s(\gamma)} |\gamma\rangle,
\end{align}
where $s(\gamma)=0,1$ depends on the loop configuration~\cite{Levin_2012}. 
The wavefunction is still real, but now necessarily contains both positive and negative amplitudes. 
Indeed, Hastings proved that the double semion ground state has an \emph{intrinsic sign problem}, meaning that no local unitary transformation can render all amplitudes non-negative~\cite{Hastings_2015}. 
Subsequent works related such sign obstructions to the braiding properties of anyons~\cite{Golan:2020kyp, Seo:2025gvq}. 
Here, we extend this line of characterization further and ask whether local unitary circuits can remove complex phases altogether.  
In other words, we ask whether certain topological phases possesses \emph{intrinsic many-body imaginarity}.

In particular, we study whether mixed states realizing $\mathbb{Z}_d^{(1)}$ topological order can be made real by local unitaries. 
This notion is closely related to LU-chirality. 
Indeed, if a state were LU-real, it could be mapped to a real density matrix $\sigma=\sigma^*$, which would imply that $\rho$ and $\rho^*$ are related by local unitaries. 
Therefore, any LU-chiral state is necessarily LU-imaginary. 
% However, the converse need not hold as a state may fail to be LU-real even if it is not LU-chiral.
% A priori, one might expect LU-imaginarity to coincide with LU-chirality. 
% Our main result of this section shows that this is not the case. 
% Even certain LU-non-chiral $\mathbb{Z}_d^{(1)}$ states possess intrinsic imaginarity that cannot be removed by local unitaries.
However, it is unclear whether the converse statement holds: does LU-nonchiral also implies LU-real?
% A priori, one might expect LU-imaginarity to coincide with LU-chirality. 

Our main result of this section shows that this is not the case: certain LU-non-chiral $\mathbb{Z}_d^{(1)}$ states possess intrinsic imaginarity that cannot be removed by local unitaries.

\begin{theorem}\label{thm:real}
The $\mathbb{Z}_{d}^{(k)}$ mixed state $\rho$ is LU-imaginary for any $d>2$ and coprime $(k,d)$. 
\end{theorem}
\noindent Note that the above result applies to pure state realizations of  $\mathbb{Z}_{d}^{(1)}$ models (e.g. $\mathbb{Z}_{p^2}^{(1)}$) as well.

Here, it is important to emphasize that non-zero chiral central charge does not necessarily imply the LU-imaginarity of $\rho$. 
For instance, the 3F topological order (with $c_{-}=4$) admits a mixed state realization with real-valued density matrices~\cite{Ellison_2023}. 
On the other hand, the $\mathbb{Z}_5^{(1)}$ topological order has $c_{-} = 0$, but any mixed state realizing this phase is LU-imaginary according to the above theorem. 
Thus, intrinsic imaginarity reflects a distinct topological obstruction beyond chirality.

\begin{proof}
%Since LU-chiral states are LU-imaginary, we will focus on LU-non-chiral states. 
%Namely, due to Lemma~\ref{lemma:factorization} and Theorem~\ref{thm:invcondition}, it suffices to establish the LU-imaginarity of the $\mathbb{Z}_{p^m}^{(1)}$ mixed states for the cases with $p=1$ (mod $4$).
%\ZL{It's not clear how the reduction goes. Logically, one can worry about a situation where LU-imaginarity $\otimes$ LU-imaginarity is somehow LU-real. Suggested fix: the argument works directly for all $d$, right?}
%\ZL{same argument works for all $\mathbb{Z}_d^{(k)}$, not just $\mathbb{Z}_d^{(1)}$, right?}

Since LU-chiral states are LU-imaginary, we will focus on LU-non-chiral states. 
Suppose, for contradiction, that a mixed state $\rho$ realizing the $\mathbb{Z}_d^{(k)}$ anyons is LU-real.
Then, there exists a local unitary $U$ such that 
\begin{align}
\sigma = U \rho U^{\dagger}, \qquad \sigma^* = \sigma,
\end{align}
in the computational product basis. 
Recall that $\rho$ admits a (strong-symmetry) string operator $\gamma$ that creates an anyon $a$ at one endpoint and $a^{-1}$ at the other. 
Under a finite-depth LU transformation, string operators are transformed into ``fattened'' string operators with the same statistical properties. 
Hence $\sigma$ admits a string operator
\begin{align}
\tilde{\gamma} = U \gamma U^\dagger
\end{align}
which creates an anyon of type $a$ at its endpoint. 
Here, a string operator means that it is a part of a strong symmetry operator for $\sigma$, in the sense of Lemma~\ref{lemma:anyon_charge}. 

Since $\sigma$ is real in the computational basis, complex conjugation leaves it invariant. 
Then, if $\tilde{\gamma}$ is a string operator for $\sigma$, then its complex conjugate $\tilde{\gamma}^*$ is also a string operator for $\sigma$. 
One can then verify that $\tilde{\gamma}^*$ must create a definite anyon type at its endpoints for $\sigma$.\footnote{More precisely, given $\tilde{\gamma}^*$, consider an operator $U^{\dagger} \tilde{\gamma}^* U$ by conjugating it with $U^{\dagger}$.
One can then verify that this operator $U^{\dagger} \tilde{\gamma}^* U$ is a string operator for $\rho$. 
Due to Lemma~\ref{lemma:anyon_charge}, $U^{\dagger} \tilde{\gamma}^* U$ must create a definite anyon type at its endpoints for $\rho$.
Hence, $\tilde{\gamma}^*$ creates a definite anyon type for $\sigma$ as well. 
}
Denote this anyon by $b$.

Complex conjugation reverses the topological spin:
\begin{align}
\theta(b) = \theta(a)^* .
\end{align}
For the $\mathbb{Z}_d^{(k)}$ theory, the anyons are $\{a^j\}_{j=0}^{d-1}$ with topological spin
\begin{align}
\theta(a^j) = \exp\!\left(\frac{2\pi i kj^2}{d}\right).
\end{align}
Thus $\theta(b) = \theta(a)^*$ implies
\begin{align}
b = a^q, \qquad q^2 \equiv -1 \pmod d.
\end{align}
Such a $q$ exists only in the LU-non-chiral cases discussed earlier ($d=p^m$ with $p=1$ mod $4$).

Next, we note that 
since $\tilde{\gamma}$ creates an anyon $a$, $\tilde{\gamma}^q$ creates an anyon $a^q$ by fusion.
Also, $\tilde{\gamma}^*$ creates $a^q$. 
Hence, $\tilde{\gamma}^q$ and $\tilde{\gamma}^*$ creates the same anyon type at their endpoints (up to a local unitary transformation acting exclusively on the neighborhood of the endpoints). That is, up to unitaries at the endpoints, $\sigma$ satisfies 
% \ZL{up to local unitaries at the end points?}
\begin{align} \label{eq: sigma imaginarity}
    \tilde{\gamma}^q\sigma \left( \tilde{\gamma}^q \right)^\dagger =  \tilde{\gamma}^*\sigma \left( \tilde{\gamma}^* \right)^\dagger.
\end{align}

We now derive a contradiction by  
complex conjugating both sides of Eq.\eqref{eq: sigma imaginarity}. This gives us
\begin{align} \label{eq: imaginarity contradiction}
    \left(\tilde{\gamma}^q\right)^*\sigma \left( \tilde{\gamma}^q \right)^\mathrm{T} =  \left(\tilde{\gamma}^*\right)^*\sigma \left( \tilde{\gamma}^* \right)^\mathrm{T}.
\end{align}
Since $(\tilde{\gamma}^*)^* = \tilde{\gamma}$, $(\tilde{\gamma}^*)^*$ creates an anyon $a$ at its endpoint. 
On the other hand, we have
\begin{align}
(\tilde{\gamma}^q)^* = (\tilde{\gamma}^*)^q
\end{align}
suggesting that $(\tilde{\gamma}^q)^*$ creates anyon ${(a^{q})}^q$ at its endpoints, which equals $a^{-1}$ because $q^2 \equiv -1 \pmod d$.
Therefore, Eq.~\eqref{eq: imaginarity contradiction} cannot be true, since $a^{-1}\neq a$.
Hence the assumption that $\rho$ can be made real by local unitaries leads to a contradiction. Therefore $\rho$ is LU-imaginary.
\end{proof}

Finally, it is useful to note that any qubit stabilizer state, whether pure or mixed, is not LU-imaginary.
This follows from the fact that any qubit stabilizer state is local Clifford equivalent to a graph state, and graph states admit a representation with real amplitudes~\cite{Nest:2004khg}.
For mixed states, this can be seen by considering their canonical purification, which is itself a stabilizer state and hence can be brought to a real form by local Clifford transformations.
This observation further suggests that the 3F theory is not LU-imaginary, as it admits a qubit stabilizer mixed-state realization.

%========================
\section{Discussion and outlook}
\label{sec6:discussion}
%======================== 

% \subsection{Relations to other chirality characterizations}

In this work, we have assessed the many-body chirality of stabilizer mixed and pure states realizing the $\Z_d^{(k)}$ and $\Z_{p^{2m}}^{(1)}$ abelian anyon theories, respectively. We found that both the mixed states and the pure states are LO-chiral whenever the anyon theory fails to be time-reversal symmetric, meaning that there is an obstruction to mapping the states to their complex conjugate with a local quantum channel. We also considered the $n$-partite chirality of the states, as introduced in Ref.~\cite{vardhan2025chiralitymagicquantumcorrelations}. We found that the $\Z_d^{(k)}$ mixed states exhibit four-partite chirality but not three-partite chirality, when the theory is not time-reversal symmetric. This implies that at least four parties are needed to diagnose the chirality of mixed states, demonstrating a fundamental limitation of the modular commutator. Finally, we considered LU-imaginarity and established that the states associated to $\Z_d^{(1)}$ are LU-imaginary, i.e., there is no finite-depth unitary that transforms the state to a time-reversal symmetric state.

To provide context for our results, it is insightful to compare them with a more conventional signature of chirality: the chiral central charge $c_{-}$. In two-dimensional topological orders, $c_-$ measures the net chirality of gapless edge modes, which manifests as a quantized thermal Hall conductance. Moreover, a nonzero $c_{-}$ implies an anomalous edge theory that cannot be realized as a purely local $(1+1)$-dimensional system, and therefore forbids a fully gapped boundary. This is illustrated by $\mathbb{Z}_3^{(1)}$ and the 3F, both of which have $c_{-}\neq 0$ and admit no fully gapped boundary. Conversely, a fully gapped boundary requires $c_{-}=0$, as in the toric code and $\mathbb{Z}_9^{(1)}$. However, this condition is not sufficient: $\mathbb{Z}_5^{(1)}$ has $c_{-}=0$ but does not admit a fully gapped boundary.

LO-chirality, by contrast, depends on the full anyon data, namely the braiding and topological spin structure.
This distinction is evident in the $3$F theory, which has $c_{-}=4$ but is LO-non-chiral since its anyon theory is equivalent to its mirror image.
Conversely, $\mathbb{Z}_9^{(1)}$ has $c_{-}=0$ and admits a gapped boundary, yet remains LO-chiral.
Thus, even in the absence of anomaly or boundary obstruction, a topological phase may retain intrinsic chirality.
For four-partite chirality, we find that it coincides with LO-chirality for mixed-state realizations.
For pure-state realizations (e.g.\ $\mathbb{Z}_9^{(1)}$), it remains an open question whether the state is four-partite chiral. 

We summarize the comparison in the following table:

% Here, we compare LO-chirality with other characterizations of chirality, including the chiral central charge, boundary gappability, four-pte chirality, and LU-imaginarity.
% Table~\ref{tab:lochiral-comparison} summarizes these for several representative examples. 

\begin{table}[H]
\centering
\renewcommand{\arraystretch}{1.25}
\setlength{\tabcolsep}{6pt}
\resizebox{\columnwidth}{!}{%
\begin{tabular}{|c|c|c|c|c|c|c|}
\hline
 & \begin{tabular}{c}\textbf{Chiral central}\\ \textbf{charge (mod 8)}\end{tabular}
 & \begin{tabular}{c}\textbf{Ungappable}\\ \textbf{boundary}\end{tabular}
 & \begin{tabular}{c} \textbf{LO chiral} \end{tabular}
  & \begin{tabular}{c} \textbf{Four-partite}\\ \textbf{chiral}\end{tabular}
 & \begin{tabular}{c}\textbf{Non-real wave function}\\ \textbf{(in local basis)}\end{tabular}\\
\hline
Toric code        & 0  & \No & \No & \No  & \No  \\
\hline
3F        & 4  & \Yes & \No  & \No & \No  \\
\hline
$\mathbb{Z}_3^{(1)}$ & 2 &\Yes  & \Yes  & \Yes & \Yes  \\
\hline
$\mathbb{Z}_5^{(1)}$ & 0  & \Yes  & \No & \No  & \Yes  \\
\hline
$\mathbb{Z}_9^{(1)}$ &
0
& \No  & \Yes &  \begin{tabular}{c} \Yes \ (mixed) \\ ? \ (pure) \end{tabular} & \Yes  \\
\hline
\end{tabular}%
}
% \caption{Different notions of chirality across representative topological orders. }
\label{tab:lochiral-comparison}
\end{table}

We conclude by commenting on the implications of our results for the chirality of three-dimensional models hosting $\Z_d^{(k)}$ theories on their boundary and the connection between symmetry-protected chirality and LO-chirality.

% We begin with the chiral central charge $c_{-}$, which characterizes the response of a $(2+1)$-dimensional system to background gravitational fields and manifests as the thermal Hall response.
% A nonzero $c_{-}$ implies an anomalous edge theory that cannot be realized as a purely local $(1+1)$-dimensional system, and therefore forbids a fully gapped boundary. This is illustrated by $\mathbb{Z}_3^{(1)}$ and the $3$-fermion theory, both of which have $c_{-}\neq 0$ and admit no fully gapped boundary.
% Conversely, a fully gapped boundary requires $c_{-}=0$, as in the toric code and $\mathbb{Z}_9^{(1)}$. However, this condition is not sufficient: $\mathbb{Z}_5^{(1)}$ has $c_{-}=0$ but does not admit a fully gapped boundary.

% It is also natural to compare these criteria with the existence of a real wavefunction in a local basis. Such a representation implies invariance under complex conjugation and therefore LO non-chirality, although the converse need not hold.
% For example, the $3$F theory admits a real representative, while $\mathbb{Z}_5^{(1)}$ is LO-non-chiral yet admits no real representative and is LU-imaginary.

\subsection{Boundary realizations of $\mathbb{Z}_d^{(k)}$ anyons}

Our discussion has focused on two-dimensional mixed-state realizations of $\mathbb{Z}_d^{(k)}$ anyon theories.
A natural question is how LO-chirality manifests in higher-dimensional systems, particularly in three-dimensional bulk phases that realize these states at their boundaries.

Recall that the $\mathbb{Z}_d^{(k)}$ anyon theories arise naturally on the boundary of three-dimensional stabilizer models~\cite{walker201131tqftstopologicalinsulators,lee2025chiralcolorcode}. In particular, as noted in Sec.~\ref{sec2}, a thin chiral color code of Ref.~\cite{lee2025chiralcolorcode} is equivalent to a canonically purified $\mathbb{Z}_d^{(k)}$ mixed state.
% This suggests that our characterization of LO-chirality extends to such three-dimensional pure-state realizations when the system is defined on a manifold with boundary.
% In particular, the Walker–Wang construction provides an explicit mechanism for generating surface mixed states with the same anyon content as the $\mathbb{Z}_d^{(k)}$ theory.
% Similarly, the chiral color code realizes boundary states that coincide with the models analyzed in this work. 
The proof of LO-chirality thus extends straightforwardly to the three-dimensional pure state by focusing on boundary-localized string operators.

\begin{observation}[Heuristic]
A Walker-Wang-type three-dimensional realization of $\mathbb{Z}_d^{(1)}$ anyon theories with boundaries is LO-chiral if its anyon content is not mirror invariant. 
\end{observation}

The situation is more subtle when the system is defined on a closed manifold, i.e., without boundaries.
We again focus on the $\mathbb{Z}_d^{(1)}$ models with odd $d$.\footnote{When $d$ is even, the three-dimensional model contains transparent fermion or boson, depending on $d=0,2 \ \mathrm{mod}\ 4$, and thus includes a decoupled three-dimensional fermionic or bosonic toric code. After removing this sector, an analogous argument applies.}
On a closed manifold, the three-dimensional $\mathbb{Z}_d^{(1)}$ model is short-range entangled for odd $d$, in the sense that its ground state can be prepared by a quantum cellular automaton (QCA)~\cite{Haah_2021}, or more generally by a local quantum channel~\cite{lee2025chiralcolorcode}.
Such a preparation effectively creates a pair of $\mathbb{Z}_d^{(1)}$ and $\mathbb{Z}_d^{(-1)}$ three-dimensional states across the system and environment.
Therefore, these systems are \emph{LO-non-chiral}.

\begin{observation}[Heuristic]
A Walker-Wang-type three-dimensional realization of $\mathbb{Z}_d^{(1)}$ (odd $d$) anyon theories without boundaries is LO-non-chiral. 
\end{observation}

However, it is an open question whether the states on a closed manifold remain \emph{LU-chiral}, in the sense that no finite-depth local unitary can map the three-dimensional $\mathbb{Z}_p^{(1)}$ state to its complex conjugate for $p=3$ \ mod $4$ (where the anyon content is not mirror invariant).
Below, we sketch a heuristic argument supporting the statement and point out its limitation.

Suppose, for contradiction, that there exists a finite-depth local unitary $U$ that maps the three-dimensional $\mathbb{Z}_p^{(1)}$ pure state to the $\mathbb{Z}_p^{(-1)}$ pure state on a closed manifold.
Now consider two copies of the $\mathbb{Z}_p^{(1)}$ model defined on a manifold with boundaries. By applying a truncated version of $U$ in the bulk, one can locally convert a region of one copy into the $\mathbb{Z}_p^{(-1)}$ state, thereby creating adjacent regions of $\mathbb{Z}_p^{(1)}$ and $\mathbb{Z}_p^{(-1)}$ states. 
% \TE{The unitary $U$ is a map of the states. For the modified argument below, we would need the stronger statement that the unitary maps the Hamiltonians, right?}

Since these two states for $\mathbb{Z}_p^{(1)}$ and $\mathbb{Z}_p^{(-1)}$ are inverses, we expect that they can be locally disentangled by a finite-depth circuit in the bulk. If the circuit can be extended to disentangle the respective Hamiltonians, that is, map them into a paramagnet Hamiltonian, then the disentangler would leave behind a two-dimensional boundary described by the product theory $\mathbb{Z}_p^{(1)} \times \mathbb{Z}_p^{(1)}$.
% (Here, recall that there exists a finite-depth unitary circuit preparing a pair of $\mathbb{Z}_p^{(1)}$ and $\mathbb{Z}_p^{(-1)}$.)
This would imply the existence of a two-dimensional frustration-free commuting-projector Hamiltonian realizing the $\mathbb{Z}_p^{(1)} \times \mathbb{Z}_p^{(1)}$ topological order.
% \TE{The circuit that disentangles the pair of $\mathbb{Z}_p^{(1)}$ and $\mathbb{Z}_p^{(-1)}$ states maps them to a product state, but it is not obvious that it maps the Hamiltonian to a trivial Hamiltonian. We need the latter to be able to leverage an obstruction to having a commuting projector Hamiltonian for a chiral phase. We can likely make a slightly weaker claim -- that there is no finite-depth Clifford circuit $U$.}  \TE{In the case of a Clifford circuit, we can strengthen the statement below. I am assuming that the circuit that prepares the $\mathbb{Z}_p^{(1)}$ and $\mathbb{Z}_p^{(-1)}$ states is translation invariant. Then we have a translation invariant topological stabilizer code on qudits. These are classified, at least for prime qudits. I would need to double check for odd in general.} \DL{I think we can disentangle $\mathbb{Z}_p^{(1)} \times \mathbb{Z}_p^{(-1)}$ (odd $p$) Hamiltonian by finite-depth circuit. Since we know that we can disentangle the bulk Hamiltonian of $\mathbb{Z}_p^{(1)}$ WW model using (possibly non-trivial) QCA, the bulk of $\mathbb{Z}_p^{(1)} \times \mathbb{Z}_p^{(-1)}$ can be disentangled by the corresponding QCA with its inverse, which can be realized by a finite-depth circuit?}
However, such a realization is believed to exist only when the corresponding anyon theory admits a Lagrangian subalgebra, which $\mathbb{Z}_p^{(1)} \times \mathbb{Z}_p^{(1)}$ does not. Therefore the LU-chirality distills down to a question of whether there exists a LU that disentangles not only the ground state but the full spectrum of the Hamiltonian for a stack of $\mathbb{Z}_p^{(1)}$ and $\mathbb{Z}_p^{(-1)}$. We leave the status of this problem as the following conjecture.

% We therefore conclude that the three-dimensional $\mathbb{Z}_p^{(1)}$ model is LU-chiral on a closed manifold whenever its anyon content is not mirror invariant, even though it remains not LO-chiral.

\begin{conjecture}[Heuristic]
A Walker-Wang-type three-dimensional realization of $\mathbb{Z}_p^{(1)}$ anyon theories without boundaries is LU-chiral when the anyon content is not mirror invariant.
\end{conjecture}

While the above discussion was phrased for $\mathbb{Z}_p^{(1)}$ with $p=3 \ \mathrm{mod}\ 4$, the argument extends more generally to $\mathbb{Z}_d^{(1)}$ anyon theories whose anyon content is not mirror invariant, provided that the product theory $\mathbb{Z}_d^{(1)} \times \mathbb{Z}_d^{(1)}$ does not admit a Lagrangian subalgebra.
An important class of exceptions is given by $\mathbb{Z}_{p^{2m}}^{(1)}$ models, which themselves admit a Lagrangian subgroup.
For these cases, the above obstruction does not apply. Indeed, it has been shown that the bulk Hamiltonian of such models can be disentangled by a finite-depth circuit with the assistance of ancillas~\cite{bauer2023disentangling}.\footnote{We note that this approach to using ancillas is distinct from our LO channel definition, as the ancillary state must remain invariant under the action of the finite-depth circuit.} By coarse-graning the bulk Hamiltonian, we expect that such models can be disentangled by a finite-depth circuit without the need for ancillas. Consequently, we expect that the bulk state is LU non-chiral. 
%and it remains an open question whether the corresponding three-dimensional realizations without boundaries are LU-chiral. \TE{In that case, isn't there a finite-depth disentangler for the bulk, so they can be mapped to the complex conjugate on a closed manifold?} \DL{I was thinking if we can cite Bauer's work to claim that this case is not LU chiral. Since his construction used ancilla, strictly speaking, it does not match with our LU definition?} %\DL{It is a bit surprising but it seems that for non-prime qudit dimension, Witt group classification of anyon theory and Clifford QCA classification does not match. The recent result of arXiv:2504.14811 points out that 3D QCA classification for $d=9$ qudit dimension is given by $\mathbb{Z}_4$, where the $\mathbb{Z}_9^{(1)}$ WW bulk disentangler corresponds to the root class. Apparently, $\mathbb{Z}_9^{(1)}$ anyon theory contains a lagrangian subgroup, but it seems that there is another algebraic obstruction to disentangle the bulk strictly using Clifford. I think non-clifford case remains to be open.   }

\subsection{LO vs.\ LU: symmetry and gauging}
\label{subsec:LOvsLu}
%\YP{Below is an alternative version, previous version is commented out}

For the family of two-dimensional stabilizer mixed states with abelian anyons
considered in this work, the characterization of chirality is identical whether
one considers LU circuits or more general LOs. This appears to
contrast with the chirality in invertible states, such as the $p+ip$ topological superconductor in two dimensions, where LU and LO notions can differ.

For invertible phases, it is more natural to use quasi-local versions of these
operations. We will write qLU for a quasi-local unitary, and qLO for the
corresponding channel obtained by adding local ancillas, applying a qLU, and
tracing out part of the enlarged system.

% \DL{Claim 1. $p + ip$ is qLU chiral} 
Consider the $p + ip$ topological superconductor in a two dimensional fermionic system. We argue that this phase is qLU-chiral based on a symmetry-gauging argument. Here, by a qLU for a fermionic system, we take each of the local gates to independently preserve fermion parity. Suppose there exists a qLU circuit that maps $p+ip$ superconductor to its complex conjugate, the $p-ip$ superconductor. Upon gauging fermion parity symmetry, $p+ip$ superconductor maps to Ising topological order, while $p-ip$ superconductor maps to its complex conjugate~\cite{Kitaev_2006}. Crucially, the Ising topological order realizes an anyon theory that is not invariant under complex conjugation. The existence of such qLU symmetric circuit would then imply that two distinct topological orders are connected by qLU circuit. We expect this to be impossible, although a rigorous proof in terms of anyons remains open due to the non-Abelian nature of the underlying anyon theory.\footnote{One could alternatively appeal to the thermal Hall conductivity and use the results of Ref.~\cite{Kapustin_2020}.}
The situation changes drastically if one allows for qLO channels. Namely, a $p+ip$ and
a $p-ip$ superconductor together have vanishing net chirality and can be generated from a product state via a qLU circuit~\cite{fidkowski2024pumpingchiralitydimensions}. Consequently, one can convert a $p+ip$ state into a $p-ip$ state by starting with an initial $p+ip$ state, introducing trivial ancilla, applying a qLU circuit to create a $p+ip$ and $p-ip$ state pair, swapping the newly generated $p-ip$ state with the original $p+ip$ state, and then tracing out the ancillary degrees of freedom.
%state. Thus the inverse chirality has not been removed, but has been placed in the environment. This is the sense in which the $p+ip$ phase is protected-unitary chiral, but not intrinsically qLO-chiral under unconstrained local channels.

% \DL{Claim 3. $p+ip$ is qLO chiral with strong symmetry of the system.}
This distinction suggests that, for non-trivial invertible phases, the allowed  qLO channel should respect the symmetry of the system to reproduce the result of the qLU classification. Namely, if the quasi-local unitary dilation respects the fermion parity of the system, i.e., the unitary is strongly symmetric, one can gauge the symmetry to demonstrate that qLO chirality of the Ising topological order implies the qLO chirality of $p+ip$ superconductor. We formalize this as a conjecture below.

\begin{comment}
gaugeable sense. With this restriction, one can no longer simply create the pair and throw away
the opposite copy. The copy that is traced out carries the chirality that
cancels the original one, and a gaugeable symmetric channel should keep track
of this sector.
\end{comment}

\begin{comment}
A useful way to make this distinction sharper is to gauge fermion parity.
Gauging a $p+ip$ superconductor gives Ising topological order, while gauging
$p-ip$ gives its complex conjugate. For example, the $\sigma$ anyon has
topological spin $\theta_\sigma=e^{i\pi/8}$ in one theory and
$\theta_\sigma^*=e^{-i\pi/8}$ in the other. Thus the chirality that could be
hidden in an environment before gauging reappears, after gauging, as part of
the anyon data. This example is outside the abelian stabilizer setting of this
paper, but it suggests the following conjectural extension of our results.
\end{comment}

\begin{conjecture}[Symmetry-protected qLO chirality of $p+ip$]
A $p+ip$ superconductor cannot be mapped to $p-ip$ by a local quantum
channel whose quasi-local unitary dilation respects the fermion parity symmetry of the system. Equivalently, after gauging fermion
parity, Ising topological order is not qLO-equivalent to its complex conjugate.
\end{conjecture}

\begin{comment}
The gauged $p+ip$ theory lies outside the theorem proved here, since it is
non-abelian rather than an abelian stabilizer realization of $Z_d^{(k)}$
anyons. It nevertheless illustrates what changes under gauging. In the ungauged
system, a local channel may compensate a chiral phase by leaving its inverse in
degrees of freedom that are eventually discarded. In the gauged theory, the
distinction between the phase and its complex conjugate is part of the
topological order itself, appearing for example in the topological spins.
\end{comment}

% \DL{Discussion on SPT phase with chirality}
An analogous argument also applies to a stack of two $p+ip$ states, which belongs to the same phase as the integer quantum Hall state in the absence of symmetries. Gauging the fermion parity symmetry results in the $U(1)_4$ abelian topological order. The $U(1)_4$ anyon theory is not mirror invariant, and thus we expect that it is qLU-chiral. For this reason, we further expect that the stack of two $p+ip$ states is symmetry-protected qLO-chiral.

More generally, assessing the qLO-chirality of a stack of $p+ip$ states requires caution. Gauging the fermion parity of eight copies of $p+ip$, for example, yields a state characterized by the 3F anyon theory, which is mirror invariant. Similarly, sixteen copies behave, after gauging, like the toric code stacked with an $E_8$ invertible state, so the anyon theory alone
does not determine the chiral edge content.
% These examples suggest several natural questions beyond the stabilizer setting studied here.
% % For continuum theories, it would be interesting to formulate a
% % gaugeable qLO notion that treats the Chern--Simons response on the same footing
% % as the anyon data in a lattice topological order.
For a stack of an even number of $p+ip$ states, 
gauging the fermion parity results in a state characterized by an abelian anyon theory. It would be interesting to understand whether these states can be related by a qLO to the zero-correlation-length mixed state studied in this work.
% , and whether qLO-chirality can
% be established for the $E_8$ state. More generally, one may hope that the right LO notion for symmetry-protected phases is a symmetry-respecting qLO whose action remains local after the symmetry is gauged. 
This would help to establish the connection between the symmetry-protected qLO-chirality of the ungauged state and the intrinsic chirality of the gauged topological order.

\subsection*{Acknowledgment}

We thank 
Yu-An Chen,
Tarun Grover,
Tim Hsieh,
Michael Levin,
Bowen Shi,
Shreya Vardhan,
Chong Wang, 
and
Yijian Zou
for useful discussions.
Research at Perimeter Institute is supported in part by the Government of Canada through the Department of Innovation, Science and Economic Development and by the Province of Ontario through the Ministry of Colleges and Universities. 
This work is supported by the Applied Quantum Computing Challenge Program at the National Research Council of Canada and by the Natural Sciences and Engineering Research Council of Canada (NSERC) through Discovery Grants.

\appendix

\section{$\mathbb{Z}_d^{(k)}$ anyons}

We recall that the braiding and self statistics in the $\mathbb{Z}_d^{(k)}$ anyon theory are given by (Eq.~\eqref{def:anyontheory}):
\begin{align}
B(a^{j_1}, a^{j_2}) = (\omega)^{2k j_1 j_2}, \qquad \theta(a^j) = (\omega)^{k j^2},\qquad \omega=e^{\frac{2\pi i}{d}}.
\end{align}

% \ZL{I moved the example subsection to after the Chiral central charge subsection, because we mention $G_d^{(k)}$ extensively in the examples, but this notation and its physical meaning is first discussed in the chiral central charge subsection. 
% Alternatively, we can keep the original order but we need to either define $G_d^{(k)}$ first (or don't mention them).}

\subsection{Proof of Lemma~\ref{lemma:factorization}}\label{A6}

\noindent\textbf{Lemma \ref{lemma:factorization}}~(Factorization).
\textit{
The $\mathbb{Z}_d^{(k)}$ anyon theory is equivalent to the following tensor product:
\begin{align}
\mathcal{A}_d^{(k)} \simeq \boxtimes_{j} \mathcal{A}_{d_j}^{(k e_j)}
\end{align}
where $d_j \equiv p_j^{m_j}$ and $e_j \equiv \frac{d}{d_j}$.}

Let us begin with the $k=1$ case. 
The anyon set $\mathcal{A}_d^{(1)}= \{ a^i \}_{i=0}^{d-1}$ can be decomposed as
\begin{align}
\mathcal{A}_d^{(1)}= \bigtimes_{j=1}^{n} \mathcal{A}(j), \qquad \mathcal{A}(j) \equiv \big\{ (a^{e_j})^i  \big\}_{i=0}^{d_{j}-1}. \label{eq:decompose}
\end{align}
Indeed, the following equation always admits a solution $\{x_j\}$  due to the CRT:
\begin{align}
1 = \sum_{j} x_j e_j, \qquad (\mbox{mod $d$}). \label{eq:CRT}
\end{align}
Hence, the RHS of Eq.~\eqref{eq:decompose} contains the elementary anyon $a$, namely $a = a^{\sum_j x_j e_j}$, which generates the entire set $\mathcal{A}_d^{(1)}$.
Furthermore, different $\mathcal{A}(j)$ are decoupled in $\mathcal{A}_d^{(1)}$, since (observe that $\mathcal{A}(j)$ is generated by $a^{e_j}$, it suffices to check $a^{e_j}$ and $a^{e_{j'}}$):
\begin{equation}
    B(a^{e_j}, a^{e_{j'}}) = \omega^{2e_je_{j'}} = 1~~~~(j\neq j').
\end{equation}

We claim that
\begin{align}\label{eq:a2theta}
\mathcal{A}(j) \simeq \mathcal{A}_{d_j}^{(e_j)}.
\end{align} 
To show it, we compute the self statistics of $a^{e_j}$ (recall Eq.~(\ref{def:anyontheory})): 
\begin{align}
\theta(a^{e_j}) = (\omega_d)^{e_j^2} = (\omega_{d_j})^{e_j}
\end{align}
%\ZL{I think it still remains to verify that different $\mathcal{A}(j)$ decouple and together they make up $\mathcal{A}_d^{(1)}$} \DL{Sorry, I am not sure if I understood this comment. Eq.(222) implies ${\cal A}_d^{(1)} \subseteq \otimes_{j=1}^n {\cal A}(j)$, and the converse holds trivially. So Eq.(220) holds? }
%\ZL{In principle, the CRT only proves (221) as sets, not as anyon theories. In particular, we need to validate the ``$\bigtimes_{j=1}^{n}$". This is indeed very very trivial, but to make a clear proof, I just added one line (223) anyways. Feel free to delete our comments if you confirm it.}
which matches with the self statistics of an elementary anyon in the $\mathbb{Z}_{d_j}^{(e_j)}$ anyon theory. 
For abelian anyon theories, the braiding statistics is fully determined by the exchange statistics.
\footnote{For anyons $a$ and $b$, we have $B(a,b) = \theta(ab) \theta(a)^{-1}\theta(b)^{-1}$.~\cite{Kitaev_2006}}
Hence, the set of anyons in $\mathcal{A}(j)$ is isomorphic to the set $\mathcal{A}_{d_j}^{(e_j)}$. 

For $k>1$, one can similarly show that ${\cal A}(j) \simeq {\cal A}^{(ke_j)}_{d_j}$, and this concludes the proof.
%by multiplying Eq.~\eqref{eq:CRT} 
%\ZL{why multiply this by $k$? shouldn't we play with \eqref{eq:a2theta} instead?}, 
% we arrive at the desired statement.

\subsection{Chiral central charge of $\mathbb{Z}_d^{(1)}$ anyon theory}
\label{A7}

In this appendix, we examine the problem of the mirror invariance, $\mathcal{A}\simeq \mathcal{A}^*$ or  $\mathcal{A}\not\simeq \mathcal{A}^*$, from the perspective of the chiral central charge $c_{-}\equiv c_L - c_R$. 
Since the chiral central charge $c_{-}$ is additive under tensor product $\boxtimes$, it makes sense to focus on indecomposable anyon theories.
For the $\mathbb{Z}_d^{(1)}$ anyon theories, due to the factorization property (Lemma~\ref{lemma:factorization}), we can focus on the $d=p^m$ cases for some prime $p$. 

The normalized sum of the self statistics of anyons gives a probe of the chiral central charge $c_{-}$. 
Namely, when the anyon theory is abelian and modular (not containing any transparent anyon), the \emph{Gauss-Milgram formula} gives:
\begin{align}
e^{\frac{2\pi i}{8}c_{-}} = \frac{1}{\sqrt{|\mathcal{A}|}}\sum_{a \in \mathcal{A}}\theta(a) \label{eq:Gauss-Milgram}
\end{align}
which fixes (defines) the chiral central charge $c_{-}$ modulo $8$. 
The formula suggests a useful criterion for verifying $\mathcal{A}\simeq \mathcal{A}^*$ or $\mathcal{A}\not\simeq \mathcal{A}^*$. Namely, if $\mathcal{A}\simeq \mathcal{A}^*$, then Eq.~\eqref{eq:Gauss-Milgram} must be invariant under complex conjugation, implying
\begin{align}
\mathcal{A}\simeq \mathcal{A}^* \quad \Rightarrow \quad c_{-}=0 \quad (\mbox{mod $4$}). \label{eq:Right-only}
\end{align}
Below, we will demonstrate that the converse does \emph{not} hold in general for the $\mathbb{Z}_d^{(1)}$ anyon theories.

\subsubsection{Odd $d$}

For odd $d$, the $\mathbb{Z}_d^{(1)}$ anyon theory does not have transparent anyon. 
The normalized sum of the self statistics is given by the quadratic Gauss sum:
\begin{align}
e^{\frac{2\pi i}{8}c_{-}} = \frac{1}{\sqrt{d}} \sum_{j=1}^d \theta(a^j) = G_{d}^{(k)} , 
\qquad G_{d}^{(k)} \equiv \frac{1}{\sqrt{d}} \sum_{j=1}^d \omega^{k j^2}.
\label{Eq:gauss sum}
\end{align}
For $k=1$, the following relation is well known:
\begin{align}
G_{d}^{(1)} =
\begin{cases} 
& 1   \qquad\quad d =1 \quad (\mathrm{mod\ } 4) \\
  & i  \qquad\quad d =3 \quad (\mathrm{mod\ } 4)
  \end{cases}
\end{align}
suggesting that the chiral central charges are given by
\begin{align}
c_{-} = 
\begin{cases} 
& 0 \quad (\mathrm{mod\ } 8)  \qquad\quad d =1 \quad (\mathrm{mod\ } 4) \\
  & 2 \quad (\mathrm{mod\ } 8)  \qquad\quad d =3 \quad (\mathrm{mod\ } 4).
  \end{cases}
\end{align}
This calculation can be generalized to $G_{d}^{(k)}$ with $k>1$ and coprime $(d,k)$, with $c_{-}=\pm 2$ modulo $8$ 
% depending on whether $d$ possesses a quadratic residue or not. 
depending on the prime factors of $d$ and whether they have $k$ as a quadradic residue.
% \ZL{I think this is still a bit confusing, since $d$ always possesses a quadratic residue (e.g., 1).\\
% A precise result would be: 
% \begin{align}
% G_{d}^{(1)} =
% \begin{cases} 
% & \left(\frac{k}{d}\right)   \qquad\quad d =1 \quad (\mathrm{mod\ } 4) \\
%   & i\left(\frac{k}{d}\right)  \qquad\quad d =3 \quad (\mathrm{mod\ } 4)
%   \end{cases},
% \end{align}
% \begin{align}
%     c_{-} = \begin{cases} 
% & 0, 4 \quad (\mathrm{mod\ } 8)  \qquad\quad d =1 \quad (\mathrm{mod\ } 4) \\
%   & 2, -2 \quad (\mathrm{mod\ } 8)  \qquad\quad d =3 \quad (\mathrm{mod\ } 4).
%   \end{cases}
% \end{align}
% Here $\left(\frac{k}{d}\right)$ is the Jacobi symbol, valued in $\pm 1$, depending on the prime factors of $d$ and whether they have $k$ as a quadradic residue.
% }

%$k$ is a quadradic residue modulo $d$ or not.
% \ZL{I think this is not true... According to wiki page on Quadratic\_Gauss\_sum, the sign depend on the Jacobi symbol. Quadradic residue implies Jacobi symbol=1, but not vice versa (for $d=$prime, yes).}

{\color{blue}

}

Here, we notice that the $d=3$ (mod $4$) cases always lead to $\mathcal{A}\not\simeq \mathcal{A}^*$. 
It turns out, however, that the $d=1$ (mod $4$) cases do not necessarily yield $\mathcal{A}\simeq \mathcal{A}^*$. 
Namely, letting $p$ be an odd prime with $p=3$ (mod $4$), consider the $d = p^{2m}$ cases. 
Then, the chiral central charge is $c_{-}=0$ mod $8$, but $\mathcal{A}_{p^{2m}}^{(1)}\not\simeq \mathcal{A}_{p^{2m}}^{(-1)}$, due to Theorem \ref{thm:invcondition}.
As such, the chiral central charge $c_{-}$ modulo $4$ fails to faithfully characterize the mirror invariance of the $\mathbb{Z}_d^{(1)}$ anyon theory, negating the converse of Eq.~\eqref{eq:Right-only}.

\subsubsection{Even $d$}

For even $d$, the $\mathbb{Z}_d^{(1)}$ anyon theory contains a transparent boson or fermion. Namely, the transparent particle is given by $t \equiv a^{\frac{d}{2}}$, and we find
\begin{align}
\theta(t) =
\begin{cases} & 1     \ \ \qquad d =0 \quad (\mathrm{mod\ } 4) \\
& -1     \qquad d =2 \quad (\mathrm{mod\ } 4)
\end{cases},
\end{align}
The sum of self statistics (the Gauss sum), is given by 
\begin{align}
G_{d}^{(1)} = 
\begin{cases}
&1 + i \ \quad d =0  \quad (\mathrm{mod\ } 4) \\
 & 0  \qquad\quad d =2 \quad (\mathrm{mod\ } 4) 
 \end{cases},
\end{align}
but we cannot use the Gauss-Milgram formula to determine the chiral central charge. 

This issue can be circumvented by condensing the transparent boson or fermion and forming a reduced anyon theory $\hat{\mathbb{Z}}_{d}^{(1)}$ with a reduced anyon set $\hat{\mathcal{A}}_{d}^{(1)} = \{ a^j \}_{j=0}^{j=\frac{d}{2}-1}$. Focusing on the $d=2^m$ cases for simplicity, the normalized self statistics gives: 
\begin{align}
\hat{G}_{d}^{(1)} \equiv
\frac{1}{\sqrt{d/2}} \sum_{j=0}^{\frac{d}{2}-1} \omega^{j^2}=
\begin{cases}
&\frac{1}{\sqrt{2}}(1 + i) \ \quad d = 2^m \quad (m\geq 2) \\
 & 1  \qquad\qquad\quad d =2 
 \end{cases},
\end{align}
suggesting that the chiral central charges are given by
\begin{align}
c_{-} = 
\begin{cases} 
& 1 \quad (\mathrm{mod\ } 8)   \quad \quad d = 2^m \quad (m\geq 2) \\
  & 0 \quad (\mathrm{mod\ } 8)  \qquad d =2 .
  \end{cases}.
\end{align}
Hence, we find that, for $d=2^m$, the chiral central charge $c_{-}$ modulo $4$ faithfully characterize the mirror invariance of the $\mathbb{Z}_d^{(1)}$ anyon theory.

\subsection{Examples}\label{app:A1}

\subsubsection*{-- $\mathbb{Z}_3^{(1)}$ anyon}

The self statistics for $\mathbb{Z}_3^{(1)}$ and $\mathbb{Z}_3^{(-1)}$ anyon theories are given by
\begin{center}
\begin{tabular}{ c | c c }
 $j$ & $\log(\theta)$ & $\log(\theta^*)$\\ \hline
 $0$ & $0$ & $0$ \\  
 $1$ & $1$ & $2$ \\    
 $2$ & $1$ & $2$ 
\end{tabular}
\end{center}
Here, $j$ represents the exponent in $a^j$ while ``$\log(\theta)$'' is a shorthand notation to denote the exponent in the topological spin, namely
\begin{align}
\theta = \omega^{\log(\theta)}. 
\end{align}
The self statistics do not match for $\mathbb{Z}_3^{(1)}$ and $\mathbb{Z}_3^{(-1)}$, suggesting that $\mathbb{Z}_3^{(1)}$ is not mirror invariant.
We also have 
\begin{align}
G_{3}^{(1)} = i, \qquad G_{3}^{(-1)} = -i. 
\end{align}

\subsubsection*{-- $\mathbb{Z}_5^{(1)}$ anyon}

The self statistics for $\mathbb{Z}_5^{(1)}$ anyon theory is given by
\begin{center}
\begin{tabular}{ c | c c }
 $j$ & $\log(\theta)$ & $\log(\theta^*)$\\ \hline
 $0$ & $0$ & $0$ \\  
 $1$ & $1$ & $4$ \\    
 $2$ & $4$ & $1$ \\
 $3$ & $4$ & $1$ \\
 $4$ & $1$ & $4$ 
\end{tabular}
\end{center}
The self statistics are equivalent up to the following relabeling
\begin{align}
a^* \sim a^2
\end{align}
suggesting that $\mathbb{Z}_5^{(1)}$ is mirror invariant.
We also have 
\begin{align}
G_{5}^{(1)} = 1, \qquad G_{5}^{(-1)} =1. 
\end{align}

\subsubsection*{-- $\mathbb{Z}_9^{(1)}$ anyon}

The exchange statistics for $\mathbb{Z}_9^{(1)}$ anyon theory is given by
\begin{center}
\begin{tabular}{ c | c c }
 $j$ & $\log(\theta)$ & $\log(\theta^*)$\\ \hline
 $0$ & $0$ & $0$ \\  
 $1$ & $1$ & $8$ \\    
 $2$ & $4$ & $5$ \\
 $3$ & $0$ & $0$ \\
 $4$ & $7$ & $2$ \\  
 $5$ & $7$ & $2$ \\    
 $6$ & $0$ & $0$ \\
 $7$ & $4$ & $5$ \\
 $8$ & $1$ & $8$ 
\end{tabular}
\end{center}
suggesting that $\mathbb{Z}_9^{(1)}$ is not mirror invariant.
We however have 
\begin{align}
G_{9}^{(1)} = 1, \qquad G_{9}^{(-1)} =1. 
\end{align}
The $\mathbb{Z}_9^{(1)}$ anyon theory contains a non-transparent boson $b = a^3$.

\subsubsection*{-- $\mathbb{Z}_2^{(1)}$ anyon}

The exchange statistics for $\mathbb{Z}_2^{(1)}$ anyon theory is given by
\begin{center}
\begin{tabular}{ c | c c }
 $j$ & $\log (\theta)$ & $\log (\theta^*)$\\ \hline
 $0$ & $0$ & $0$ \\  
 $1$ & $1$ & $1$ 
\end{tabular}
\end{center}
suggesting that $\mathbb{Z}_2^{(1)}$ is mirror invariant.
We have %\ZL{0?}
\begin{align}
G_{2}^{(1)} = 0 .
%, \qquad G_{2}^{(-1)} =0. 
\end{align}
The $\mathbb{Z}_2^{(1)}$ anyon theory consists only of a fermion. 

\subsubsection*{-- $\mathbb{Z}_4^{(1)}$ anyon}

The exchange statistics for $\mathbb{Z}_4^{(1)}$ anyon theory is given by
\begin{center}
\begin{tabular}{ c | c c }
 $j$ & $\log(\theta)$ & $\log(\theta^*)$\\ \hline
 $0$ & $0$ & $0$ \\  
 $1$ & $1$ & $3$ \\
  $2$ & $0$ & $0$ \\
 $3$ & $1$ & $3$ 
\end{tabular}
\end{center}
suggesting that $\mathbb{Z}_4^{(1)}$ is not mirror invariant. After removing a transparent boson $t = a^2$, we have a reduced $\hat{\mathbb{Z}}_4^{(1)}$ anyon theory
\begin{center}
\begin{tabular}{ c | c c }
 $j$ & $\omega(\theta)$ & $\omega(\theta^*)$\\ \hline
 $0$ & $0$ & $0$ \\  
 $1$ & $1$ & $3$ 
\end{tabular}
\end{center}
which lead to 
\begin{align}
\hat{G}_{2}^{(1)} = \frac{1}{\sqrt{2}}(1+i), \qquad \hat{G}_{2}^{(-1)} =\frac{1}{\sqrt{2}}(1-i). 
\end{align}
The $\hat{\mathbb{Z}}_4^{(1)}$ anyon theory consists of a chiral semion.

\section{Stacking $\mathbb{Z}_d^{(k)}$ anyons}
\label{app:stacking-Zdk}

We have established the necessary and sufficient condition for when the $\mathbb{Z}_d^{(k)}$ anyon theory is mirror invariant.
A single copy of $\mathbb{Z}_d^{(k)}$, however, does not capture the full landscape of modular abelian anyon theories. 
Fortunately, the Witt class classification implies that, by taking multiple copies of $\mathbb{Z}_d^{(k)}$ theories and allowing for anyon condensation, one can realize every modular abelian anyon theory \cite{Davydov2010TheWG, Davydov2013}.
This observation naturally motivates us to investigate the mirror invariance under stacking of $\mathbb{Z}_d^{(k)}$ anyon theories. 

\subsection{Stacking two copies}

We begin with the odd $d$ cases. 
We find that two copies of $\mathbb{Z}_d^{(1)}$ anyon theories are always mirror invariant.

\begin{theorem}[Stacking odd $d$]
The $\mathbb{Z}_d^{(1)} \boxtimes \mathbb{Z}_d^{(1)}$ anyon theory is mirror invariant for odd $d$. 
\end{theorem}

This is consistent with the fact that a single copy of the theory has $c_{-}= \pm 2$ modulo $8$ when $d=3$ modulo 4, and stacking two copies yield $c_{-}=4$ modulo $8$.

The proof follows from an elementary number theory argument.

\begin{proof}
It suffices to consider the cases where $d = p^m$ with $p=3$ mod $4$. 
Let $a_1, a_2$ be elementary anyons for $\mathbb{Z}_d^{(1)} \boxtimes \mathbb{Z}_d^{(1)}$. 
Let $b_1, b_2$ be elementary anyons for $\mathbb{Z}_d^{(-1)} \boxtimes \mathbb{Z}_d^{(-1)}$. 
We show
\begin{align}
\mathcal{A}_d^{(1)} \times \mathcal{A}_d^{(1)} \simeq \mathcal{A}_d^{(-1)} \times \mathcal{A}_d^{(-1)}. 
\end{align}
For this to hold,  $b_1$ and $b_2$ need to be expressed in terms of $a_1$ and $a_2$:
\begin{align}
b_1 = a_1^{x_1} a_2^{x_2}, \qquad b_2 = a_1^{y_1} a_2^{y_2}.
\end{align}
This requires 
\begin{align}
x_1^2 + x_2^2 = -1, \qquad y_1^2 + y_2^2 = -1, \qquad 2(x_1y_1 + x_2 y_2) = 0.
\end{align}
Let us assume that $x_1^2 + x_2^2 = -1$ has a solution. We then choose $y_1=x_2$ and $y_2 = - x_1$, which satisfies all the above equations. Thus, the remaining task is to prove that the following equation has a solution:
\begin{align}
x_1^2 + x_2^2 = -1 \qquad \mathrm{mod \ } p^m.
\end{align}
Here, we provide a proof for $m=1$. Lifting it to $m>1$ is straightforward. Let $N$ be the total number of solutions to this equation. We then have
\begin{align}
N = \frac{1}{p} \sum_{t=0}^{p-1} \sum_{x=0}^{p-1} \sum_{y=0}^{p-1} \omega^{t(x^2 + y^2 +1)} 
=  \sum_{t=0}^{p-1}  \omega^{t} G_{p}^{(t)}G_{p}^{(t)},
\end{align}
For $p=3$ mod $4$, we have 
\begin{align}
G_{p}^{(t)} =
\begin{cases}& \pm i \qquad t\not=0 \\
            & \sqrt{p}  \qquad t=0.
            \end{cases}
\end{align}
Hence, we have 
\begin{align}
N =  p -  \sum_{t\not=0} \omega^{t}  
=  p+1.
\end{align}
This suggests that the equation indeed has solutions. 
\end{proof}

\subsection{Stacking four copies}

Next, we show that four copies of $\mathbb{Z}_{d}^{(1)}$ theories are mirror invariant when $d$ is even. 

\begin{theorem}[Stacking even $d$]
The $\mathbb{Z}_d^{(1)} \times \mathbb{Z}_d^{(1)} \times \mathbb{Z}_d^{(1)} \times \mathbb{Z}_d^{(1)}$ anyon theory is mirror invariant for $d=2^m$. 
\end{theorem}

This is consistent with the fact that a single copy of the theory has chiral central charge  $c_{-}=1$ modulo $8$, and stacking the four copies yields $c_{-}=4$ modulo $8$.

The proof follows from elementary number theory techniques, namely by Lagrange's four-square theorem and quaternions. 

\begin{proof}
Our goal is to show
\begin{align}
\mathcal{A}_d^{(1)} \times \mathcal{A}_d^{(1)} \times \mathcal{A}_d^{(1)} \times \mathcal{A}_d^{(1)}\simeq \mathcal{A}_d^{(-1)} \times \mathcal{A}_d^{(-1)} \times \mathcal{A}_d^{(-1)}\times \mathcal{A}_d^{(-1)}. 
\end{align}
Letting $a_1,a_2,a_3, a_4$ be elementary anyons for $\mathcal{A}_d^{(1)}$, and $b_1,b_2,b_3, b_4$ be elementary anyons for $\mathcal{A}_d^{(-1)}$, we set
\begin{align}
b_1 &= a_1^{x_1} a_2^{y_1} a_3^{z_1} a_4^{w_1} \\
b_2 &= a_1^{x_2} a_2^{y_2} a_3^{z_2} a_4^{w_2} \\
b_3 &= a_1^{x_3} a_2^{y_3} a_3^{z_3} a_4^{w_3} \\
b_4 &= a_1^{x_4} a_2^{y_4} a_3^{z_4} a_4^{w_4}.
\end{align}
To produce the correct topological spin, they need to satisfy 
\begin{align}
v_\ell\cdot v_\ell = x_\ell^2 + y_\ell^2 + z_\ell^2 + w_\ell^2 = -1 \qquad \mathrm{mod \ } 2^m,\label{eq:square}
\end{align}
where $\cdot$ represents the standard inner product and $v_\ell = (x_\ell,y_\ell,z_\ell,w_\ell) $.
Also, $b_\ell$ need to braid trivially with each other, which implies  $ v_\ell \cdot v_{\ell'} =0$ mod $2^{m-1}$.
Here, it is sufficient to impose the stronger condition
\begin{align}
 v_\ell \cdot v_{\ell'} =0\qquad (\ell\neq \ell').\label{eq:square_ortho}
\end{align}

Recall that due to Lagrange's four square theorem, for an arbitrary positive integer $n$, the following equation always has a solution
\begin{align}
x^2 + y^2 + z^2 + w^2 = n. 
\end{align}
Hence, by reducing this to modulo $2^m$, Eq.~\eqref{eq:square} always has a solution.
Let $v =(x,y,z,w)$ be such a solution. Here, it will be convenient to introduce a quaternion representation of the vector:
\begin{align}
q(v) = x + y i + z j + w k
\end{align}
where $1,i,j,k$ are the basis of quaternion:
\begin{align}
i^2 = j^2 = k^2 = -1, \qquad ijk = -1. 
\end{align}
In the vector representation, $1,i,j,k$ correspond to $(1,0,0,0),(0,1,0,0),(0,0,1,0),(0,0,0,1)$. 
Here we construct $v_j$ from the following constraint:
\begin{align}
q(v_1) = q(v) 1, \qquad q(v_2) = q(v) i, \quad q(v_3) = q(v) j, \quad q(v_4) = q(v) k.
\end{align}
That is, we compute $q(v_2) = q(v) i$ by using the quaternion rule, and then convert it back to a vector to construct $v_2$.

Recalling that the norm is multiplicative in the quaternion representation and $1,i,j,k$ have unit norm, we find $v_\ell\cdot v_\ell = v\cdot v = -1 ~(\text{mod~} 2^m)$.
This gives us vectors $v_\ell$ satisfying Eq.~\eqref{eq:square}. 
On the other hand,
\begin{align}
v_\ell\cdot v_{\ell'}
=\operatorname{Re}\left(q(v_\ell)\overline{q(v_{\ell'})}\right)
=\operatorname{Re}\left(q(v) u \overline{q(v)}\right),
\end{align}
where $u\in\{i,j,k\}$ since $\ell\neq \ell'$. 
Note that $\overline{q\,u\,\bar q}=q\,\bar u\,\bar q=-\,q\,u\,\bar q$, so $q\,u\,\bar q$
is purely imaginary and its real part vanishes.
Therefore, $v_\ell\cdot v_{\ell'}\equiv0\pmod{2^m}$, satisfying Eq.~\eqref{eq:square_ortho}.
\end{proof}

\subsection{Examples}

\subsubsection{$\mathbb{Z}_3^{(1)}\times \mathbb{Z}_3^{(1)}$}

The $\mathbb{Z}_3^{(1)}\times \mathbb{Z}_3^{(1)}$ anyon theory can be written as 
\begin{align}
\{ 1, a_1, a_1^2 \} \times \{ 1, a_2, a_2^2 \}. 
\end{align}
Defining
\begin{align}
b_1 \sim a_1 a_2, \qquad b_2 \sim a_1 a_2^{2},
\end{align}
the $\mathbb{Z}_3^{(-1)}\times \mathbb{Z}_3^{(-1)}$ anyon theory can be expressed as
\begin{align}
\{ 1, b_1, b_1^2 \} \times \{ 1, b_2, b_2^2 \}. 
\end{align}
Hence, two theories are isomorphic. 

\subsubsection{$\mathbb{Z}_4^{(1)}\times \mathbb{Z}_4^{(1)}\times \mathbb{Z}_4^{(1)}\times \mathbb{Z}_4^{(1)}$}

The $(\mathbb{Z}_4^{(1)})^{\times 4}$ anyon theory can be written as 
\begin{align}
\{ 1, a_1, a_1^2, a_1^3 \} \times \{ 1, a_2, a_2^2, a_2^3 \} \times
\{ 1, a_3, a_3^2, a_3^3 \} \times \{ 1, a_4, a_4^2, a_4^3 \} . 
\end{align}
Defining 
%\ZL{note that they are not the solutions given by the quaternions, and do not satisfy Eq.\ref{eq:square_ortho} mod $2^m$?}
\begin{align}
b_1 &= a_2 a_3 a_4 \\ 
b_2 &= a_1 a_3 a_4 \\ 
b_3 &= a_1 a_2 a_4 \\ 
b_4 &= a_1 a_2 a_3,
\end{align}
the $(\mathbb{Z}_4^{(-1)})^{\times 4}$ anyon theory can be written as 
\begin{align}
\{ 1, b_1, b_1^2, b_1^3 \} \times \{ 1, b_2, b_2^2, b_2^3 \} \times
\{ 1, b_3, b_3^2, b_3^3 \} \times \{ 1, b_4, b_4^2, b_4^3 \} . 
\end{align}
Indeed, one can check that $\theta(b_j) = \omega^{-1}$ and $b_j$s' commute with each other.\footnote{
This choice does not satisfy the stronger orthogonality condition in Eq.~\eqref{eq:square_ortho} modulo $4$, but it satisfies the weaker condition required for trivial mutual braiding. The choice satisfying the stronger condition, corresponding to the quaternions, is given by
$v_1 = (1,1,1,0)$, $v_2 = (3,1,0,3)$, $v_3 = (3,0,1,1)$,
$v_4 = (0,1,3,1)$.
}

\subsubsection{Three fermion}

% The $(\mathbb{Z}_4^{(1)})^{\times 4}$ anyon theory generates the so-called 3 fermion theory after condensing the following bosons:
The three-fermion theory is obtained from the $(\mathbb{Z}_4^{(1)})^{\times 4}$ anyon theory after condensing the following bosons:
\begin{align}
a_1^2 = a_2^2 = a_3^2 = a_4^2 \underset{\text{condense}}{=} 1 \cr 
b = a_1 a_2 a_3 a_4  \underset{\text{condense}}{=} 1
\end{align} 
%\ZL{(finally, a physics question...) we said above that $a_1^2$ is a transparent boson, do they survive after condensing $b$? } \DL{Actually, I think they survive and we further need to condense transparrent $a_i^2$ anyon to get 3F I guess? I modified the above equation.}
Namely, one obtains the condensed theory
\begin{align}
\{ 1, f_1, f_2, f_3 \}
\end{align}
where fermions are given by
\begin{align}
f_1 = a_1 a_4, \quad f_2 = a_2 a_4, \quad f_3 = a_3 a_4.
\end{align}
% That three fermion theory emerges in four copies of $(\mathbb{Z}_4^{(1)})$ suggests that it is mirror invariant. 
The fact that the three-fermion theory is obtained from condensation in four copies of $\mathbb{Z}_4^{(1)}$ suggests that it is mirror invariant, since four copies of $\mathbb{Z}_4^{(1)}$ is itself mirror invariant.
% \ZL{(not sure whether this question makes sense) should I think 3F as a ``subtheory" of $(\mathbb{Z}_4^{(1)})^{\times 4}$ or a ``quotient theory" of it? which one preserve the nonchiral property (I guess the quotient one)?}\YP{I think you're right, I made some changes here.}

%\mciteSetMidEndSepPunct{}{\ifmciteBstWouldAddEndPunct.\else\fi}{\relax}
\bibliographystyle{JHEP}
\bibliography{references.bib}
%\printbibliography

\end{document}